\documentclass[a4wide]{article}

\usepackage{a4wide}
\usepackage{float}
\usepackage[all]{xy}
\usepackage{amsfonts}
\usepackage{amsmath,amssymb,amsthm}
\usepackage{mathtools}
\usepackage{stmaryrd}
\usepackage{bussproofs}
\usepackage{xspace}
\usepackage{xcolor,colortbl}
\usepackage{xifthen}
\usepackage{hyperref}

\usepackage{graphicx}
\usepackage{pgf,pgfarrows,pgfnodes}
\usepackage{tikz}
\usetikzlibrary{arrows,automata}
\usepackage{diagbox,xspace}
\usepackage{algorithm,algorithmicx,algpseudocode}

\usepackage{enumerate}

\newtheorem{theorem}{Theorem}

\newtheorem{proposition}{Proposition}
\newtheorem{corollary}{Corollary}
\newtheorem{lemma}{Lemma}

\theoremstyle{definition}
\newtheorem{definition}{Definition}
\newtheorem{example}{Example}

\newcommand{\psf}{\mathcal{P}}

\newcommand{\brak}[1]{\langle\![ #1   ]\!\rangle}
\newcommand{\nexttime}{\mathsf{X}}
\newcommand{\until}{\mathsf{U}}
\newcommand{\always}{\mathsf{G}}
\newcommand{\tset}[1]{[\![#1]\!]}

\newcommand{\typeone}{\ensuremath{\mathsf{Type}\until}\xspace}
\newcommand{\typetwo}{\ensuremath{\mathsf{Type}\always}\xspace}
\newcommand{\chfor}{\ensuremath{\mathsf{char}}\xspace}

\newcommand{\indf}[2]{\mathsf{ind}(#1, #2)}
\newcommand{\unf}[1]{\mathsf{unfold}(#1)}

\DeclareMathOperator{\Agt}{\mathsf{Agt}}

\DeclareMathOperator{\Prop}{\mathsf{AP}} 
\DeclareMathOperator{\Act}{\mathsf{Act}} 
\DeclareMathOperator{\act}{\mathsf{act}}  
\DeclareMathOperator{\pth}{\mathsf{path}} 
\DeclareMathOperator{\paths}{\mathsf{paths}}

\DeclareMathOperator{\gAct}{\mathsf{Act}^{\Gamma}} 
\DeclareMathOperator{\gact}{\mathsf{act}^{\Gamma}} 
\DeclareMathOperator{\gacta}{\mathsf{act}_{\aga}^{\Gamma}}

\newtheorem{innercustomclaim}{Claim}
\newenvironment{customclaim}[1]
  {\renewcommand\theinnercustomclaim{#1}\innercustomclaim}
  {\endinnercustomclaim}

\newcommand{\techrep}[1]{}

\newcommand{\gform}{\mathcal{G}}
\newcommand{\gmap}{\mathfrak{g}}
\newcommand{\agi}{\ensuremath{\mathsf{i}\xspace}}
\newcommand{\aga}{\ensuremath{\mathsf{a}\xspace}}
\newcommand{\agb}{\ensuremath{\mathsf{b}\xspace}}
\newcommand{\agc}{\ensuremath{\mathsf{c}\xspace}}

\newcommand{\coop}[2][]{[{#2}]_{_{\!\mathit{#1}}}}

\newcommand{\coal}[2][]{\langle\!\langle{#2}\rangle\!\rangle_{_{\!\mathit{#1}}}}

\newcommand{\atlx}{\mathord \mathsf{X}\, }
\newcommand{\atlf}{\mathord \mathsf{F}\, }
\newcommand{\atlg}{\mathord \mathsf{G}\, }
\newcommand{\atlu}{\, \mathsf{U} \, }

\newcommand{\X}{\atlx}
\newcommand{\F}{\atlf}

\newcommand{\ifff}{\leftrightarrow}

\newcommand{\cgm}{\ensuremath{\mathcal{M}}}

\newcommand{\Logicname}[1]{\ensuremath{\mathsf{#1}}}
\newcommand{\CL}{\Logicname{CL}\xspace}
\newcommand{\LTL}{\Logicname{LTL}\xspace}
\newcommand{\CTL}{\Logicname{CTL}\xspace}

\newcommand{\ATL}{\Logicname{ATL}\xspace}
\newcommand{\ATLs}{\Logicname{ATL^*}\xspace}
\newcommand{\ATLplus}{\Logicname{ATL^+}\xspace}

\newcommand{\SFCL}{\Logicname{SFCL}\xspace}
\newcommand{\SFCLs}{\Logicname{SFCL_{1}}\xspace}
\newcommand{\GPCL}{\Logicname{GPCL}\xspace}
\newcommand{\cga}{\Logicname{TLCGA}\xspace}
\newcommand{\xcga}{\Logicname{XCGA}\xspace}
\newcommand{\fcga}{\Logicname{FTLCGA}\xspace}

\newcommand{\out}{\ensuremath{\mathsf{out}\xspace}}

\newcommand{\states}{\ensuremath{\mathsf{S}\xspace}}
\newcommand{\actprof}{\ensuremath{\zeta\xspace}}
\newcommand{\strat}{\ensuremath{\sigma\xspace}}
\newcommand{\strprof}{\ensuremath{\Sigma\xspace}}
\newcommand{\Out}{\ensuremath{\mathsf{Out}\xspace}}
\newcommand{\outcomes}{\ensuremath{\mathsf{O}\xspace}}
\newcommand{\play}{\ensuremath{\mathsf{play}\xspace}}
\newcommand{\Plays}{\ensuremath{\mathsf{Plays}\xspace}}
\newcommand{\ActProf}{\ensuremath{\mathsf{ActProf}\xspace}}
\renewcommand{\path}{\ensuremath{\pi\xspace}}
\newcommand{\fac}{\ensuremath{\mathcal{F}\xspace}}
\newcommand{\Stratprof}{\ensuremath{\mathsf{StratProf}\xspace}}

\newcommand{\sfml}{\ensuremath{\mathsf{StateFor}\xspace}}
\newcommand{\pfml}{\ensuremath{\mathsf{PathFor}\xspace}}
\newcommand{\supp}{\ensuremath{\mathsf{Support}\xspace}}
\newcommand{\powerset}[1]{\ensuremath{\mathcal{P}({#1})\xspace}}

\newcommand{\defstyle}{\textbf}

\newcommand{\cgoal}[1]{\langle\![ #1]\!\rangle}
\newcommand{\cgoalpos}[1]{\langle\![ #1]\!\rangle_{\mathsf{0}}}

\newcommand{\gmod}{\ensuremath{\mathcal{M}}}

\newcommand{\sat}{\ensuremath{\Vdash}}
\newcommand{\nsat}{\ensuremath{\nVdash}}
\newcommand{\gass}{\triangleright} 

\renewcommand{\models}{\sat}

\newcommand{\langcga}{\ensuremath{\mathcal{L}^\cga\xspace}}
\newcommand{\xlangcga}{\ensuremath{\mathcal{L}^\xcga\xspace}}

\newcommand{\ecl}{\ensuremath{\mathsf{ecl}}}

\newcommand{\mrk}{\ensuremath{\mathbf{m}}}

\newcommand{\network}{\ensuremath{\mathcal{N}}}
\newcommand{\last}{\ensuremath{\mathit{l}}}

\newcommand{\hist}{\ensuremath{\mathsf{Hist}\xspace}}

\newcommand{\fun}{\mathsf{F}}

\newcommand{\xfor}{\mathsf{XFor}}
\newcommand{\ufor}{\mathsf{UFor}}
\newcommand{\gfor}{\mathsf{GFor}}
\newcommand{\lfor}{\mathsf{UGFor}}
\newcommand{\diffof}[1]{{\Delta #1}}
\newcommand{\xpart}{\vert_{\xfor}}
\newcommand{\ugpart}{\vert_{\lfor}}
\newcommand{\gammaof}[1]{\diffof{\gamma}\{#1\}}
\newcommand{\fgam}{\mathsf{Finish}(\gamma)}
\newcommand{\ugam}{\mathsf{UHolds}(\gamma)}
\newcommand{\agam}{\mathsf{GHolds}(\gamma)}
\newcommand{\gamefun}{\mathsf{G}}

\newcommand{\lset}[2]{[#1]_{#2}}

\newcommand{\trn}{\textbf{tr}}

\newcommand{\vcut}[1]{}

\newcommand{\append}[1]{#1}

\newcommand{\Wn}{\ensuremath{\mathsf{W}\xspace}}
\newcommand{\Ls}{\ensuremath{\mathsf{L}\xspace}}
\newcommand{\iWn}{\ensuremath{\mathsf{W_{0}}\xspace}}
\newcommand{\iLs}{\ensuremath{\mathsf{L_{0}}\xspace}}

\newcommand{\unfold}{\unf}

\title{The temporal logic of coalitional goal assignments in concurrent multi-player games} 

\author{Sebastian Enqvist and Valentin Goranko \\ 
Emails: \textrm{\{sebastian.enqvist, valentin.goranko\}@philosophy.su.se}}

\begin{document}
\maketitle

\begin{abstract} 
We introduce and study a natural 
extension of the Alternating time temporal logic \ATL, called \emph{Temporal Logic of Coalitional Goal Assignments} (TLCGA). It features one new and quite expressive coalitional strategic operator, called the \emph{coalitional goal assignment} operator $\brak{\gamma}$, where $\gamma$ is a mapping assigning to each set of players in the game its coalitional \emph{goal}, formalised by a path formula of the language of TLCGA, i.e. a formula prefixed with a temporal operator 
$\nexttime, \until$, or $\always$, representing a temporalised objective for the respective coalition, describing the property of the plays on which that objective is satisfied. 
Then, the formula $\brak{\gamma}$ intuitively says that there is a strategy profile 
$\strprof$ for the grand coalition $\Agt$ such that for each coalition $C$, the restriction $\strprof \vert_C$ of $\strprof$ to $C$ is a collective strategy of $C$ that enforces the satisfaction of its objective $\gamma(C)$ in all outcome plays enabled by $\strprof \vert_C$. 

We establish fixpoint characterizations of the temporal goal assignments in a 
$\mu$-calculus extension of TLCGA, discuss its expressiveness and illustrate it with some examples, {prove bisimulation invariance and Hennessy-Milner property for it with respect to a suitably defined notion of bisimulation}, construct a sound and complete axiomatic system for TLCGA, and obtain its decidability via finite model property. 
\end{abstract}


\section{Introduction}
\label{sec:intro} 

Formalising strategic reasoning has become an increasingly rich and attractive direction of active research and applications of multi-agent modal logics over the past few decades. Early logical systems capturing agents' abilities were developed with philosophical motivations and applications in mind, including  Brown's modal logic of ability \cite{brown:88a} and Belnap and Perloff's STIT logics ~\cite{belnap:88a}. In the late 1990s -- early 2000s two seminal works in the area appeared independently: Pauly's Coalition logic \CL, introduced in \cite{Pauly01phd,Pauly02modal}, and Alur, Henzinger and Kupferman's Alternating time temporal logic \ATL  introduced (in its final version) in \cite{AHK-02}, cf also \cite{TLCSbook}.  

The logic \CL was  introduced with the explicit intention to formalise reasoning about one-step (local) strategic abilities of coalitions of  agents to guarantee the achievement of designated objectives in the immediate outcome of their collective action, regardless of the respective actions of the remaining agents. 
The logic \ATL, on the other hand, was introduced as a logical formalism for formal specification and verification of open (interacting with environment) computer systems, where the agents represent concurrently executed processes. However, it was gradually adopted in the research on logics for multi-agent systems as one of the most standard and popular logical systems for reasoning about long-term strategic abilities of agents and coalitions in concurrent multi-player games. The logic \ATL can be described as an extension of \CL with the long-term temporal operators $\always$ and $\until$, adopted in the branching-time temporal logic \CTL, which can be regarded as a single-agent fragment of \ATL. 
More precisely, both \CL and \ATL feature special modal operators\footnote{We use here the notation from \cite{AHK-02}, which was more widely adopted.}  $\coal{C}{}$, indexed with groups (coalitions) of agents $C$, such that for any formula $\phi$, regarded as expressing the coalitional objective of $C$, the formula $\coal{C}{\phi}$ intuitively says that the coalition $C$ has a collective strategy $\sigma_{C}$ that guarantees the satisfaction of $\phi$ in every outcome (state for \CL, respectively, play for \ATL) that can occur when the agents in $C$ execute their strategies in $\sigma_{C}$, regardless of the choices (strategic or not) of actions of the agents not in $C$. 

Thus, both \CL and \ATL capture reasoning about \emph{absolute powers} of agents and coalitions to act in pursuit of their goals and succeed unconditionally against \emph{any} possible behaviour of their opponents, which are thus regarded as adversaries (in the context of \CL) or as randomly behaving environment (in the context of \ATL).  This is a somewhat extreme perspective, as strategic interactions of rational agents in the real world usually involve a complex interplay of \emph{cooperation} and  \emph{competition}, both driven by the individual and collective objectives of all agents, be them proponents or opponents of the objective in focus. To capture these adequately, expressively richer formal logical languages are needed. In the recent precursor \cite{GorankoEnqvist18} of the present work we proposed two such extensions of \CL with additional coalitional operators, respectively implementing the following two ideas relating cooperation and competition in social context: 
\begin{description}
\item[]  \emph{Social friendliness:} agents can achieve private goals while leaving room for cooperation with the others and with the rest of the society. 
\item[]  \emph{Group protection:} agents can cooperate within the society while simultaneously protecting their private goals.
\end{description}
 
 The second extension mentioned above, called \emph{Group Protecting Coalition Logic} (GPCL) is the starting point of the present work, which introduces and studies its extension in \ATL-like style, called Temporal Logic of Coalitional Goal Assignments (\cga). The logic \cga features one, very expressive, coalitional strategic operator, viz. the \emph{coalitional goal assignment} operator of the type 
$\brak{\gamma}$, where $\gamma$ is a mapping assigning to each coalition (subset)  in the family of all agents $\Agt$ its coalitional \emph{goals}, which is formalised by a path formula of the language of \cga, i.e. a formula prefixed with a temporal operator $\atlx, \atlg$, or $\atlu$, representing the temporalised objective for the respective coalition. 
Then, the formula $\brak{\gamma}$ intuitively says that there is a strategy profile 
$\strprof$ for the grand coalition $\Agt$ such that for each coalition $C$, the restriction $\strprof \vert_C$ of $\strprof$ to $C$ is a collective strategy of $C$ that enforces the satisfaction of its objective $\gamma(C)$ in all outcome plays enabled by $\strprof \vert_C$.  The intuition is that each agent participates in the grand coalition with its individual strategy so that, while contributing to the achievement of the common goal, each coalition also guarantees the protection of its coalitional interest agains any possible deviation of all other agents. 
	The logic \cga naturally extends \ATL (in particular, \CL) and exhibits richer expressiveness, both purely technical, but also in terms of potential applications (cf. Section \ref{sec:applications}). Notably, it turns out (cf. Section \ref{subsec:CGAsemantics-variations}) that the semantics of \cga is more sophisticated than the semantics of \ATL and most of its extensions studied so far, in sense that it is essential that the class of (memory based) strategies underlying the semantics of \cga includes all \emph{play-based strategies} (taking into account the full history, consisting not only of the sequence of visited states, but also including the action profiles causing the transitions between them) rather than only the \emph{path-based strategies} (based on the state histories only), which has been the customary choice both for the logics in the \ATL/\ATLs family and for the family of Strategy logics (see further) studied so far. More precisely, the two semantics differ for \cga, both in terms of truth at a state in a model and in terms of validity (respectively, satisfiability), as shown in Section \ref{subsec:CGAsemantics-variations}, where we also argue that the semantics using play-based strategies is the more natural and faithful one for the intended meaning of the operator $\brak{\cdot}$. We also analyse there how the two semantics relate technically, and  show how model checking for the semantics using play-based strategies can be reduced to model checking for the semantics using path-based strategies.  
 
	Besides $\cga$, we also introduce and study, though in less detail, two extensions of this language: $\cga^+$, in which conjunctions  of path formulas are allowed, and a fixpoint language $\xlangcga_\mu$ in the style of the modal $\mu$-calculus. The first of these two extensions appears naturally in some applications with a connection to game theory, described shortly. The full fixpoint language $\xlangcga_\mu$ has a theoretical interest on its own, as a very expressive, yet quite well behaved logic for multi-player games. Furthermore, since the logics $\cga$ and $\cga^+$ embed as fragments of $\xlangcga_\mu$, the latter can be used as a tool to study these logics and obtain technical results about them. Here, it is used to prove decidability  and finite model property of both $\cga$ and $\cga^+$, as well as to establish upper complexity bounds for their satisfiability problems.
 
	As we demonstrate with examples in Section \ref{sec:applications},  the logic \cga enables the expression of various natural and important nuanced patterns of multi-player strategic interaction. 
 In particular, the logic \cga captures a concept that we call ``\emph{co-equilibrium}'', which we define and promote here as a new, alternative solution concept on the border of  non-cooperative and cooperative game theory \cite{OR2}. We argue (in Section \ref{subsec:co-equilibria}) that is more natural and applicable than the standard notion of Nash equilibrium in the context of  concurrent multi-player games with individual qualitative objectives. Existence of a co-equilibrium can be expressed quite simply in \cga using the operator $\brak{\gamma}$.
 We also show (in Section \ref{subsec:stable-outcomes}) how other naturally defined notions of stable individually and coalitionally stable outcomes can be formalised in $\cga^+$.
 
 Further motivation for the present work comes from \emph{cooperative game theory} 
\cite{OR2}, 
\cite{BranzeiDimitrovTijs}, 
\cite{DBLP:series/synthesis/2011Chalkiadakis}. One natural link is the apparent relationship between concurrent game models and coalitional goal assignments in them studied here, and some classes of cooperative games, such as the  
so called \emph{simple games} \cite{Ramamurthy90},  \cite{vanDeemen97}, where the characteristic function defining the game assigns a payoff 0 or 1 to each coalition. An important class of such games are the \emph{voting games} \cite{DBLP:series/synthesis/2011Chalkiadakis}. We only mention here these links with cooperative game theory, but they are left to be explored in a further work. 
In this paper we only illustrate briefly (in Section \ref{subsec:applications2GT}) the expressiveness of \cga and $\cga^+$ by showing how natural qualitative analogues of the key notions of \emph{stable strategy profiles} and \emph{core}, defined and studied for a type cooperative games based on concurrent game models in \cite{DBLP:conf/atal/GutierrezKW19}, can be expressed in $\cga^+$.

\paragraph{Main contributions.}
Besides the introduction of the logic  \cga and its extensions $\cga^+$ into $\xlangcga_\mu$,  the main technical contributions of this paper are: 
\begin{itemize}
\item Fixpoint characterizations of the main types of long-term goal assignments and translation of $\cga$ and $\cga^+$ into $\xlangcga_\mu$. 

\item  bisimulation invariance and Hennessy-Milner property for the logic \cga and for the full $\xlangcga_\mu$ with respect to the GPCL-bisimulation introduced in \cite{GorankoEnqvist18}.  

\item sound and complete axiomatic system for \cga.  

\item finite model property and decidability (with triple exponential bound) for the extended logic $\cga^+$ (hence also of \cga).

\item double exponential bound on the satisfiability problem for the fixpoint language $\xlangcga_\mu$.

\item complexity bounds for the model checking problems and the  satisfiability problems for \cga and $\cga^+$.
\end{itemize}

\paragraph{Related work and results.}
In addition to the links and references mentioned so far, the present work bears both conceptual and technical connections with several previously studied logics for strategic reasoning and multi-player games, including: the logic \ATL with irrevocable strategies \cite{AgotnesTARK2007,Jamroga08commitment-tr}, \ATL with strategy contexts \cite{BrihayeLLM09}, coalitional logics of cooperation and propositional control \cite{HoekW05, TroquardHW09}, cooperative concurrent games \cite{DBLP:conf/atal/GutierrezKW19}, and especially with the family of Strategy Logics, originally introduced in \cite{DBLP:journals/iandc/ChatterjeeHP10} and further extended and  studied in \cite{fsttcs/MogaveroMV10}, 
\cite{tocl/MogaveroMPV14}, 
\cite{corr/MogaveroMPV16},  
\cite{aamas/AminofMMR16},
\cite{DBLP:conf/aaai/AcarBM19},
\cite{DBLP:journals/mst/GardyBM20}, etc. 
Indeed, the operator $\brak{\gamma}$ in $\cga^+$  
with the `path-based semantics' mentioned earlier (cf. Section \ref{subsec:CGAsemantics-variations}) can be translated to the Conjunctive Goals fragment\footnote{We thank an anonymous reviewer for suggesting to consider the translation of \cga into that fragment.} SL[CG] of Strategy Logic  \cite{DBLP:conf/lics/MogaveroMS13},  \cite{DBLP:journals/mst/GardyBM20}, in a way similar to the standard translation of modal logics to first-order logic. 
 As shown in Section \ref{subsec:ST2SL}, that translation can be used for applying model checking algorithms for SL[CG] to $\cga^+$ (in particular, \cga) and for obtaining \textrm{2ExpTime} complexity upper bound for the model checking problem for $\cga^+$, for both versions of its semantics.
Likewise, that translation can be used for applying algorithms and obtaining \textrm{PSpace} complexity upper bound for the satisfiability problem in flat fragment of $\cga^+$ with the path-based semantics
by reduction to the satisfiability problem in flat fragment FSL[CG] of SL[CG], shown in  \cite{DBLP:conf/aaai/AcarBM19} to be decidable and \textrm{PSpace-complete}.    
On the other hand, a direct argument in Section \ref{sec:FMP}, using embedding of \cga and $\cga^+$ to the $\mu$-calculus extension $\xlangcga_\mu$ 
shows that the satisfiability problem in the whole $\cga^+$ (in particular, \cga), with its standard play-based semantics, is in \textrm{2ExpTime}.

	Notwithstanding these important technical benefits, we note that translating \cga to (a fragment of) Strategy Logic could generally result in unwanted surplus of expressiveness, or even in a technical overkill, for which we have both conceptual and computational reasons to avoid, whenever possible.  
On the conceptual side, translating \cga to Strategy Logic would lose of the elegant succinctness and clear focus of the operator $\brak{\gamma}$ as the main high-level logical construct of the language and would replace it with its low-level description in  Strategy Logic\footnote{Admittedly, this is a subjective argument and mostly a matter of taste, which we only state but do not try to impose here.}.   
On the technical side, such translation would map a syntactically simple propositional language to a generally quite more expressive and syntactically heavier, essentially second-order language, explicitly involving quantification over strategies (being functions from finite sequences of states to actions). Essentially these are the same arguments in favour of preferring modal logic over first-order logic, but amplified by the technical complexity of quantifying over functions rather than individuals. 
Thus, we eventually adopt and advocate the pragmatic approach of adhering as much as possible to the propositional logic framework of \cga for the purposes of expressing and reasoning about properties of strategic interactions of the type mentioned above, while resorting to using translation to  fragments of Strategy Logic only when necessary or practically expedient, e.g. for the purpose of using already developed tools for model or satisfiability checking in these fragments of Strategy Logic. 

Our work is also essentially connected with \emph{coalgebraic modal logic} \cite{moss1999coalgebraic,pattinson2003coalgebraic,cirstea2007modular}, which is an abstract framework for modal logics of state-based evolving systems. Together with the fixpoint characterization of \cga, this makes \cga in essence a fragment of a \emph{coalgebraic fixpoint logic}  
\cite{venema2006automata,fontaine2010automata,cirstea2011exptime}. 
This connection is used to establish  decidability and finite model property for our logic. Beyond that, however, our presentation is mostly self-contained, and will not require familarity with coalgebra. 

Still, we want to emphasize that the connection with coalgebraic logic, and coalgebraic fixpoint logics in particular, is implicitly present throughout the paper. In particular the notion of \emph{one-step completeness}, and the idea of lifting one-step completeness to completeness for the full language, is at the heart of our completeness proof. The notion of one-step completeness is inherently coalgebraic and has been studied in depth by a number of authors \cite{schroder2009pspace,schroder2008expressivity,pattinson2003coalgebraic}. 

Furthermore, the fact that our translation into fixpoint logic requires only a single recursion variable means that \cga is a fragment of a \emph{flat}  fixpoint logic \cite{SantocanaleV10,SchroderV10,enqvist2018flat}, and completeness of flat fixpoint logics can be obtained by simpler techniques than the full $\mu$-calculus. There are two main reasons why our completeness proof is not explicitly formulated in coalgebraic terms: first, we note that \cga is not a flat $\mu$-calculus per se, but rather embeds into such a logic via a fairly intricate translation. So the results in \cite{SchroderV10} do not apply directly here, as far as we can see. Second, and more importantly, we want the proof to be as self-contained and accessible without prior knowledge in coalgebraic $\mu$-calculus as possible. It is possible that one could ``transfer'' completeness of flat coalgebraic $\mu$-calculi to obtain completeness for \cga via our translation, but we believe a direct completeness proof is more transparent and provides better understanding of \cga and its semantics.

Lastly, we also note the relationship of the present work with the logic for local conditional strategic reasoning CSR introduced in \cite{LORIVII-GorankoJu}. 
Furthermore, we point out the direct applicability of the logic \cga for adequate alternative formalisation of the ideas of \emph{rational synthesis} \cite{tacas/FismanKL10} and \emph{rational verification} \cite{aaai/WooldridgeGHMPT16}. 
These connections and possible applications are left to future work.

\paragraph{Structure of the paper.} 
After some preliminaries in Section \ref{sec:prelim} on concurrent game models, plays and strategies in them, in Section \ref{sec:TLCGA} we introduce and study the formal syntax and semantics of the logic \cga, and illustrate its expressiveness with some examples.   
In Section  \ref{sec:fixpoints} we obtain fixpoint characterizations of the long-term goal assignments expressed in a suitable $\mu$-calculus extension of \cga. We then discuss the connection with coalgebraic modal logic.  
In Section \ref{subsec:Bisimulations} we introduce the relevant notion of bisimulation for \cga and prove bisimulation invariance and the Hennessy-Milner property for it. 
  In Section \ref{sec:axiomatization} we provide an axiomatic system for \cga for which we prove soundness and completeness. In Section \ref{sec:FMP} we show decidability of \cga  via finite model property. We then end with brief concluding remarks in Section \ref{sec:concluding}.

\section{Preliminaries and background}
\label{sec:prelim} 

\subsection{Concurrent game models, plays, strategies}
\label{subsec:CGM}

We fix a finite set of \defstyle{players/agents} $\Agt = \{\aga_1,...,\aga_n\}$ 
and a set of  \defstyle{atomic propositions} $\Prop$. 
Subsets of $\Agt$ will also be called \defstyle{coalitions}.  

Given a set $W$, we denote by $W^*$ the set of finite words over $W$, by $W^+$ the set of non-empty words from $W^*$, and by $W^\omega$ the set of infinite words over $W$.

\begin{definition}
Let $\outcomes$ be any non-empty set. A \defstyle{(strategic) game form over the set of outcomes $\outcomes$} is a tuple 
\[\gform = 
(\Act,\act,\outcomes,\out)\] 
where 
\begin{itemize}
\item $\Act$ is a non-empty set of  \defstyle{actions}, 

\item $\act: \Agt \to \psf^{+}(\Act)$ is a mapping assigning to each $\aga \in \Agt$ a non-empty set $\act_\aga$ of \defstyle{actions available to the player $\aga$},

\item $\out: \Pi_{\aga \in \Agt}\act_\aga \to \outcomes$ is a map assigning to every \defstyle{action profile} 
$\actprof \in \Pi_{\aga \in \Agt} \act_\aga$ a unique \defstyle{outcome} in $\outcomes$.  
\end{itemize}

\end{definition}

\begin{definition}
A \defstyle{concurrent game model} is a tuple  
\[\gmod = (\states,\Act,\gmap,V)\] 
where
\begin{itemize}
\item  $\states$ is a non-empty set of \defstyle{states},

\item $\Act$ is a 
non-empty set of  \defstyle{actions}, 

\item $\gmap$ is a \defstyle{game map}, assigning to each state 
$w \in \states$ a strategic game form 
$\gmap(w) = (\Act,\act_w,\states,\out_w)$ over the set of outcomes $\states$.

\item $V: \Prop \to \powerset{\states}$ is a \defstyle{valuation} of the atomic propositions in $\states$;  
\end{itemize}

For every concurrent game model $\gmod = (\states,\Act,\gmap,V)$ we define the following. 
\begin{itemize}
\item 
For each $\aga \in \Agt$ and $w \in \states$, the set  
$\act_w(\aga)$ consists of the \defstyle{locally available actions} for $\aga$ in $w$. 
It will also be denoted by $\act(\aga,w)$.   
We also define the set $\act_\aga :=  \bigcup_{w\in\states} \act_w(\aga)$ of \defstyle{globally available actions} for $\aga$.

\item 
An  \defstyle{action profile} is a tuple of actions 
 $\actprof \in \Pi_{\aga \in \Agt} \act_\aga$. 
A  \defstyle{locally available action profile at state $w$} is  any tuple of locally available actions 
 $\actprof \in \Pi_{\aga \in \Agt} \act_w(\aga)$. 
  The set of these action profiles will be denoted by $\ActProf_w$. 
  
\item $\out_\gmod$ 
is the \defstyle{global outcome function} assigning to every state $w$ and a local action profile $\actprof$ at $w$ 
a unique \defstyle{outcome} $\out_\gmod(w,\actprof) := \out_w(\actprof)$.  
When $\gmod$ is fixed by the context, it will be omitted from the subscript.

\item 
 Given a coalition $C \subseteq \Agt$, a \defstyle{joint action} for $C$ in $\gmod$ is a tuple of  individual actions $\actprof_{C} \in \prod_{\aga \in C} \act_\aga$. 
In particular, for any action profile 
 $\actprof \in \Pi_{\aga \in \Agt} \act_\aga$, $\actprof \vert_C$ is the joint action obtained by restricting $\actprof$ to $C$.  

\item 
 For any $w\in \states$, $C \subseteq \Agt$, and joint action $\actprof_{C}$ that is available at $w$, we define:
\[
\Out[w,\actprof_{C}] = \left\{u \in S \mid \exists \actprof \in 
\prod_{\aga \in \Agt} \act_{w}(\aga):  
\; \actprof \vert_C = \actprof_{C} \mbox{ and }  \out(w,\actprof) = u \right\}.  
\]
\end{itemize}

A \defstyle{partial play}, or a \defstyle{history}  
in $\gmod$ is either an element of $\states$ or a finite word of the form:
$$w_0 \actprof_0 w_1 ... w_{n - 1} \actprof_{n - 1} w_n$$
where $w_0,...,w_n \in \states$ and for each $i < n$, $\actprof_i$ is a locally available action profile in $\Pi_{a \in \Agt}\act(a,w_i)$. 
The last state in a history $h$ will be denoted by $\last(h)$. 
The set of histories in $\gmod$ is denoted by $\hist(\gmod)$.

A \defstyle{(memory-based) strategy for player $\aga$} is a map $\strat_\aga$ assigning to each history $h = w_0 \actprof_0... \actprof_{n-1}w_n$ in $\mathsf{Play}$ an action $\strat_\aga(h)$ from $\act(\aga,w_n)$. 
Note that strategies are defined here in terms of \emph{histories}, i. e. \emph{partial plays}, not just sequences of states, as it is customary for \ATLs and in particular \ATL \cite{AHK-02}, cf also 
\cite{TLCSbook} or \cite{BGJ15}. This distinction will turn out to be essential for the semantics of the logic introduced here.  
A strategy $\strat_\aga$ is \defstyle{memoryless}, or  \defstyle{positional}, if it assigns actions only based on the current (last) state, i.e. $\strat_\aga(h) = \strat_\aga(h')$ whenever $\last(h) = \last(h')$.

 Given a coalition $C \subseteq \Agt$, a \defstyle{joint strategy} for $C$ in the model 
 $\gmod$ is a tuple $\strprof_{C}$ of individual strategies, one for each player in $C$.   
A \defstyle{(global) strategy profile} $\strprof$ is a joint strategy for the grand coalition $\Agt$, i.e. an assignment of a strategy to each player. 
We denote the set of all strategy profiles in the model $\gmod$ by $\Stratprof_\gmod$, and the set of all joint strategies for a coalition  $C$ in $\gmod$ by  $\Stratprof_\gmod(C)$. Thus,  $\Stratprof_\gmod = \Stratprof_\gmod(\Agt)$.

Given a strategy profile $\strprof$, the \defstyle{play} induced by $\strprof$ at $w \in \states$ is the unique infinite word 
\[\play(w,\strprof) =  w_0 \actprof_0 w_1 \actprof_1 w_2 \actprof_2...\]
such that $w_0 = w$ and, for each $n < \omega$ we have $w_{n + 1} = \out(\actprof_n,w_n)$, and 
\[
\actprof_{n + 1} = \strprof(w_0 \actprof_0 ... \actprof_{n}w_{n +1})
\]
The infinite word $w_0w_1w_2...$ obtained by simply forgetting the moves of players in this infinite play is called the \defstyle{computation path} induced by $\strprof$ at $v$, and denoted $\pth(\strprof,v)$.

More generally, given a coalition $C \subseteq \Agt$, a state $w \in \states$ 
and a joint strategy $\strprof_{C}$  for $C$ we define the  \defstyle{set of outcome plays induced by the joint strategy $\strprof_{C}$ at $w$} to be the set of plays 
\[\Plays(w,\strprof_{C}) =  \big\{\play(w,\strprof) \mid \strprof \in \Stratprof_\gmod  
\mbox{ such that }  \strprof(\aga) = \strprof_{C}(\aga) \mbox{ for all } \aga \in C \big\}
\]
Given a strategy profile $\strprof$ we also denote 
$\Plays(w,\strprof,C) := \Plays(w,\strprof\vert_{C})$.  
We will likewise use the notation  
$\paths(w,\strprof,C)$ for the set of computation paths obtained from the plays in 
$\Plays(w,\strprof,C)$. 
Since these only depend on the strategies assigned to players in $C$, we shall freely use the notation $\Plays(w,\strprof,C)$  and $\paths(w,\strprof,C)$ even when $\strprof$ is defined for all members of $C$, but not for all other players. 
\end{definition}

The strategies in our semantics will be memory-based: moves of players in a strategy may depend on previous moves of other players, and players have perfect information and recall about previous moves. 
Indeed, as we will show in Section \ref{subsec:CGAsemantics}, 
just like for \ATLplus and \ATLs, but unlike \ATL, the restriction to positional strategies generates different semantics for the logic \cga which we introduce here.

%
\section{The temporal logic of coalitional goal assignments \cga}
\label{sec:TLCGA}

\subsection{Goal assignments, language and syntax of \cga}
\label{subsec:CGAsyntax}

Given a fixed finite set players
$\Agt$ and a set $G$ of objects, called `goals', a \defstyle{(coalitional) goal assignment for $\Agt$ in $G$} is a mapping $\gamma: \psf(\Agt) \to G$.   
 
We now define the set $\sfml$ of 
\defstyle{state formulae} and the set $\pfml$ of 
\defstyle{path formulae} of \cga by mutual induction, using the following  BNF:

\medskip
$\sfml: \ \ \ \ \ \varphi := p  \mid \top 
\mid \neg \varphi 
\mid (\varphi \wedge \varphi) 
 \mid (\varphi \lor \varphi) 
\mid  
\brak{\gamma}$

\smallskip
$\pfml: \ \ \ \ \ \theta := \nexttime \varphi \mid \varphi \until \varphi \mid \always \varphi$
\smallskip

where $p \in \Prop$ and $\gamma: \psf(\Agt) \to \pfml$ is a goal assignment for $\Agt$ in 
$\pfml$. The other propositional connectives $\bot$, $\to$ and $\ifff$, as well as the temporal operator $\F$, are defined as usual. We write:  $\xfor$ for the set of path formulas 
of the form $\nexttime \varphi$; $\ufor$ for the set of path formulas 
of the form $\varphi \until \psi$; 
$\gfor$ for the path formulas 
of the form $\always \varphi$; and $\lfor$ for $\ufor \cup \gfor$.  

We denote the language by $\langcga$, and its nexttime fragment (where 
$\pfml$ is restricted to $\xfor$) by $\xlangcga$. 
The latter is essentially (with some minor notational changes) the language of the logic \GPCL introduced in \cite{GorankoEnqvist18}. 

Intuitively, the path formulae can be regarded as temporal goals. The goal $\nexttime \top$ is called a \defstyle{trivial goal} and all other goals in $\pfml$ are  \defstyle{non-trivial goals}. The family of coalitions $\fac$ to which the goal assignment $\gamma$ assigns non-trivial goals is called the \defstyle{support of $\gamma$}, denoted $\supp(\gamma)$, and $\gamma$ is said to be 
\defstyle{supported by $\fac$}.  

Sometimes we will write a goal assignment $\gamma$ explicitly, like 
\[C_1\gass \theta_1,...,C_n\gass \theta_n,\]
meaning that 
$\supp(\gamma) = \{C_1,...,C_n\}$ and $\gamma(C_i) = \theta_i$, for $i=1,...,n$.

More notation:  
\begin{itemize}
\item $\gamma^\top$ is the \defstyle{trivial goal assignment}, mapping each coalition to $\nexttime \top$. 

\item 
The goal assignment $\gamma[C \gass \theta]$ is like $\gamma$, but mapping $C$ to $\theta$.

\item 
The goal assignment $\gamma \setminus C$ defined as  $\gamma[C \gass \nexttime \top]$  is like $\gamma$, but excluding $C$ from its support, by replacing its goal with 
$\nexttime \top$.

\item The goal assignment $\gamma\vert_C$ is defined by mapping each  $C' \subseteq C$ to $\gamma(C')$ and mapping all coalitions not contained in $C$ to $\nexttime\top$.

\end{itemize}

As a convention, if $\gamma$ is 
the unique goal assignment with empty support, we will identify the formula $\brak{\gamma}$ with $\top$.

We will also consider a slight extension of the logic $\cga$ which allows forming conjunctions of path formulas. This extension, which we call $\cga^+$, has the same definition of state formulae, whereas the path formulae of $\cga^+$ are defined as follows: 

\smallskip
$\pfml: \ \ \ \ \ \theta := \nexttime \varphi \mid \varphi \until \varphi \mid \always \varphi \mid \theta \wedge \theta$
\smallskip 

 In this paper we will focus mainly on \cga, but most of the technical results extend likewise to $\cga^+$, and we will note that explicitly for some of them.

\subsection{Semantics of \cga}
\label{subsec:CGAsemantics}

The semantics of \cga is defined in terms of truth of state formulae at a state, respectively truth of path formulae on (the path generated by) a play, 
in a concurrent game model 
$\gmod = (\states,\Act,\gmap,\out,V)$. 
The truth clauses are like in classical logic for the boolean connectives and like in LTL for the temporal operators. 
The only new clause, for $\cgoal{\gamma}$, is as follows, 
where $s \in \states$:

\smallskip
\begin{quotation}
$\gmod, s \models \cgoal{\gamma}$ \ iff \ there exists a strategy profile 
$\strprof \in \Stratprof_\gmod$ such that, \\ 
for each $C \subseteq \Agt$, it holds that 
$\gmod, \path \models \gamma(C)$ for every $\path \in \paths(s,\strprof,C)$.   
\end{quotation}

\smallskip
For any state formula $\varphi \in \sfml$ we define \defstyle{the extension of $\varphi$  in $\gmod$} to be the set of states in $\gmod$ where $\varphi$ is true: 
$\tset{\varphi}_{\gmod} = \{ s \in S \mid \gmod, s \models \varphi \}$.  
Likewise, we define the extension of any path formula $\theta \in \pfml$
to be the set of paths in $\gmod$ where $\theta$ is true: 
$\tset{\theta}^{p}_{\gmod} = \{ \path \in S \mid \gmod, \path  \models \theta \}$. 
The truth clause for $\cgoal{\gamma}$ can now be re-stated in terms of formula extensions  as follows: 
\[
\tset{\cgoal{\gamma}}_{\gmod} = 
\big\{ s \in S \mid \exists\, \strprof  \in \Stratprof_\gmod:  
\paths(s,\strprof,C) \subseteq \tset{\gamma(C)}^{p}_{\gmod} 
\ \mbox{for each} \ C \subseteq \Agt \big\}.
\]

A strategy profile $\strprof$ is said to \defstyle{witness the goal assignment $\gamma$}  at a state $s$ of a model $\gmod$, denoted by $\strprof,s \Vvdash \gamma$, if, for every coalition $C$ in the support of $\gamma$ and every path $\pi \in \paths(s,\strprof,C)$ in  $\gmod$ we have  $\gmod, \pi \sat \gamma(C)$. 
We then also say that \defstyle{$\strprof$ witnesses the formula $\brak{\gamma}$} at the state $s$ in $\gmod$. 
Thus, $\gmod, s \sat \brak{\gamma}$ iff ${\gamma}$ is witnessed by some strategy profile at $s$ in $\gmod$.

A \cga formula $\phi$ is \defstyle{valid}, denoted 
$\sat \phi$, if $\gmod, s \sat \phi$ for every concurrent game model $\gmod$ and a state $s$ in it; respectively $\phi$ is \defstyle{satisfiable} if $\gmod, s \sat \phi$ for some concurrent game model $\cgm$ and some state $s$ in it. Likewise, (local) logical consequence and  logical (semantic) equivalence in \cga are defined as expected.

\smallskip
We note that the formula $\cgoal{C \gass \X\phi, \, \Agt \gass \X \psi}$ is semantically equivalent to the strategic operator $\coop{C}(\phi; \psi)$ in the logic \SFCL defined in \cite{GorankoEnqvist18}. Therefore, the corresponding fragment \SFCLs of \SFCL  embeds into \cga. Note also that the strategic operator $\coop{C}$ from Coalition logic \CL  is definable as a special case: 
$\coop{C}\phi := \coop{C}(\phi; \top) \equiv \cgoal{C\gass\X\phi}$. 

Furthermore, the strategic operator $\coal{C}$ from \ATL  is also a special case: 
$\coal{C}{\theta}  \equiv \cgoal{C\gass \theta}$. 
Thus, the logic \ATL is embedded as a simple fragment of \cga.

Also, we note that a natural (downward) monotonicity condition on goal assignments can be imposed, viz. that $\models \gamma(C) \to \gamma(C')$ for every coalitions $C$, $C'$, such that $C' \subseteq C$. This condition can be imposed semantically, up to equivalence, by replacing each $\gamma(C)$ with $\bigwedge_{C' \subseteq C} \gamma(C')$, though the resulting goals cannot be expressed in \cga, but in $\cga^+$. A special case of monotone goal assignments, worth noting, are the `individualistic' goal assignments, where $\gamma(C) \equiv \bigwedge_{\aga \in C} \gamma(\aga)$ for every coalition $C$. 
%

\subsection{Variations of the semantics of \cga}
\label{subsec:CGAsemantics-variations}

\subsubsection{Memory-based and the memoryless semantics}

Let us introduce ad hoc the variation $\cgoalpos{\cdot}$ of $\cgoal{\gamma}$, with semantics restricted to positional strategies, i.e.:  
$\gmod, s \models \cgoalpos{\gamma}$ \ iff \ there exists a \emph{positional} strategy profile 
$\strprof \in \Stratprof_\gmod$ such that, for each $C \subseteq \Agt$, it holds that 
$\gmod, \path \models \gamma(C)$ for every $\path \in \paths(s,\strprof,C)$. 

\begin{proposition}[No positional determinacy of \cga] 
\label{prop:NoPositionalDeterminacy} 
Let $\Agt = \{\aga,\agb\}$.  
There exist concurrent game model $\gmod = (\states,\Act,\gmap,\out,V)$, state $s \in \states$ and a coalitional goal assignment $\gamma$, such that 
$\gmod, s \models \cgoal{\gamma}$, but 
$\gmod, s \not\models \cgoalpos{\gamma}$. 

Consequently, the memory-based and the memoryless semantics of $\cgoal{\cdot}$ 
are not equivalent. 
\end{proposition}

\begin{proof}
Consider the model  $\gmod = (\states,\Act,\gmap,\out,V)$ on Figure \ref{exampleA} and a goal assignment $\gamma$, such that  
$\gamma(\{\aga,\agb \}) = p \atlu q$ and 
$\gamma(\{\aga \}) = \top  \atlu \neg(p \lor q)$.

Then, $\gmod, s \models \cgoal{\gamma}$, witnessed by any strategy profile $\strprof$ such that  
$\strprof_{\aga} (s) = a_1$ and $\strprof_{\aga} (s s_1 s) = a_2$. 

However, there is no positional strategy profile witnessing the truth of 
$\cgoalpos{\gamma}$ at $s$ because any positional strategy for $\aga$ would have to assign a unique action to any history ending at $s$, hence not enabling both the satisfaction of $p \atlu q$  and of $\top  \atlu \neg(p \lor q)$ there.

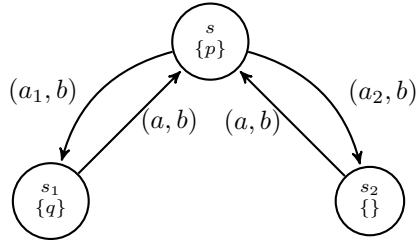
\begin{figure}
\begin{center}
\medskip
\begin{tikzpicture}[->,>=stealth', shorten >=1pt, node distance=30mm, thick]
\tikzstyle{every state}=[fill=white,draw=black,text=black]

\node[state] (0) {$s \atop  \{p \}$};
\node[state] (1) [below left of=0] {$s_1 \atop \{ q\}$}; 
\node[state] (2) [below right of=0] {$s_2 \atop \{ \}$};

\path

(0) edge [bend right] node {$(a_1, b)$\ \ \ \ \ \ \ \ \ \ \ \ \ ~} (1)
(0) edge [bend left] node {$~ \ \ \ \ \ \ \ \ \ \ \ \ \ (a_2, b)$} (2)

(1) edge [right] node {$(a, b)$} (0)
(2) edge [left] node {$(a, b)$} (0);
\end{tikzpicture}
\end{center}
\caption{Example showing that memory is needed}
\label{exampleA} 
\end{figure}
\end{proof}

Hereafter, we will only work with the memory-based semantics. 

\subsubsection{Semantics with path-based strategies vs play-based strategies}
Furthermore, note that (memory-based) strategies are defined here in terms of \emph{plays}, not just paths, as it is customary for \ATLs and in particular \ATL \cite{AHK-02} (cf also \cite{BGJ15} or \cite{TLCSbook}), as well as for Strategy Logic \cite{DBLP:journals/iandc/ChatterjeeHP10}, \cite{fsttcs/MogaveroMV10}, \cite{tocl/MogaveroMPV14}. Indeed, the two versions of strategy types affect essentially the semantics, as  shown by the following example.

\begin{example}[Path-based strategies vs play-based strategies]
\label{exampleB}

Consider the model $\gmod$ below, with 3 players: $\{1,2,3\}$, where the triples of actions correspond to the order $(1,2,3)$ and $*$ denotes any (or, a single) action.

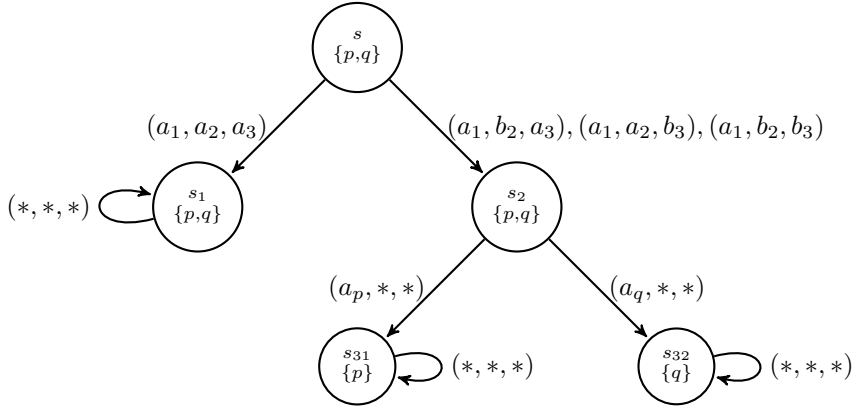
\begin{center}
\medskip
\begin{tikzpicture}[->,>=stealth', shorten >=1pt, node distance=30mm, thick]
\tikzstyle{every state}=[fill=white,draw=black,text=black]

\node[state] (0) {$s \atop  \{p, q\}$};

\node[state] (01) [below left of=0] {$s_1 \atop \{p, q\}$}; 
\node[state] (02) [below right of=0] {$s_2 \atop \{p,q \}$};

\node[state] (1) [below left of=02] {$s_{31} \atop \{ p\}$}; 
\node[state] (2) [below right of=02] {$s_{32} \atop \{ q\}$};

\path
(0)  edge [left] node {$(a_1,a_2,a_3)$} (01)

(0)  edge [right] node {$(a_1, b_2, a_3), (a_1, a_2,b_3), (a_1, b_2, b_3)$} (02)

(01) edge [loop left] node {$(*,*,*)$} (01)

(1) edge [loop right] node {$(*,*,*)$} (1)
(2) edge [loop right] node {$(*,*,*)$} (2)
 
(02) edge [left] node {$(a_p,*,*)$} (1)
(02) edge [right] node {$(a_q, *,*)$} (2); 

\end{tikzpicture}
\end{center}
\label{fig:exampleB} 

\medskip
Consider the goal assignment $\gamma$, such that $\gamma(\{1,2 \}) = \atlg p$ and 
$\gamma(\{1,3 \}) = \atlg q$. 
The following hold: 

\begin{enumerate}  

\item  $\gmod, s \models \cgoal{\gamma}$ in terms of the semantics with plays-based strategies adopted here.   

Indeed, the strategy profile $\strprof$ prescribing the following action profiles: 
$(a_1,a_2,a_3)$ on the play $s$;  
$(a_p,*,*)$ on the play $s (a_1,a_2,b_3) s_2$;  
$(a_q,*,*)$ on the plays $s (a_1,b_2,a_3) s_2$ and $s (a_1,b_2,b_3) s_2$;  
and $(*,*,*)$ on any play ending at $s_1$, $s_{31}$, and $s_{32}$,
 would ensure the truth of  $\cgoal{\gamma}$ at $s$.

\item  $\gmod, s \not\models \cgoal{\gamma}$  in terms of the semantics with path-based strategies.

This is because player $1$ does not have such strategy for which both:  

\begin{enumerate}
\item the coalition $\{1,2\}$ ensures satisfaction of the goal $\atlg p$ by transition from $s_2$ to $s_{31}$, if $3$ acts $b_3$ at $s$ and the game goes to $s_2$, and 

\item the coalition $\{1,3\}$ ensures satisfaction of the goal $\atlg q$ by transition from $s_2$ to $s_{32}$, if $2$ acts $b_2$ at $s$ and the game goes to $s_2$. 

\end{enumerate}

\end{enumerate}
\end{example}

Thus, two different semantics for \cga emerge. Let us call them respectively \defstyle{play-based semantics}, hereafter indicated by $ \models_{\mathsf{play}}$, and 
\defstyle{path-based semantics}, hereafter indicated by $ \models_{\mathsf{path}}$.    
Now, a natural question arises: which is the better / more correct one? 
We argue that, as the example above indicates, this is the semantics based on play-based strategies, adopted here, because only play-based strategies can detect agents' deviations from the adopted strategy profile of the grand coalition, so as to ensure that the execution of these strategies by the non-deviating agents will still guarantee the fulfilment of their individual and collective goals.

\subsubsection{Relating the path-based and the play-based semantics}

Still, the two semantics can be reconciled in the example above by splitting off the node $s_2$ into 3 copies, each being the successor node of  $s_1$ for exactly one action profile, as on Figure \ref{fig:exampleB2}. (In fact, for this example, 2 copies suffice, just to split the outcomes from $(a_1, b_2, a_3)$ and 
$(a_1, a_2,b_3)$.) That way, different action profiles applied at any given state lead to different successor nodes, so different plays correspond to different paths, hence the two semantics coincide in the resulting model. 

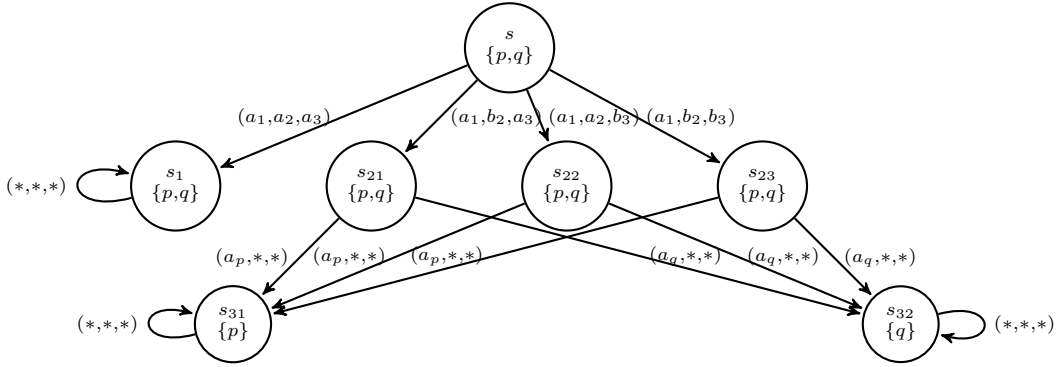
\begin{center}
\medskip
\begin{tikzpicture}[->,>=stealth', shorten >=1pt, node distance=26mm, thick]
\tikzstyle{every state}=[fill=white,draw=black,text=black]

\node[state] (0) {$s \atop  \{p, q\}$};

\node[state] (021) [below left of=0] {$s_{21} \atop \{p, q\}$}; 
\node[state] (01) [ left of=021] {$s_1 \atop \{p, q\}$}; 
\node[state] (022) [right of=021] {$s_{22} \atop \{p,q \}$};
\node[state] (023) [ right of=022] {$s_{23} \atop \{p,q \}$};

\node[state] (1) [below left of=021] {$s_{31} \atop \{ p\}$}; 
\node[state] (2) [below right of=023] {$s_{32} \atop \{ q\}$};

\path
(0)  edge [left] node {$_{(a_1,a_2,a_3)}$} (01)

(0)  edge [right] node {$_{(a_1, b_2, a_3)}$} (021)
(0)  edge [right] node {$_{(a_1, a_2,b_3)}$} (022)
(0)  edge [right] node {$_{(a_1, b_2, b_3)}$} (023)

(01) edge [loop left] node {$_{(*,*,*)}$} (01)

(1) edge [loop left] node {$_{(*,*,*)}$} (1)
(2) edge [loop right] node {$_{(*,*,*)}$} (2)

(021) edge [left] node {$_{(a_p,*,*)}$} (1)
(021) edge [right] node {$_{(a_q, *,*)}$} (2)
 
(022) edge [left] node {$_{(a_p,*,*)}$} (1)
(022) edge [right] node {$_{(a_q, *,*)}$} (2)

(023) edge [left] node {$_{(a_p,*,*)}$} (1)
(023) edge [right] node {$_{(a_q, *,*)}$} (2); 

\end{tikzpicture}
\end{center}
\label{fig:exampleB2} 

This construction generalises to any concurrent game model to produce concurrent game models $\cgm$ having the following property:  for every state $s \in \cgm$ and action profiles $\actprof_1,\actprof_2$ available at $s$, if $\actprof_1 \neq \actprof_2$ then  $\out_\gmod(s,\actprof_1) \neq \out_\gmod(s,\actprof_2)$. We call  such models \defstyle{injective}. 
Every non-injective concurrent game model $\cgm$ can be transformed into an injective one $\iota(\cgm)$ by generalising the construction in the example above. That can be done in various ways, but a most economical one is to multiply every state $w$ into as many copies as the maximal number of incoming to $w$ transitional arrows leading from any other given state $u$ labelled with different action profiles applied at $u$. The labels of these copies, as well as all available action profiles at each of them and their outcomes,  are copied from those at $w$. Thereafter, all multiply-labelled transitions from any state $s$ to $w$ are respectively re-directed to the different copies of $w$. We leave out the routine technical details of this construction.  
We will call this construction \defstyle{state-copying and outcome-splitting}, abbreviated \defstyle{SCOS}.

 Note that plays and paths in any injective model are in a trivial 1-1 correspondence, hence play-based and path-based strategies in injective models coincide.  
It is then straightforward to show the following, by induction on the formulae in $\cga^+$.

\begin{proposition}
\label{prop:injective} 
In every injective concurrent game  model $\cgm$, $\models_{\mathsf{play}}$ \ and \ $ \models_{\mathsf{path}}$ coincide, i.e. for every state $s \in \cgm$ and a state formula 
$\varphi \in \cga^+$: \  $\cgm, s \models_{\mathsf{path}} \varphi$ iff  
$\cgm, s \models_{\mathsf{play}} \varphi$.   
\end{proposition}

 We further note that the resulting model from applying the SCOS construction is \cga-bisimilar to the original one, in terms of the notion of \cga-bisimulation defined in Section \ref{subsec:Bisimulations}. 
Consequently, due to the bisimulation invariance theorem \ref{thm:bisimulation invariance+} 
established there, the SCOS construction preserves, inter alia, truth of all $\cga^+$ formulae with respect to the play-based semantics, as follows. 

\begin{proposition}
\label{prop:SCOS} 
Given any concurrent game model $\cgm$, state $s \in \cgm$, and a copy $s^i$ of $s$ in the injective model $\iota(\cgm)$ obtained by applying SCOS to  $\cgm$, for every state formula $\varphi \in \cga^+$:  \ $\cgm, s \models_{\mathsf{play}} \varphi$ \ iff \   
$\iota(\cgm), s^i \models_{\mathsf{play}} \varphi$.   
\end{proposition}

Combining the two observations above, we obtain the following.  

\begin{corollary}
\label{cor:MCreduction} 
The model checking problem (MC) for formulae of $\cga^+$ with the play-based semantics is reducible to 
the model checking problem for $\cga^+$ with the path-based semantics, at the cost of at most quadratic blow-up of the size of the model. 
\end{corollary}

Observe also that  in any concurrent game model $\cgm$, 
every path-based strategy is also a play-based strategy (prescribing the same action on every two plays generating the same path). Therefore, for every formula $\cgoal{\gamma}\in \cga^+$, in  any concurrent game model $\cgm$ and a state $s \in \cgm$: \ 
if $\cgm, s \models_{\mathsf{path}}  \cgoal{\gamma}$ then 
$\cgm, s \models_{\mathsf{play}}  \cgoal{\gamma}$.

\medskip
We now show that the two semantics also differ with respect to the respective validities (hence, also 
with respect to the satisfiable formulae) of \cga. For that we will use the idea of Example \ref{exampleB}. Consider the goal assignment $\gamma$ defined there, as well as the goal assignment $\gamma'$ with support $\{\{1,2\}, \{1,3\}, \{1,2,3\}\}$, where: \\  
$\gamma'(\{1,2 \}) = \atlx  \cgoal{\{1\} \gass \atlg p}$, 
$\gamma'(\{1,3 \}) =  \atlx  \cgoal{\{1\} \gass \atlg q}$,  and 
$\gamma'(\{1,2,3 \}) =  \atlg (p \land q)$.
 
Then the following hold: 

\begin{enumerate}

\item $\models_{\mathsf{play}} \cgoal{\gamma'} \to \cgoal{\gamma}$. 

Indeed, in any given model, a strategy profile satisfying $\cgoal{\gamma}$ can be produced from a strategy profile satisfying $\cgoal{\gamma'}$ by amending the strategy for player $1$ in the former strategy profile with those strategies for $1$ claimed to exist in the respective successor outcomes from the joint strategies for $\{1,2 \}$ for $\{1,3 \}$. 
 These amendments are done as follows. 
 Consider any successor state $w$ of the current state $s$, obtained as the outcome from applying at $s$ an action profile $\actprof$ obtained, for instance, by extending the joint action of $\{1,2 \}$ prescribed at $s$ by the strategy profile $\strprof'$ witnessing the truth of $\cgoal{\gamma'}$ at $s$ with any action of 3 different from the one prescribed by its strategy in $\strprof'$. Then, the strategy for $1$ is re-defined on any play of the type $\pi s \actprof w \pi'$ to prescribe the action which a strategy for $1$ that is claimed by $\gamma'(\{1,2 \})$ to exist at $w$ prescribes at the play $w \pi'$ in order to ensure the truth of $\atlg p$ on any resulting play. The strategy for $1$ on all plays passing through $\actprof' w'$ when $\actprof'$ comes from a joint action of $\{1,3 \}$ at $s$ is re-defined likewise. Lastly, when all 3 players are following the strategy profile $\strprof'$ witnessing the truth of $\cgoal{\gamma'}$ at $s$, in the resulting play both $\atlg p$ and $\atlg q$ hold, so no strategy amendment is needed. 
  Since the strategies that we consider are play-based, the two cases of strategy amendments are independent from each other, as they apply to disjoint sets of plays, and are therefore unproblematic to combine (which is not the case for path-based strategies, as Example \ref{exampleB} shows).  
  It is now straightforward to show that the strategy profile $\strprof$, resulting from replacement of the strategy for 1 in $\strprof'$ with the amended strategy described above, witnesses the truth of $\cgoal{\gamma}$. 

\item $\not\models_{\mathsf{path}} \cgoal{\gamma'} \to \cgoal{\gamma}$. 

Indeed, that formula fails in the model displayed in Example \ref{exampleB}, because it is straightforward to show that $\gmod, s \models_{\mathsf{path}}  \cgoal{\gamma'}$. 
\end{enumerate}

\subsection{Standard translation of the path-based semantics of $\cga^+$ into
fragments of Strategy Logic and the complexity of model checking \cga.
}  
\label{subsec:ST2SL}
 
 We refer here to a standard version SL of Strategy logic as defined e.g. in \cite{fsttcs/MogaveroMV10}, 
\cite{tocl/MogaveroMPV14},  involving variables ranging over strategies that can be associated with any agents within the formulae by means of  strategy assignments, and quantification over such variables. 
Here we mostly follow the notation for strategy assignments from \cite{DBLP:conf/lics/MogaveroMS13}, and  \cite{DBLP:conf/aaai/AcarBM19}, which slightly differ from the one in \cite{DBLP:journals/mst/GardyBM20}.

First, we show that the logic $\cga^+$ \emph{with the path-based semantics} can be translated to SL 
in a way very similar to the standard translation of modal logic to first-order logic. The key clause is the translation of the operator $\brak{\gamma}$, defined as follows\footnote{This form of the translation was suggested by an anonymous reviewer.} (using the notation from
 \cite{DBLP:conf/aaai/AcarBM19}), 
for 
$\Agt = \{\aga_1,...,\aga_n\}$: 
\[\trn(\brak{\gamma}) =  
 \exists x_1 ... \exists x_n \forall y_1 ... \forall y_n
\bigwedge_{C \subseteq \Agt}
\flat_{C} \gamma(C)
\]
where the binding prefix $\flat_{C}$ assigns to each agent $\aga_i$ in $\Agt$ the value of the strategy variable
$x_i$ if $\aga_i \in C$, and the value of the strategy variable $y_i$, otherwise. 
To shorten the formula up to equivalence, the big conjunction above can be restricted to range only over those coalitions $C \subseteq \Agt$ for which $\gamma(C) \neq \X \top$. 
 
 Since the translation above places only conjunctions of temporal goals in the scope of all quantifiers over strategies, $\cga^+$ translates in the Conjunctive-Goal fragment SL[CG]. Using this, we obtain the following: 

\begin{proposition}
\label{prop:MC complexity} 
For each of the path-based semantics and the play-based semantics of $\cga^+$, the respective model checking problem is \textrm{PTIME-complete} in the size of the model and in \textrm{2ExpTime} in the size of the formula. 
\end{proposition}

\begin{proof} 
For MC in the path-based semantics, the claim follows from the translation $\trn$ above and the respective results in \cite[Theorem IV.2]{DBLP:conf/lics/MogaveroMS13} and in \cite{DBLP:journals/mst/GardyBM20}, for the upper bounds. The lower bound in terms of the size of the model follows from the \textrm{PTIME-complete} complexity of MC for ATL \cite{AHK-02}, which is embedded into \cga. The precise complexities of MC  for \cga and $\cga^+$ in the size of the formula are currently still open. 

For the MC in the play-based semantics the complexity bounds are obtained from those above, by using the reduction to MC in the path-based semantics provided by Corollary \ref{cor:MCreduction}. 
\end{proof}

Note further that, when $\brak{\gamma}$ is in the \defstyle{flat fragment $\fcga^+$ of $\cga^+$}, i.e. all goals $\gamma(C)$ are conjunctions of purely temporal goals (not containing nested $\gamma$ operators), the translation  formula above is in the flat sub-fragment FSL[CG] of SL[CG]  \cite{DBLP:conf/aaai/AcarBM19}. Consequently, that translation can also be used for applying algorithms and obtaining \textrm{PSpace} complexity bound for the satisfiability problem in $\fcga^+$ with the 
path-based semantics 
 by reduction to the satisfiability problem in FSL[CG], shown in  \cite{DBLP:conf/aaai/AcarBM19} to be decidable and \textrm{PSpace-complete} (notably, lower than the \textrm{2ExpTime} complexity of model checking for FSL[CG], as shown in  \cite{DBLP:conf/aaai/AcarBM19}).  
On the other hand, as we show in Section \ref{sec:FMP}, using embedding of $\cga^+$ to a $\mu$-calculus extension $\xlangcga_\mu$ of the nexttime fragment $\xlangcga$ of \cga, defined in Section \ref{subsec:mu-TLCGA}, 
the satisfiability problem in $\xlangcga_\mu$ is in \textrm{2ExpTime}, hence the satisfiability problem in the whole $\cga^+$, with its standard 
play-based semantics, is in \textrm{3ExpTime}.

\section{Some applications of \cga and $\cga^+$}
\label{sec:applications}

\subsection{Examples of expressing and reasoning about group objectives with \cga}
\label{subsec:expressing}

\subsubsection{Example 1: Password protected data sharing}
\label{example1}

This example is adapted from \cite{GorankoEnqvist18}, where it was adapted from \cite{parikh1985logic}. Consider the  following scenario involving two players, Alice  (denoted $A$) and Bob (denoted $B$). Each of them owns a server storing some data, the access  to which is protected by a password. Alice and Bob want to exchange passwords, but neither of them is sure whether to trust the other.  So the common goal of the two players is to cooperate and exchange passwords, but each player also has the private goal not to give away their password in case the other player turns out to be untrustworthy and not provide his/her password. When and how can the two players cooperate to exchange passwords? The answer depends on the kind of actions that Alice and Bob can perform while  attempting to achieve their common objective. However, we are more interested now in formalising the problem in \cga. 

Let us first try to express the common objective by a \cga formula. For that, we write $H_A$ for ``Alice has access to the data on Bob's server'' and $H_B$ for ``Bob has access to the data on Alice's server''. Then an obvious candidate for a formula expressing the common goal is the goal assignment formula 
\[\cgoal{\{A,B\}  \gass \atlf (H_A \wedge H_B)} \]
stating that Alice and Bob have a joint strategy to eventually reach their common objective. 
However, it is easy to see that this is not good enough. Indeed, while common desired eventual outcome is $H_A \wedge H_B$, but for $\mathsf{Alice}$ the worst possible outcome is $\neg H_A \wedge H_B$, whereas the worst possible outcome for Bob is $H_A \wedge \neg H_B$, and each of them would like to avoid their worst possible outcome to happen while trying to achieve the common goal. Thus, the common goal can be formulated better as ``\textit{eventually reach a state where both players can access each other's data and until then no player should be able to unilaterally access the other's data}", expressed by the following goal assignment formula: 
\[  \cgoal{ \{A,B\}\gass (H_A \ifff H_B) \until (H_A \wedge H_B ) } \]
The formula above is ok if both players follow a strategy profile that would realise that goal, but it does not express yet the stronger requirement that even if one of them deviates from that strategy profile the other should still be able to protect her/his interests while still following her/his strategy. For that, we need to enrich the goal assignment above with individual goals: 
\[ \cgoal{ \{A,B\}\gass (H_A \ifff H_B) \atlu (H_A \wedge H_B ); \ 
 A \gass \atlg(H_B \rightarrow H_A ); \ B \gass \atlg(H_A \rightarrow H_B) } \]

Note that the common goal can now be simplified to the original one, to produce an equivalent to the above formula: 
\[  \cgoal{ \{A,B\} \gass \atlf (H_A \wedge H_B ); \ 
 A \gass \atlg(H_B \rightarrow H_A ); \ B \gass \atlg(H_A \rightarrow H_B) } \]

\subsubsection{Example 2: Sheep and wolves: a fragile alliance}
\label{example2}

This example is a remake with a twist of a well known children's puzzle. 
A group of 3 wolves and 3 sheep is on the one side of a river and they want to cross the river by boat. There is only one boat that can take 2 animals at a time, but there is no boatman, so one animal has to take the boat back every time, until they all cross the river. The main problem, of course, is that if the wolves ever outnumber the sheep on either side of the river, or on the boat, then the sheep in minority will be promptly eaten up by the wolves. The question is whether, -- and if so, how -- all animals can cross the river without any sheep being eaten. 
(Spoiler alert: the answer will be gradually revealed further, so the reader may wish to pause here and think on the puzzle before reading further.)  

Let us formalise the problem in \cga. First, some notation. Let $\mathsf{Sheep}$ denote the set of all sheep, $\mathsf{Wolves}$ denote the set of all wolves, 
$\mathbf{c}$ denote the proposition ``\textit{all animals have crossed the river}'' and  
$\mathbf{e}$ denote the proposition ``\textit{a sheep gets eaten}''. Then the problem seems to be expressed succinctly as the question whether the following formula is true: 
 \[
 \brak{\mathsf{Sheep} \cup \mathsf{Wolves}\gass (\neg \mathbf{e}) \atlu \mathbf{c}}\ 
 \]
As in the previous example, this formula is too weak to express the important subtlety that, even if such strategy exists, nothing guarantees that the wolves will not decide to deviate from it and have a gourmet feast with a sheep before (or after)  crossing the river. Thus, we need to add an extra goal for all sheep, protecting their interest to stay alive: 
 \[
 \brak{\mathsf{Sheep} \cup \mathsf{Wolves}\gass (\neg \mathbf{e}) \atlu \mathbf{c}; 
 \  \mathsf{Sheep}\gass \always \neg \mathbf{e}}
 \]
 Now, the common goal can clearly be simplified, while preserving the formula up to equivalence: 
 \[
 \brak{\mathsf{Sheep} \cup \mathsf{Wolves}\gass \atlf \mathbf{c}; \ \mathsf{Sheep}\gass\always \neg \mathbf{e}}
 \]
  
   Let us now try to model it as a concurrent game model.  We can assume that the river crossing happens instantaneously, so each state of the game is described uniquely (up to re-shuffling of the sheep and of the wolves, which can be considered identical) by the numbers of sheep and wolves on each side of the river, plus the position of the boat (on one or the other side of the river). At each river crossing round, each of the animals has two possible actions: `stay' or `go on the boat and cross the river'. The respective transitions are then readily defined, by ensuring that only legitimate transitions can occur, so e.g., if more than two animals decide to jump on the boat at the same time, the state does not change (the transition is a loop). The states satisfying $\mathbf{e}$ are precisely those where there are more wolves than sheep on any one side of the river, whereas only one state satisfies $\mathbf{c}$, where all animals have crossed the river.  

And, now, the question: is there a strategy profile satisfying the goal assignment above? The answer, perhaps surprisingly, depends on the specific design of the `river crossing game'. If it presumes that all animals act simultaneously, then it is easy to see that any joint strategy realising the common goal can be abused by the wolves deviating from it and eating some of the strategy-abiding sheep. For example consider the joint strategy resulting in the play shown below:
\[
\begin{array}{l  | c  c c | r}
S\; S\;S \; W\;W\;W & B & & &  \\
S\; S \; W\;W &  &  &  B & S\; W \\
S\; S \; S\; W\;W & B    & & &  W \\
S\; S \; S &  &  &  B & W\; W \; W \\
S\; S\; S\; W &  B  & &   & W\;W \\
S\;W &  & &   B & S \; S \; W\; W \\
S\; S \; W \; W &  B  & &   & S\; W \\
 W\;W &  & &   B & S\; S \; S\; W \\
W\; W\;W &  B  & &   & S\; S \; S \\
 W &  & &   B & S\; S \; S \; W\; W \\
 W\;W &  B  & &  & S\; S \;S\; W \\
 &  & &    B & S\; S\;S \; W\;W\;W
\end{array}
\] 
At the very first round of this play, a sheep and a wolf cross the river together. If the wolf deviates from this action and stays instead, then two sheep are left to fend against three wolves on one side of the river. We leave it to the reader to convince themselves that any joint strategy that achieves the common goal must encounter a similar situation. 

So, the answer to our question in this case is `No'. However, to level the playing field, the game can be modified so that at every state \emph{first all wolves choose how to act and then all sheep choose how to act}, i.e. formally, every round gets split into two sub-rounds with intermediate states (thus, making it a partly turn-based game). The effect of this change is that now a strategy profile satisfying the goal assignment above could be designed in such a way that the joint strategy of the sheep could involve a suitable joint counter-action to any possible deviation of the wolves that would jeopardise a sheep. Indeed, the joint strategy shown previously can now easily be modified to make it sheep-friendly.

\subsection{Some applications of \cga and $\cga^+$
to non-cooperative and cooperative game theory}
\label{subsec:applications2GT}

\subsubsection{Expressing Nash equilibria}
\label{subsec:equilibria}  
The fundamental game-theoretic concept of Nash equilibrium can be applied in the concurrent games that we consider, where, given a goal assignment $\gamma$, the payoff from each play for every player is binary:  1, if that player's goal defined by $\gamma$ is satisfied on that play (i.e., the player is a `winner' in the play), and 0 otherwise  (i.e., the player is a `loser' in the play). 
However, this notion makes little sense in such qualitative setting, because every strategy profile where no `loser' can deviate unilaterally to satisfy her objective is a weak Nash equilibrium. That gives no individually rational reasons for the losers to adhere to that strategy profile, because a deviation cannot be penalised any further by making their payoff even worse than it already is. Thus, we are rather sceptical about the use of  Nash equilibria in such games with qualitative objectives as those considered here, and we argue further for an alternative solution concept in games with such objectives. 
Accordingly, even though the language of \cga is not designed with the explicit purpose of expressing equilibria by means of formulae, that can be done in terms of characterising the equilibria profiles as those witnessing suitable goal assignments defined by means of the original goal assignment $\gamma$. In this way, \cga also enables expressing \emph{existence of equilibria} by means of formulae.  

First, we add some terminology and notation. Let us fix a concurrent game model $\cgm$ with initial state $w$. 
Then, any pair $(\strprof,\gamma)$ of a strategy profile $\strprof$ and a goal assignment $\gamma$, the pair $(\strprof,\gamma)$ determines a partition of the family of possible coalitions of agents $\powerset{\Agt}$ in two disjoint subsets: the set $\Wn(\strprof,\gamma)$ of  \defstyle{winning coalitions}, whose goals assigned by $\gamma$ are satisfied in the play $\play(w,\strprof)$ in $\cgm$ starting from $w$ and induced by $\strprof$, and  the set $\Ls(\strprof,\gamma)$ of \defstyle{losing coalitions}, 
 whose goals assigned by $\gamma$ are not satisfied in $\play(w,\strprof)$. 
 In particular, the pair $(\strprof,\gamma)$ determines a partition of $\Agt$ in two disjoint subsets, respectively the set $\iWn(\strprof,\gamma)$ of (individual) winners and the set $\iLs(\strprof,\gamma)$ of (individual) losers from $\strprof$ with respect to $\gamma$. 

We now illustrate the idea of expressing equilibria in \cga in the case of a nexttime goal assignment $\alpha$ with an example for 2 players, $A$ and $B$, with respective individual nexttime goals  $\alpha_A$ and $\alpha_B$. First, we can express an equilibrium satisfying any fixed combination of individual goals.  In the case when both goals are satisfied by the equilibrium profile, no deviations can possibly improve any player's payoff, so 
$\strprof$ is such equilibrium profile iff $\strprof$ witnesses the goal assignment 
$\{A,B\} \gass \atlx (\alpha_A \wedge \alpha_B)$, hence existence of such equilibrium profile is simply expressed by the formula
\[\brak{\{A,B\} \gass \atlx (\alpha_A \wedge \alpha_B)} \]
 A more interesting case is when the equilibrium profile $\strprof$ satisfies only one goal, say $\alpha_A$.
 This is the case precisely when $\strprof$ witnesses the goal assignment 
 $\{A,B\} \gass \atlx (\alpha_A \wedge \neg \alpha_B); \ A \gass \atlx \neg \alpha_B$. The goal $\atlx \neg \alpha_B$ of $A$ encodes the claim that if $A$ follows the equilibrium strategy then $B$ cannot deviate to satisfy its goal, thus $B$ has no unilateral beneficial deviation from  the equilibrium strategy profile. Thus, the following formula expresses existence of such equilibrium: 
\[\brak{\{A,B\} \gass \atlx (\alpha_A \wedge \neg \alpha_B); \ A \gass \atlx \neg \alpha_B} \]
Likewise, the following formula expresses existence of an equilibrium not satisfying anyone's goal: 
\[\brak{\{A,B\} \gass \atlx (\neg \alpha_A \wedge \neg \alpha_B); \ {A} \gass \atlx \neg \alpha_B; \ {B} \gass \atlx \neg \alpha_A} \]
Thus, the language of \cga allows for expressing more refined descriptions of equilibria and 
the disjunction of all such formulae expresses the existence of any equilibrium with respect to $\alpha$.  

In the more general case of any set of agents $\Agt$ and pair $(\strprof,\alpha)$ of a strategy profile $\strprof$ and a nexttime goal assignment $\alpha$, the strategy profile $\strprof$ is a Nash equilibrium with respect to $\alpha$ iff $\strprof$ witnesses the goal assignment  $\alpha^{0}_{\strprof}$ defined as follows: 
  $\alpha^{0}_{\strprof}(\Agt) = \atlx \bigwedge_{\aga \in \iWn(\strprof,\alpha)} \alpha(\aga)$, \  
$\alpha^{0}_{\strprof}(\Agt \setminus \{\aga\}) = \lnot \alpha(\aga)$ for each 
$\aga \in \iLs(\strprof,\alpha)$, \ and $\alpha^{0}_{\strprof}(C) = \top$ in all other cases.   

Then, the disjunction of all such formulae over all subsets $\iWn(\strprof,\alpha)$ of 
$\Agt$ expresses the existence of any equilibrium, though the size of that formula grows exponentially in the number of agents. 

Lastly, for the most general case of any goal assignment $\gamma$, the idea is the same, but using $\cga^+$ 
(which allows conjunctions of path formulae of $\cga$ as goals) in order to define the goal 
  $\gamma^{0}_{\strprof}(\Agt)$. Note that in the case when all individual goals are of the type $\atlg \psi$, the language of \cga again suffices, because $\atlg$ distributes over conjunctions.   
  
Thus, computing the sets $\iLs(\strprof,\gamma)$ and answering the question of whether the game has any 
Nash equilibria can be solved by reducing to model checking in $\cga^+$; in the special case when all goals are nexttime formulae, or all are $\atlg$-type formulae, model checking in \cga suffices.

\subsubsection{Co-equilibria} 
\label{subsec:co-equilibria}  
Here we define and promote a new, alternative solution concept, that naturally arises in our framework, vis that of a `\emph{co-equilibrium}', which is also one of the main motivations for the introduction of the operator $\brak{\cdot}$. Recall that an equilibrium strategy profile means that no player can deviate individually to improve their performance, and that concept makes very good sense when players pursue \emph{quantitative} individual objectives which are usually achieved to some degree, leaving room both for possible optimisation and for punishment by the other players when deviating, hence can serve as an effective deterrent from deviation. As we argued above, it does not make very good sense when the individual objectives are \emph{qualitative}, i.e. \emph{win} or \emph{lose}, as losing is the worst possible outcome for the player, hence there  can be no deterrent from deviation from a strategy profile where that player is losing anyway. 
Furthermore, players usually participate simultaneously in several coalitions with mutually consistent, yet different objectives.  Assuming that they are first of all individually rational and only then collectively rational, players try to adjust their strategic behaviour so as to serve the collective objectives as much as possible while first protecting and ensuring their individual interests. In particular, all players usually have one common, societal objective -- say to keep the entire system live and safe -- so they enter into a global `social contract' over that common objective, but only on condition that pursuing it would not compromise the achievement of their individual objectives.  
These aspects of strategic interaction of individually rational agents serve as our motivation to define the 
somewhat dual to equilibrium notion of \emph{co-equilibrium} in the context of collective and individual qualitative objectives, as a strategy profile that not only ensures satisfaction of the collective objective (the `social contract') if all players follow it, but moreover also guarantees to every player who adheres to it that even if all other players deviate, that would not affect the satisfaction of his/her individual objective\footnote{The notion of co-equilibrium, when applied to possibly quantitative objectives, is essentially equivalent to the special case of  `$t$-immune strategy profile',  introduced in \cite{DBLP:conf/podc/AbrahamDGH06}, when $t = n-1$, where $n$ is the number of players.}.
Thus, a co-equilibrium is a strongly stable solution concept that, we argue, makes better sense than a Nash equilibrium in games with  \emph{qualitative} individual objectives and existence of a co-equilibrium is an important criterion for stability of a society of strategically interacting individually rational agents.  
Formally, a strategy profile $\strprof$ is a \defstyle{ co-equilibrium} with respect to a goal assignment $\gamma$ iff $\strprof$ witnesses the goal assignment  $\gamma^{*}$ which is the restriction of $\gamma$ with support consisting of the grand coalition $\Agt$ and all singleton sets of agents. Respectively, existence of a co-equilibrium with respect to $\gamma$ can be expressed in \cga simply as $\brak{\gamma^{*}}$.

\subsubsection{Expressing other stable outcomes}
\label{subsec:stable-outcomes}  

The notion of co-equilibrium is one of a family of stable strategy profiles that can be defined by varying 
the notion of stability with respect to non-existence of various \emph{beneficial deviations} of players or coalitions. 
In the case of co-equilibrium, no player or coalition is interested to deviate simply because they are all satisfied by the co-equilibrium strategy profile, so there are no beneficial deviations. This, of course, is the ideal case, which is often not possible in reality, so we will look at some natural relaxations of the notion of stable outcome. 

 Let us fix a concurrent game model $\cgm$ with initial state $w$, a coalitional goal assignment $\gamma$, and a strategy profile $\strprof$. In Section \ref{subsec:equilibria}, we defined winning players and coalitions in the outcome play $\play(w,\strprof)$ with respect to $\gamma$, the sets 
 $\Wn(\strprof,\gamma), \Ls(\strprof,\gamma)$ and respectively 
 $\iWn(\strprof,\gamma), \iLs(\strprof,\gamma)$. In addition, we say that a coalition $C$ is 
 \defstyle{individually winning in $\play(w,\strprof)$ with respect to $\gamma$} if 
$C \subseteq  \iWn(\strprof,\gamma)$; 
respectively, 
$C$ is \defstyle{individually losing in $\play(w,\strprof)$ with respect to $\gamma$} if 
$C \subseteq  \iLs(\strprof,\gamma)$.

Now, with reference to the fixed  coalitional goal assignment $\gamma$, we say that a strategy profile 
$\strprof$ is: 

\begin{enumerate}
\item \defstyle{individually stable at $w$}, if no losing player in $\play(w,\strprof)$ can deviate from  
$\strprof$ to become a winning player in the resulting strategy profile $\strprof'$. 

This is precisely equivalent to the standard notion of Nash equilibrium.

\item \defstyle{strongly individually stable at $w$}, if no group of losing players in $\play(w,\strprof)$ can collectively deviate from $\strprof$ to become a group of winning players in the resulting strategy profile $\strprof'$. 

This corresponds to Aumann's notion of \emph{strong equilibrium} \cite{aumann59} and  is similarly expressible in \cga (for nexttime goal assignments) or $\cga^+$ (for any goal assignments):    

 $\strprof$ is strongly individually stable with respect to $\gamma$ iff $\strprof$ witnesses the goal assignment  $\gamma^{s}_{\strprof}$ defined as follows: 
 
  $\gamma^{s}_{\strprof}(\Agt) = \bigwedge_{\aga \in \iWn(\strprof,\gamma)} \gamma(\aga)$, \\ 
  $\gamma^{s}_{\strprof}(C) = \lnot \bigwedge_{\aga \in \overline{C}} \gamma(\aga)$
 for each $C$ such that 
$\overline{C} \subseteq \iLs(\strprof,\gamma)$,  \\ 
and $\gamma^{s}_{\strprof}(C) = \top$ in all other cases.

\item \defstyle{coalitionally stable at $w$}, if no losing coalition in $\play(w,\strprof)$ can collectively deviate from $\strprof$ to become a winning coalition in the resulting strategy profile $\strprof'$. 

This is similarly expressible in \cga (for nexttime goal assignments) or $\cga^+$ (for any goal assignments):  

 $\strprof$ is coalitionally stable  with respect to $\gamma$ iff $\strprof$ witnesses the goal assignment  $\gamma_{\strprof}$ defined as follows: 
 
  $\gamma_{\strprof}(\Agt) = \bigwedge_{C \in \Wn(\strprof,\gamma)} \gamma(C)$, \\ 
  $\gamma_{\strprof}(C) = 
 \lnot \gamma(\overline{C})$ for each $C$ such that 
$\overline{C} \in \Ls(\strprof,\gamma)$,  \\ 
and $\gamma_{\strprof}(C) = \top$ in all other cases.   

\end{enumerate}

\subsubsection{Cooperative games and \cga}
\label{subsec:CoopGT}
The technical ideas outlined above can also be applied to express in \cga (or $\cga^+$) some key notions of the theory of \emph{cooperative games (with transferable utility)} \cite{Ramamurthy90},  
\cite{BranzeiDimitrovTijs}, \cite{DBLP:series/synthesis/2011Chalkiadakis}, including  \emph{beneficial deviations}, \emph{stable outcomes}, and \emph{core} of a cooperative game. We leave the exploration of  these to a further work, and here we only outline one example showing how that can be done for the kind of \emph{cooperative concurrent games} studied in \cite{DBLP:conf/atal/GutierrezKW19}. 
These games are played in concurrent game structures, where each player has a goal expressed by a set ot plays  starting from some fixed initial state, regarded as the "winning plays" for that player. In particular, cooperative concurrent games with goals expressible in the linear time logic \LTL (cf e.g. \cite{TLCSbook}) are studied in \cite{DBLP:conf/atal/GutierrezKW19}, for which it is shown that a suitably defined notion of a core of such a game can be logically characterised using the logic \ATLs 
and the computational complexities of certain decision problems associated with that  core have been established. 

 Let us fix a state $w$ in a concurrent game model $\cgm$ and consider the cooperative concurrent game $G = G(\cgm,w)$ generated at $w$ in $\cgm$ and a strategy profile $\strprof$ in $G$. 
We note that the terminology from \cite{DBLP:conf/atal/GutierrezKW19}, differs somewhat from ours, as they call a winning (resp. losing) coalition\footnote{Only individual, but no coalitional goals are considered in \cite{DBLP:conf/atal/GutierrezKW19}.} what we call in Section \ref{subsec:stable-outcomes} \emph{individually winning} (resp. \emph{individually losing}) coalition.  For consistency, we will adhere to our terminology. 

To make the parameter $w$ referring to the current initial state explicit, we will denote the set of all losing players in $\play(w,\strprof)$ with respect to $\gamma$ by $\iLs(w,\strprof,\gamma)$.  
Thus,  
$C \subseteq \iLs(w,\strprof,\gamma)$, i.e. the coalition $C$ is individually losing in $\play(w,\strprof)$,  
iff $\strprof$ witnesses the  $\cga^+$ goal assignment  $\Agt  \gass \bigwedge_{\agi \in C} \lnot \gamma(\agi)$. In the case when all $\gamma(\agi)$ are nexttime formulae $\atlx \psi(\agi)$, that goal assignment  can be replaced by the $\cga$ goal assignment $\Agt  \gass \atlx \bigwedge_{\agi \in C} \lnot \psi(\agi)$. 

Now, if $C \subseteq \iLs(w,\strprof,\gamma)$, 
then a \emph{individually beneficial deviation}\footnote{Note that this notion, defined in  \cite{DBLP:conf/atal/GutierrezKW19} as ``beneficial deviation'', does not actually depend on $\strprof$ in any other way but that $C$ is losing in $\play(w,\strprof)$.} for $C$ is any joint strategy $\strprof _C$ of $C$ that guarantees for $C$ to be individually winning on any outcome play from $w$ induced by 
$\strprof _C$.  
The \defstyle{core of the game $G$}, denoted $\mathit{core}(G)$, is defined as the set of \defstyle{stable strategy profiles} 
in $G$, viz. those  that admit no beneficial deviations (by any losing coalition). 
Thus, $C$ has a beneficial deviation at the state $w$ iff $\brak{\gamma_{C}}$ is true at $w$, where $\gamma_{C}$ is the goal assignment $C \gass \bigwedge_{\agi \in C} \gamma(\agi)$.
Therefore, existence of a coalition that has a beneficial deviation with respect to $\gamma$ (and $\strprof$) can be expressed as the disjunction of all formulae $\brak{\gamma_{C}}$, over all individually losing coalitions $C \subseteq  \iLs(w,\strprof,\gamma)$. 
Then, $\strprof$ is in $\mathit{core}(G)$ iff it does not satisfy that formula at $w$:  
\[ \strprof \in \mathit{core}(G) \qquad \mbox{ iff }  \qquad 
\cgm, w \models \bigwedge_{C \subseteq  \iLs(w,\strprof,\gamma)} \lnot \brak{\gamma_{C}}\]
 
Thus, like in the case of Nash equilibria and the other stable outcomes defined in Section \ref{subsec:stable-outcomes}, computing the sets $\iLs(w,\strprof,\gamma)$ and answering the question of whether the core of such a game is non-empty can be solved by reducing to model checking in $\cga^+$; in the special case when all goals are nexttime formulae, or all are $\atlg$-type formulae, model checking in \cga suffices.

\section{Fixpoint characterizations of temporal formulae in \cga}
\label{sec:fixpoints}

In this section we will show how to embed the logic $\cga$ into a suitable fixpoint logic. In fact, the embedding can be extended likewise to the logic $\cga^+$, and we sketch along the way how this is done. The results in this section will be stated and proved for \cga, but their extensions to $\cga^+$ are quite routine.

\subsection{Types of goal assignments} 
\label{subsec:TypesAssignments}

\begin{definition}
A goal assignment $\gamma$ supported by a family of coalitions $\fac$ 
will be called \defstyle{long-term temporal} 
  if $\gamma$ maps every coalition in $\fac$ either to a $\until$-formula or a $\always$-formula, that is, if $\gamma[\fac]  \subseteq \lfor$, where 
  $\gamma[\fac] = \{ \gamma(C) \mid C \in \fac \}$. 
 
   A goal assignment is called  \defstyle{local}, or \defstyle{nexttime}, if  $\gamma$ maps every coalition in $\fac$ to a $\nexttime$-formula, i.e., $\gamma[\fac] \subseteq \xfor$.  
  
  A formula $\phi$ is said to be in \defstyle{normal form} if, for every subformula of the form $\brak{\gamma}$, the goal assignment $\gamma$ is either a nexttime or a long-term temporal goal assignment. 
\end{definition}

To extend this definition to $\cga^+$, we say that a goal assignment $\gamma$ in the extended language is long-term temporal if for each coalition $C$, each conjunct of $\gamma(C)$ is either a $\until$-formula or a $\always$-formula. Clearly, this reduces to the previous definition for the special case of goal assignments in $\cga$.

\begin{definition}
Let $\gamma$ be a long-term temporal goal assignment supported by the family $\fac$. We say that $\gamma$ is:
\begin{itemize}
\item of \defstyle{type $\until$} if $\gamma$ maps at least one element of $\fac$ to an 
$\until$-formula,

\item of \defstyle{type $\always$ } 
 if $\gamma$ maps every element of $\fac$ to an $\always$-formula.
\end{itemize} 

We denote the sets of goal assignments of type $\until$ and type $\always$  respectively by $\typeone$ and $\typetwo$.  
\end{definition}

Again, this can be extended to $\cga^+$: we say that a long-term temporal goal assignment $\gamma$ is of type $\until$ if there is some coalition $C$ in the support of $\gamma$, such that at least one conjunct of $\gamma(C)$ is an $\until$-formula. Otherwise, $\gamma$ is of type $\always$.

\subsection{The fixpoint property of goal assignments}

\begin{definition}
Given a family of coalitions $\fac$ and a goal assignment 
$\gamma$ supported by $\fac$, we write $\gamma\ugpart$ for the restriction of $\gamma$ to the family $\fac\ugpart = \{C \in \fac \mid \gamma(C) \in \lfor\}$. Similarly we write $\gamma\xpart$ for the restriction of $\gamma$ to the family $\fac\xpart \subseteq \fac$ defined as $\{C \in \fac \mid \gamma(C) \in \xfor\}$.
\end{definition}

To extend this notion to $\cga^+$, we define $\gamma\ugpart$ ($\gamma\xpart$) by setting, for each coalition $C$, $\gamma\ugpart(C)$ to be the conjunction of all formulas $\alpha \until \beta$ or $\always \chi$ that appear as conjuncts of $\gamma(C)$, provided that $\gamma(C)$ has at least one such conjunct, and $\gamma\ugpart(C) = \X \top$ otherwise. The goal assignment $\gamma\xpart$ is defined similarly. 

\begin{definition} 
Given a family of coalitions $\fac$ and a goal assignment 
$\gamma$ supported by $\fac$, the 
 \defstyle{nexttime-extension of $\gamma$} is the goal assignment 
$\diffof{\gamma}$
 defined as follows. First, we define  
 $\sup{\diffof{\gamma}} := 
 \big\{\bigcup \fac' \mid \emptyset \neq \fac' \subseteq \fac\big\}$, 
Then, for each $C \in \sup{\diffof{\gamma}}$ we define 
\[
\diffof{\gamma}(C) := \X \Big(\bigwedge \big\{\varphi \mid \mbox{there exists } C' \in \fac, C' \subseteq C  \mbox{ such that }  \gamma(C') = \X\varphi \big\} 
\wedge \brak{(\gamma\vert_C)\ugpart} \Big),
\]
where as a convention we remove from this formula any conjuncts that reduce to $\top$, which can appear as the result of a conjunction of the empty set (the left conjunct reduces to $\top$) or as $\brak{\gamma}$ where $\gamma$ is the empty goal assignment (the right conjunct reduces to $\top$). For all coalitions that are not in $\sup{\diffof{\gamma}}$, $\diffof{\gamma}$ assigns the trivial goal. Given any formula $\phi$, we will sometimes abbreviate the formula $\diffof{\gamma}[\bigcup \fac  \gass  \nexttime \phi]$ by $\gammaof{\phi}$. 
\end{definition}

The definition above may look a bit opaque, but what it does is quite simple. Intuitively, the goal assignment $\diffof{\gamma}$ describes the conditions that each coalition must ensure to hold for the \emph{next state}, with respect to a given strategy profile $\strprof$, in order for that strategy profile to fulfill the goal assignment $\gamma$ at the \emph{current state}. To make things more concrete, let us consider two examples.
\begin{example}
If $\brak{\gamma}$ is $\brak{C \gass \varphi \until \psi}$, then $\brak{\diffof{\gamma}} = \brak{C \gass \X\brak{C \gass \varphi \until \psi}} = \brak{C \gass \X\brak{\gamma}}$. 
So in this special case the nexttime-extension simply pushes the eventuality $\varphi \until \psi$ one step into the future, so to speak. Similarly, if $\brak{\gamma}$ is $\brak{C \gass \always \varphi}$, then $\brak{\diffof{\gamma}} = \brak{C \gass \X\brak{C \gass \always \varphi}} = \brak{C \gass \X\brak{\gamma}}$. 
\end{example}
\begin{example} Consider the example of a goal assignment $\gamma$ supported by $\fac = \{\{a,b\},\{c\},\{b,c\}\}$ and defined by the assignment:
$$\{a,b\}\gass p\until q, \;\; \{c\} \gass \always r, \;\; \{b,c\}\gass\X s$$
The support of $\diffof{\gamma}$ will be $\fac \cup \{\{a,b,c\}\}$. The action of $\diffof{\gamma}$ is shown below:
\begin{itemize}
\item  $\{a,b\} \gass \nexttime \brak{\{a,b\} \gass p\until q}$
\item $\{c\}  \gass  \nexttime \brak{\{c\} \gass  \always r}$
\item $\{b,c\}  \gass  \nexttime (s \wedge \brak{\{c\} \gass \always r})$
\item $\{a,b,c\}  \gass  \nexttime (s \wedge \brak{\{a,b\} \gass p\until q, \{c\} \gass  \always r})$
\end{itemize} 
\end{example}
The procedure for computing $\diffof{\gamma}$ goes, informally, as follows: for each $\subseteq$-downset\footnote{I.e. set $\mathcal{D}$ such that for all $Z \in \mathcal{D}$ and all $Z' \in \fac$, if $Z' \subseteq Z$ then $Z' \in \mathcal{D}$.} $\mathcal{D}$ of coalitions from $\fac$, we collect all the formulas $\varphi$ for which some coalition in $\mathcal{D}$ is mapped to the goal $\X\varphi$ into a conjunction, add a conjunct collecting all the longterm goals for coalitions in $\mathcal{D}$ into a single goal assigment, and finally put the resulting conjunction in the scope of an $\X$-operator and assign this goal to the union of $\mathcal{D}$.

We are now ready to define one of the key concepts of the paper.
\begin{definition}
Let $\gamma$ be a goal assignment, supported by $\fac $. Then we define the following formula:  
\[\unf{\gamma} := 
\bigvee \mathsf{Finish}(\gamma) \vee \bigg(\bigwedge \ugam \wedge \bigwedge \agam \wedge \brak{\diffof{\gamma}}\bigg),
\]
where:
\begin{itemize}
\item $\mathsf{Finish}(\gamma) : = \big\{\beta \wedge \brak{\gamma \setminus C} \mid 
\gamma(C) = 
\alpha \until \beta\big\}$

\item $\ugam :=  \big\{\alpha \mid  
\gamma(C) = \alpha \until \beta, 
\mbox{ for some } C,\beta \big\}$

\item $\agam : = \big\{\chi \mid  
\gamma(C) = 
\always \chi,  
\mbox{ for some } C \big\}$
\end{itemize}
As before, by convention we remove  from this formula all conjuncts that reduce to $\top$ and all disjuncts that reduce to $\bot$. 
So, for example, if the set $\mathsf{Finish}(\gamma) = \emptyset$, and hence also $\ugam = \emptyset$, then the formula $\unf{\gamma}$ reduces to:
$$\bigwedge \agam \wedge \brak{\diffof{\gamma}}.$$  
We call $\unf{\gamma}$ \defstyle{the unfolding formula} of $\gamma$. 
\end{definition}
The formula $\unf{\gamma}$ may require some additional explanation. We shall see that $\unf{\gamma}$ is in fact equivalent to $\brak{\gamma}$, and the formula can be seen as an analysis of the different possibilities for \emph{how} a given strategy profile may fulfil the goal assignment $\gamma$. The first disjunction $\mathsf{Finish}(\gamma)$ describes those cases  where one of the coalitions $C$ has an eventuality $\alpha \until \beta$ as its goal formula; if the formula $\beta$ happens to be true at the current state then the coalition $C$ trivially attains its goal regardless of its actions. So, for a strategy profile to fulfil the goal assignment $\gamma$, it suffices that it fulfils the goals of all coalitions besides $C$, i.e. the goal assignment $\gamma \setminus C$. In the remaining case, the conditions that a strategy profile must satisfy in order to fulfil the goal assignment $\gamma$ are divided into two parts: ``local'' conditions that must be true at the \emph{current state}, which are outside the agents' control, and those conditions that each coalition must ensure for the \emph{next} state. The latter conditions are described by the formula $\brak{\diffof{\gamma}}$, as explained earlier. The local conditions are derived from the interpretation of temporal (path) formulas: if the goal of coalition $C$ is an eventuality $\alpha \until \beta$, and this goal is not trivially fulfilled because $\beta$ happens to be true, then the goal can only be attained if $\alpha$ is true at the current state. Similarly, a goal $\always \chi$ can only be fulfilled if $\chi$ is currently true. These conditions are captured by the conjuncts $\ugam$ and $\agam$, respectively.

\begin{definition}
Let $\gamma$ be a long-term temporal goal assignment.  Then we define the  \defstyle{induction formula for $\gamma$ on $\phi$} as follows 

\[
\indf{\gamma}{\phi} :=
\bigvee \mathsf{Finish}(\gamma) \vee \Big(\bigwedge \ugam \wedge \bigwedge \agam \wedge 
\brak{\gammaof{\phi}}\Big), 
\]
after removing redundant conjuncts and disjuncts, as before. So, this formula is like 
$\unf{\gamma}$, except that the largest coalition in the support of $\diffof{\gamma}$ will be mapped to $\nexttime\phi$. 
\end{definition}

\begin{proposition}
\label{prop:unfold} 
For every long-term temporal goal assignment $\gamma$ we have: 
$$\unf{\gamma} = \indf{\gamma}{\brak{\gamma}}.$$ 
\end{proposition}
\begin{proof}
If $\gamma$ is long-term temporal and supported by $\fac$ then we get: 
\[(\gamma\ugpart)\vert_{\bigcup \fac} = \gamma.\]
Thus, since there are no nexttime formulas to consider, $\diffof{\gamma}$ will map $\bigcup \fac$ to $\nexttime\brak{\gamma}$, hence 
$\diffof{\gamma}[\bigcup \fac  \gass  \nexttime\brak{\gamma}] = \diffof{\gamma}$.  
\end{proof}

It is not hard to see that, if $\gamma$ is a nexttime goal assignment, then $\unf{\gamma} \equiv \brak{\gamma}$.  For example, suppose $\gamma$ is supported by $\{\{a\},\{b\}\}$ and maps $\{a\}$ to $\nexttime p$ and $\{b\}$ to 
$\nexttime q$. Then $\unf{\gamma}$ is equalt to $\brak{\diffof{\gamma}}$, which is the following formula:
$$\brak{\{a\} \gass \nexttime p, \{b\}  \gass  \nexttime q, \{a,b\} \gass  \nexttime (p \wedge q)}$$
which is clearly equivalent to 
$\brak{\gamma} = \brak{\{a\} \gass \nexttime p, \{b\}  \gass  \nexttime q\}}$. 
In fact, the equivalence always holds; this is by design, and will play a key role for our axiomatization.

\begin{theorem}[Fixpoint property]
\label{p:fixpoint-property}
For any goal assignment $\gamma$:
$$ \brak{\gamma}\equiv \unf{\gamma},$$
and hence for any long-term temporal goal assignment $\gamma$:
$$\brak{\gamma} \equiv \indf{\gamma}{\brak{\gamma}}.$$
\end{theorem}

\append{
\begin{proof}
We prove each implication separately.

\paragraph{Left to right:} suppose that $\gmod,s \sat \brak{\gamma}$, where $\gmod = (\states,\Act,\gmap,\out,V)$, and let $\strprof$ be some profile witnessing $\gamma$ at $s$. Assuming that $\gmod, s \nsat \bigvee \fgam$, we show that:
 $$\gmod, s \sat \bigwedge \ugam \wedge \bigwedge \agam \wedge \brak{\diffof{\gamma}}.$$  
We treat these conjuncts separately. First, note that if $\ugam = \emptyset$ or $\agam = \emptyset$ then these conjunctions reduce to $\top$ and hence are trivially satisfied. 

Suppose $\alpha \in \ugam$. Then there is some coalition $C$ and some $\beta$ for which 
$\gamma(C) = \alpha \until \beta$. The set $\paths(s,\strprof,{C})$ is always non-empty, so consider an arbitrary member $\path$ and recall that its first element is $s$.   Since $\strprof, s \Vvdash \gamma$ it follows that $\path \models \alpha \until \beta$. Since we assumed that $\gmod, s \nsat \bigvee \fgam$, we cannot have $\gmod, s \sat \beta$ since this would give   $\gmod, s \sat \beta \wedge \brak{\gamma}$ which entails $\gmod, s \sat \beta \wedge \brak{\gamma\vert_{C}}$, and $\beta \wedge \brak{\gamma\vert_{C}}$ is a member of $\fgam$. Hence we have $\gmod, s \sat \alpha$, as required. The proof that each conjunct from $\agam$ is satisfied is similar (but simpler). 

We now show that $\strprof, s \Vvdash \diffof{\gamma}$. Pick an arbitrary coalition $C$ in the support of $\diffof{\gamma}$ and an arbitrary path $\path \in \paths(s,\strprof,C)$. We need to show that $\path \models \diffof{\gamma}(C)$. Suppose that $\path$ is the path generated by a play in $\Plays(w,\strprof,C)$ of the form:
$$w_0 \actprof_0 w_1 \actprof_{1} w_2 ...$$
where $w_0 = s$. We need to show that:
\begin{enumerate}
\item For each $C' \subseteq C$ in the support of $\gamma$, if $\gamma(C') = \nexttime \psi$ then $\gmod,w_1 \sat \psi$, 
\item $\gmod, w_1 \sat \brak{(\gamma\ugpart)\vert_{C}}$.
\end{enumerate}
The first item is straightforward, so we focus on item (2). We need to come up with a strategy profile $\strprof'$ that witnesses $(\gamma\ugpart)\vert_{C}$ at $w_1$. We define $\strprof'$ as follows. 
Given a history of the form:
$$v_0 \actprof'_0 v_1... v_{n -1 }\actprof'_{n - 1} v_n$$ 
where $v_0 = w_1$, and a player $\aga \in C$, we set:
$$\strprof'(\aga, v_0 \actprof'_0 v_1... v_{n -1 }\actprof'_{n - 1} v_n) := \strprof(\aga, w_0 \actprof_0 v_0 \actprof'_0 v_1... v_{n -1 }\actprof'_{n - 1} v_n)$$
Now let $C' \subseteq C$ be a coalition for which $\gamma(C') = \alpha \until \beta$, and let $\path' \in \paths(w_1,\strprof',C')$. Then $s\path' \in \paths(s,\strprof,C')$ by construction of $\strprof'$, hence $s\path' \models \alpha \until \beta$. 
Since $\gmod,s\nsat \beta$,  
we get $\path' \models \alpha \until \beta$ as well. Similarly we can show that if  $C' \subseteq C$ is a coalition for which $\gamma(C') = \always \alpha$, and $\path' \in \paths(w_1,\strprof',C')$, then $\path' \models \always \alpha$. This shows that $\strprof',w \Vvdash (\gamma\ugpart)\vert_{C}$, as required. 

\paragraph{Right to left:}
 Suppose $\gmod, s \sat \brak{\unf{\gamma}}$. We show that $\gmod, s \sat \brak{\gamma}$.

There are two cases to consider: either some formula in $\fgam$ holds at $s$, or:
$$
\gmod, s \sat \bigwedge \ugam \wedge \bigwedge \agam \wedge \brak{\diffof{\gamma}}.
$$
In the first case there is some $C$ in the support $\fac$ of $\gamma$ for which $\gamma(C) = \alpha \until \beta$ and:
$$ \gmod, s \sat \beta \wedge \brak{\gamma \setminus C}.$$
But then $\alpha \until \beta$ holds at any path beginning from $s$, and it follows that $\gmod, s \sat \brak{(\gamma \setminus C)[C \gass \alpha \until \beta]}$. This formula is the same as $\brak{\gamma}$, so we are done. 

Suppose the second, more challenging case, and fix a strategy profile $\strprof$ for $\Agt$ witnessing $\brak{\diffof{\gamma}}$ at $s$. Given any locally available action profile $\actprof \in \ActProf_s$, 
we define the set 
$\mathsf{fol}(\actprof)$ of \defstyle{followers of $\strprof$ relative to $\actprof$}: 
\[
\mathsf{fol}(\actprof) := \bigcup \big\{C \in \fac \mid \actprof(\aga) = \strprof(\aga,s) 
\mbox{ for all } \aga \in C
\big\}.
\]
Note that $\mathsf{fol}(\actprof)$ belongs to the support of $\diffof{\gamma}$. 
By unfolding the definition of $\brak{\diffof{\gamma}}$, we see that the following conditions hold for each locally available action profile $\actprof\in \ActProf_s$
and any $C \in \fac$ such that $C \subseteq \mathsf{fol}(\actprof)$:
\begin{enumerate}
\item  $\out(\actprof,s) \in \tset{\brak{(\gamma\ugpart)\vert_{C}}}_{\gmod}$.
\item If $\gamma(C) = \nexttime \varphi$ then 
$\out(\actprof,s) \in \tset{\varphi}_{\gmod}$.
\end{enumerate} 
This motivates the following definition: for each locally available action profile 
$\actprof\in  \ActProf_s$  
we pick a strategy profile  $\underline{\actprof}$ defined 
for all players in $\bigcup \fac$, such that for all $C \subseteq \mathsf{fol}(\actprof)$: 
\[
\quad \underline{\actprof},\out(\actprof,s) \Vvdash (\gamma\ugpart)\vert_{C}.
\] 
Now, we will build a strategy profile $\Omega$ using the strategy profile $\strprof$ and all strategy profiles $\underline{\actprof}$ for each possible locally available action profile $\actprof$ at $s$:  
  
Given a player $\aga \in \bigcup \fac$ and $w \in \states$, let $\Omega(\aga,w) = \strprof(\aga,w)$ if $w = s$, and some arbitrary available move otherwise. For a history of the form $w_0\actprof_0...\actprof_n w_{n+1}$, if $w_0 \neq s$ then we can again define the move of player $\aga$ arbitrarily. 
Otherwise, we set 
\[
\Omega(\aga,w_0\actprof_0 w_1....\actprof_n w_{n+1}) := 
\underline{\actprof_0}(\aga, w_1 \actprof_1 ... \actprof_n w_{n + 1}).
\]
We will show that the strategy profile $\Omega$ witnesses $\brak{\gamma}$ at $s$. Let $C \in \fac$, and let $\path \in \paths(s,\Omega,C)$. We need to show that $\path \models \gamma(C)$.  
The case where $\gamma(C) = \nexttime \varphi$ is immediate, by definition of $\diffof{\gamma}$ and of $\Omega$ at $s$. 
We focus on the case where $\gamma(C)$ is of the form $\alpha \until \beta$. We assume that $\path$ is the path generated by a play in $\Plays(s,\Omega,C)$ of the form:
$$w_0 \actprof_0 w_1 \actprof_1 w_2 ...$$
so that $\path$ equals $w_0w_1w_2...$, $w_0 = s$ and $w_1 = \out(s,\actprof_0)$. But then $C \subseteq \mathsf{fol}(\actprof_0)$, and by construction of the strategy profile $\Omega$, the play $w_1 \actprof_1 w_2 \actprof_2 w_3...$ belongs to $\Plays(w_1,\underline{\actprof}_0,C)$. So the path $w_1w_2w_3...$ satisfies $\alpha \until \beta$ since, by definition, $\underline{\actprof}_0, w_1 \Vvdash  (\gamma\ugpart)\vert_{C}$. Since $\alpha \in \ugam$, we have $\gmod,s \sat \alpha$. Since $s = w_0$ it follows that $\path \models \alpha \until \beta$ as required.     
  Lastly, the case where $\gamma(C)$ is of the form $ \always \chi$ is analogous.      
\end{proof}
}

Finally, we show how to extend these definitions to the logic $\cga^+$. In this setting, goal assignments are of a more general kind as they map coalitions to conjunctions of $\X$-formulas, $\until$-formulas and/or $\always$-formulas, so we have to account for this slight complication. Given a goal assignment $\gamma$, and a coalition $C$ in its support, we can think of $\gamma(C)$ as the set of its conjuncts of the form $\X \alpha$, $\alpha \until \beta$ or $\always \chi$. So we abuse notation slightly and write $\theta \in \gamma(C)$ to say that $\theta$ has one of these three forms, and is a conjunct of $\gamma(C)$. With this notation in place, we extend the definition of the
nexttime-extension $\diffof{\gamma}$ of a goal assignment $\gamma$ as follows. The support of $\diffof{\gamma}$ is defined as before, and for each $C \in \sup{\diffof{\gamma}}$ we define 
\[
\diffof{\gamma}(C) := \X \Big(\bigwedge \big\{\varphi \mid \mbox{there exists } C' \in \fac, C' \subseteq C  \mbox{ such that }  \X\varphi \in \gamma(C') \big\} 
\wedge \brak{(\gamma\vert_C)\ugpart} \Big),
\]
As before we abbreviate the formula $\diffof{\gamma}[\bigcup \fac  \gass  \nexttime \phi]$ by $\gammaof{\phi}$. 

Next we define the unfolding of a goal assignment:
Let $\gamma$ be a goal assignment, supported by $\fac$. Given a coalition $C$ and a conjunct $\theta$ of $\gamma(C)$, we write $\gamma\setminus(C\gass\theta)$ for the goal assignment that is like $\gamma$, except:
$$(\gamma\setminus(C\gass \theta))(C) = \bigwedge \{\theta' \mid \theta' \in \gamma(C) \wedge \theta' \neq \theta\}$$
We now define:
\[\unf{\gamma} := 
\bigvee \mathsf{Finish}(\gamma) \vee \bigg(\bigwedge \ugam \wedge \bigwedge \agam \wedge \brak{\diffof{\gamma}}\bigg),
\]
where:
\begin{itemize}
\item $\mathsf{Finish}(\gamma) : = \big\{\beta \wedge \brak{\gamma \setminus (C \gass \alpha \until \beta)} \mid 
\gamma(C) = 
\alpha \until \beta\big\}$

\item $\ugam :=  \big\{\alpha \mid  
 \alpha \until \beta \in \gamma(C), 
\mbox{ for some } C,\beta \big\}$

\item $\agam : = \big\{\chi \mid   
\always \chi \in \gamma(C),  
\mbox{ for some } C \big\}$
\end{itemize}

The induction formula for $\indf{\gamma}{\phi}$ is defined as before, by 

\[
\indf{\gamma}{\phi} :=
\bigvee \mathsf{Finish}(\gamma) \vee \Big(\bigwedge \ugam \wedge \bigwedge \agam \wedge 
\brak{\gammaof{\phi}}\Big), 
\]
where $\gammaof{\phi}$ is $\diffof{\gamma}[\bigcup \fac  \gass  \nexttime \phi]$. With these extended definitions, the proofs of Proposition \ref{prop:unfold} and Theorem \ref{p:fixpoint-property} go through as before. 

\subsection{A $\mu$-calculus of goal assignments}
\label{subsec:mu-TLCGA}

The  $\mu$-calculus extension of the language $\langcga$ of \cga will be denoted by 
$\langcga_\mu$, and the $\mu$-calculus extension of the nexttime fragment 
$\xlangcga$ -- by  $\xlangcga_{\mu}$.

Formally the language $\langcga_\mu$ is given by the following grammar:  
$$\sfml: \ \ \ \ \ \varphi := p  \mid \top 
\mid \neg \varphi 
\mid (\varphi \wedge \varphi) 
 \mid (\varphi \lor \varphi) 
\mid  
\brak{\gamma} \mid \mu x.\varphi
$$
$$\pfml: \ \ \ \ \ \theta := \nexttime \varphi \mid \varphi \until \varphi \mid \always \varphi$$
Here, in $\mu x. \varphi$ the formula $\varphi$ is subject to the usual constraint that every occurrence of the variable $x$ in $\varphi$ is positive, in the sense that it is under the scope of an even number (possibly zero) of negations. We usually denote bound variables $x,y,z...$ rather than $p,q,r...$, but formally we do not introduce a separate supply of fixpoint variables. In the formula $\mu x.\varphi$ the variable $x$ is simply a propositional variable. 
We define the greatest fixpoint operator as usual: 
\[\nu x.\varphi := \lnot \mu x. \lnot \varphi [\lnot x / x],\] 
where $\varphi [\lnot x / x]$ is the result of uniform substitution of $\lnot x$ for $x$ in 
$\varphi$.

A model for $\langcga_\mu$ is just like a model for $\langcga$, viz a tuple,  
$\gmod = (\states,\Act,\gmap,\out,V)$, 
but now the valuation $V$ assigns values to the variable(s) $z$ used in formulae  
$\mu z.\psi$. Now, for each $Z \subseteq \states$, we define the amended valuation $V^Z := V[z \mapsto Z]$, which is like $V$ except that it maps $z$ to $Z$. 
We will denote 
 $\gmod^Z := (\states,\Act,\gmap,\out,V^Z)$ and for any formula $\psi (z)$ in which the variable $z$ may occur free, we will write 
 $\gmod, s \sat \psi (Z)$ to state that  $\gmod^Z, s \sat \psi (z)$.

The semantics of $\langcga_\mu$ extends that of $\langcga$ with the additional clause that the extension $\tset{\mu z. \varphi(z)}$ of a least fixpoint-formula in a model $\gmod = (\states,\Act,\gmap,\out,V)$ is given by: 
\[
\tset{\mu z. \varphi(z)} : = \bigcap \big\{Z \subseteq \states \mid \tset{\varphi(Z)} \subseteq Z \big\},
\]
where, as expected:
$$\tset{\varphi(Z)} := \big\{w \in \states \mid \gmod^Z,w \sat \varphi(z) \big\}.$$

Given a model 
$\gmod = (\states,\Act,\gmap,\out,V)$, 
 let  $f_\gamma$ be the monotone map on $\psf(\states)$ induced by 
 $\indf{\gamma}{z}$ in the usual way, i.e., for each $Z \subseteq \states$, $f_\gamma(Z)$ is the set of states satisfying $\indf{\gamma}{z}$ with respect to the amended valuation $V^Z := V[z \mapsto Z]$.

The proofs of the following two propositions are in the appendix. 

\begin{proposition}[Fixpoint characterization of $\typeone$ temporal goal assignments]
\label{prop:typeone}
Suppose that $\gamma$ is a long-term temporal  goal assignment in $\typeone$, and let $z$ be a fresh variable not occurring in $\brak{\gamma}$. 
Then $\brak{\gamma} \equiv \mu z. \indf{\gamma}{z}$.
\end{proposition}

\begin{proposition}[Fixpoint characterization of $\typetwo$ temporal goal assignments]
\label{prop:typetwo}
Suppose that $\gamma$ is a long-term temporal goal assignment in $\typetwo$. 
Then $\brak{\gamma} \equiv \nu z. \indf{\gamma}{z}$.
\end{proposition}

Recall that (memory-based) strategies are defined here in terms of \emph{plays}, not just paths  
as for \ATLs in \cite{AHK-02} (cf also \cite{BGJ15} or \cite{TLCSbook}), and we showed in Example \ref{exampleB} that the two versions affect essentially the semantics. In fact, the model in that example also shows that  $\nu z. \indf{\gamma}{z} \to \brak{\gamma}$ is not valid in the semantics with path-based strategies.

Using the fixpoint characterizations of long-term temporal goal assignments we have established here, we can define an explicit translation $t : \langcga \to \xlangcga_\mu$, preserving the semantics of each $\langcga$-formula. To make this precise, the translation has to be defined by induction on a certain wellfounded order $\prec$ over $\langcga$-formulas, defined as the smallest transitive relation closed under the following rules:
\begin{itemize}
\item If $\varphi$ is a proper subformula of $\psi$ then $\varphi \prec \psi$. 
\item If $\gamma$ is a goal assignment which is neither nexttime nor long-term temporal, i.e. the support of $\gamma$ contains both some coalitions mapped to nexttime formulas and some coalitions mapped to long-term temporal goals, then $\unfold{\gamma} \prec \brak{\gamma}$.
\item If $\gamma$  is a $\typeone$ or $\typetwo$-goal assignment, and $z$ is a propositional variable not occurring in $\brak{\gamma}$, then $\indf{\gamma}{z} \prec \brak{\gamma}$.  
\end{itemize}
To see that this order is wellfounded, consider for example the clause where $\gamma$ is in $\typeone$. Then the formula $\indf{\gamma}{z}$ can  be viewed as being built up using boolean connectives from proper subformulas of $\brak{\gamma}$ and formulas of the form $\brak{\gamma'}$ where $\gamma'$ is either $\gammaof{z}$ or $\gamma\setminus C$ for some coalition $C$ in the support of $\gamma$. Note that the goal assignment $\gamma \setminus C$ has a smaller support than $\gamma$. Furthermore, the goal assignment $\gammaof{z}$ is a nexttime goal assignment, the coalition $\bigcup \mathcal{F}$ is mapped to $\X z$ (where $\mathcal{F}$ is the support of $\gamma$) and each coalition $C \neq \bigcup \mathcal{F}$ in the support of $\gammaof{z}$ is mapped to $\X \brak{\gamma\vert_C}$. But each goal assignment $\gamma \vert_C$ for $C \neq \bigcup \mathcal{F}$ has smaller support than $\gamma$, since if $C$ a proper subset of $\bigcup \mathcal{F}$ then there is some $C' \in \mathcal{F}$ which is not contained in $C$. So in $\gamma\vert_C$, the coalition $C'$ will be assigned the trivial goal $\X \top$. 

With the wellfounded order $\prec$ in place, the translation can be defined by $t(\brak{\gamma}) = \mu z. t(\indf{\gamma}{z})$ if $\gamma$ is a type $\until$ goal assignment, $t(\brak{\gamma}) = \nu z. t(\indf{\gamma}{z})$ if $\gamma$ is of type $\always$ and $t(\brak{\gamma}) = t(\unfold{\gamma})$ if $\gamma$ contains both $\until$-formulas and $\always$-formulas as goals. The induction steps for booleans and nexttime goal assignments are handled in the standard manner, simply letting the translation commute with these connectives.

As noted at the beginning of the section, it is a relatively routine task to extend our translation to cover the richer language $\mathcal{L}^{\cga^+}$ underlying the logic $\cga^+$. Again, the translation is by induction on a wellfounded ordering over formulas, defined as above but with the extended definitions of $\unfold{\gamma}$ and $\indf{\gamma}{z}$. The proofs of Proposition \ref{prop:typeone} and \ref{prop:typetwo} go through essentially as before.

\subsection{A note on coalgebras}
One reason why the translation of $\langcga$ into $\xlangcga$  we have presented is useful is because it establishes a connection with \emph{coalgebraic modal logics}. We  note that the nexttime fragment  $\xlangcga$ of $\langcga$ is, in fact, an instance of the general framework of coalgebraic modal  logics, making the language $\xlangcga_\mu$ a \emph{coalgebraic fixpoint logic}, as studied in \cite{venema2006automata,cirstea2011exptime,fontaine2010automata}. This connection helps to clarify  the place of the logic \cga in the landscape of modal fixpoint logics for various kinds of state-based evolving systems. The theory of \emph{universal coalgebra} appears in computer science as a generic framework for modelling a wide range of state-based evolving systems in a uniform manner \cite{rutten2000universal}. It is formulated using the language of basic category theory (functors and natural transformations), see  \cite{maclane1971categories} for a standard reference. The key idea of universal coalgebra is to pack all the information about the type of transitions that a system can make (deterministic, non-deterministic, probabilistic etc.) into a functor on the category of sets,  which then can be considered as a variable parameter featuring in abstract definitions and results. A coalgebra for a functor $\fun$ is then just a set $X$ together with a map $f : X \to \fun X$, which intuitively represents the evolution of the state-based system. Concurrent game models are in fact coalgebras for a certain functor (together with a valuation of propositional variables), as has been observed several times, e.g. \cite{schroder2009pspace}. Furthermore, it can be checked that the semantics of modalities $\brak{\gamma}$ in which each goal formula is of the form $\X \theta$ for some $\theta$ can be phrased in terms of \emph{predicate liftings} for this functor, which is the currently most usual way of interpreting modalities in coalgebras.  We omit the details here. 

The coalgebraic representation of the logic $\xlangcga_{\mu}$, together with the translation from \cga into $\xlangcga_{\mu}$, gives us access to a wealth of known general results on coalgebraic fixpoint logics. For example we shall use it to derive decidability and finite model property for the logic $\cga$, as well as a complexity bound on the satisfiability problem for the language $\xlangcga_\mu$. However, we emphasize that some caution is required here. Universal coalgebra and coalgebraic logic are valuable frameworks, unifying a large class of systems and associated logics, just like universal algebra provides a common language that puts many different algebraic structures under one roof. But, of course, universal algebra does not tell us everything we want to know about specific classes of algebras, like groups or Heyting algebras. Generic results are helpful, but a detailed study of concrete special cases is usually required. The same is true here. In particular, our main technical result in this paper on completeness for an axiomatization of \cga makes heavy use of ideas from the literature on coalgebraic logic, in particular the notion of ``one-step completeness'' is essential \cite{schroder2009pspace,cirstea2011exptime}. But the proof of one-step completeness is not a trivial consequence of generic results from coalgebra, it requires a careful study of the semantics of nexttime goal assignments.

Once we have one-step completeness in place, the next step will be to handle least and greatest fixpoints.  We have shown that \cga can be translated into  
$\xlangcga_{\mu}$ using a single recursion variable, in effect embedding \cga as a fragment of a ``flat coalgebraic fixpoint logic'' in the sense of Schr\"{o}der and Venema \cite{schroder2010flat}, who prove a generic completeness result for such logics. It is possible that we could obtain completeness for our axiom system by `transferring' Schr\"{o}der's and Venema's  completeness result via our translation, but we have chosen to present a direct proof instead since we believe this will be more instructive. That said, we do consider our approach as being very much coalgebraic in spirit, and the connection with coalgebraic fixpoint logic is conceptually important. In particular, the idea that one-step completeness of a (multi-)modal logic can be ``lifted'' to give a complete axiomatization of a fixpoint extension \cite{GorDrim06,schroder2010flat,cirstea2011exptime,enqvist2019completeness} is at the heart of our proof.

\section{Bisimulations and bisimulation invariance for \cga}
\label{subsec:Bisimulations}

As noted earlier, the logic \GPCL introduced in \cite{GorankoEnqvist18} is essentially the nexttime fragment $\xlangcga$ of \cga. Therefore, the notion of  \GPCL-bisimulation (ibid.) also applies to \cga. For the reader's convenience, we introduce it again here, now called 
\defstyle{\cga-bisimulation} and extend the bisimulation invariance result from \cite{GorankoEnqvist18} to the full logic \cga. 
This notion of bisimulation corresponds to the play-based semantics which is of importance for us. A similar notion can be defined for the path-based semantics and bisimulation invariance of \cga formulae with respect to it can be proved likewise, which we leave out.

We only define \cga-bisimulation within a single concurrent game model, and generalise to bisimulations between game models via disjoint unions.  

\begin{definition}[\cga-bisimulation]
\label{def:GPCLbisimulation} 
Let 
\[\gmod = (\states,\Act,\gmap,\out,V)\] 
be a game model for the set of agents $\Agt$. A binary relation $\beta  \subseteq \states^{2}$ is a \defstyle{\cga-bisimulation in \gmod} 
if it satisfies the following conditions for every pair of states $(s_1, s_2)$  
such that $s_1 \beta s_2$: 

\begin{description}
\itemsep = 2pt 
\item[Atom equivalence:] For every $p \in \Prop$:  $s_1 \in V(p)$ iff 
$s_2 \in V(p)$.

\item[Forth:]  For  any action profile $\actprof^{1}$ of $\Agt$ at $s_1$ there is 
an action profile $\actprof^2$ of $\Agt$ at $s_2$ such that:

\begin{description}
\itemsep = 1pt 
\item[LocalBack:] For every coalition $C$ and every $u_{2} \in \Out[s_{2},\actprof^{2}\vert_C]$, there is some $u_1 \in \Out[s_1,\actprof^1\vert_C]$ such that $u_1 \beta u_2$.
\end{description}

\item[Back:]  For  any joint action $\actprof^{2}$ of $\Agt$ at $s_2$ there is 
a joint action $\actprof^1$ of $\Agt$ at $s_1$ such that:

\begin{description}
\itemsep = 1pt 
\item[LocalForth:] For every coalition $C$ and every $u_{1} \in \Out[s_{1},\actprof^{1}\vert_C]$, there is some $u_2 \in \Out[s_2,\actprof^2\vert_C]$ such that $u_1 \beta u_2$.
\end{description}
\end{description}

States $s_1, s_2 \in \gmod$  are \defstyle{\cga-bisimulation equivalent}, or just \defstyle{\cga-bisimilar}, 
if there is a bisimulation $\beta$ in $\gmod$ such that $s_1 \beta s_2$. 
\end{definition}

The proof of the following theorem is relatively routine, and can be found in the appendix. 

\begin{theorem}[\cga-bisimulation invariance]
\label{thm:bisimulation invariance}
Let  $\beta$ be a \cga-bisimulation in a game model $\gmod$. Then for every \cga-formula $\varphi$ and every pair $s_1, s_2 \in \gmod$  such that $s_1 \beta s_2$: 
\[
\gmod, s_{1} \models \varphi \ \mbox{iff} \ \gmod, s_{2} \models \varphi 
\]
\end{theorem}

In fact, the proof of Theorem \ref{thm:bisimulation invariance} essentially amounts to a proof of bisimulation invariance for the whole language $\xlangcga_\mu$, and therefore also for the fragment $\cga^+$.

\begin{theorem}
\label{thm:bisimulation invariance+}
Every formula of $\xlangcga_\mu$ is bisimulation invariant. 
\end{theorem}

Furthermore, we also have the following result from \cite{GorankoEnqvist18}, showing that \cga already suffices to capture bisimulation invariance in finite models.

\begin{proposition}[Hennessy-Milner property, cf 
\cite{GorankoEnqvist18}, Proposition 4.6] 
Let  $\beta$ be a \cga-bisimulation in a finite game model $\gmod = (\states,\Act,\gmap,\out,V)$. 
Then for any pair $s_1, s_2 \in \states$, $s_1 \beta s_2$ holds iff $s_1$ and $s_2$  satisfy the same $\cga$-formulae.  
\end{proposition}

\begin{proof}  
For the non-trivial direction we will use only formulae from the fragment
$\xlangcga$. 
Since $\gmod$ is finite, we can define, by a standard construction, a `characteristic formula' $\chfor(s)$ for each state $s$ in $\gmod$, such that $s_1,s_2$ are $\xlangcga$-equivalent if and only if $\chfor(s_1)  = \chfor(s_2)$, and that $\chfor(s_1) \land \chfor(s_2) \equiv \bot$ whenever $s_1,s_2$ are not $\xlangcga$-equivalent. 
 For a set  of states $Z$, let $\chfor[Z] = \bigvee \{\chfor(v) \mid v  \in Z\}$. Our goal is to show that the relation of \cga-equivalence is itself a \cga-bisimulation, and the key observation is that each state $s $ satisfies the $\xlangcga$-formula:
\begin{eqnarray*} \bigwedge_{\actprof \in \ActProf_s} \cgoal{C_1 \gass \X \chfor[\Out[s,\actprof\vert_{C_1}],...,C_k \gass \X \chfor[\Out[s,\actprof\vert_{C_k}] }
\end{eqnarray*}
where we list the set $\mathcal{P} (\Agt)$ of all possible coalitions as $C_1,...,C_k$.
\end{proof}

\section{Axiomatization and one-step completeness of \cga}
\label{sec:axiomatization}

In this section we focus exclusively on \cga and leave the extension of the axiomatic system presented here to $\cga^+$, as well as the respective axiomatizations for the path-based semantics (which may turn out to be more problematic) for future work.  

\subsection{Axiomatic system for \cga}
\label{subsec:Axiomatization}

\begin{definition}
Let $\fac$ be a set of coalitions. 
A \defstyle{voting profile} for $\fac$ is a mapping $f$ 
assigning to each $\aga_i \in \Agt$ a goal assignment $f(\aga_i)$. 
If $f(\aga_i)(C)$ is a nexttime formula for each $i$ and $C \in \fac$, we say that $f$ is a \defstyle{one-step voting profile} for $\fac$. 
\end{definition}

The notion of merging a voting profile, defined next, will be used in some proofs later  and we will need some derivable formulae that use it, listed further. 

\begin{definition}
Let $f$ be a voting profile. We define the goal assignment $\mathsf{merge}(f)$ as follows: 

\begin{itemize}
\itemsep = 2pt
\item 
$\mathsf{merge}(f) (C) := \theta$,  
if 
{$C \neq \emptyset$ and $f(\aga_i)(C)= \theta$} for each $a_i\in C$, \ 
\item $\mathsf{merge}(f)(C) := \nexttime \top$, 
{if $C = \emptyset$ or} the above holds for no $ \theta$.
\end{itemize}

\end{definition}

Our axioms are as follows (recall notation on goal assignments from Section \ref{subsec:CGAsyntax}).

\subsubsection{I. General axiom schemes for goal assignments}

\begin{description}     
\itemsep = 2pt

\item[(Triv)] $\brak{\gamma^\top}$ \ \  (Recall that $\gamma^\top$ is the trivial goal assignment, mapping each coalition to $\nexttime \top$)

\item[(Safe)] $\neg \brak{\Agt \gass  \nexttime \bot}$

\item[(Merge)] $\brak{C_1 \gass \theta_1} \wedge ... \wedge \brak{C_n \gass \theta_n} \to \brak{C_1  \gass  \theta_1,...,C_n  \gass \theta_n}$, \  where $C_1,...,C_n$ are pairwise disjoint. 

This axiom generalises the Superadditivity axiom of Coalition Logic. The idea is simple: if the coalitions $C_1,...,C_n$ are pairwise disjoint, then they can join their collective strategies for their respective coalitional goals into one strategy profile that ensures achievement of all these collective goals.

\item[(GrandCoalition)]  
$\brak{\gamma} \rightarrow (\brak{\gamma[\Agt  \gass  \nexttime(\varphi \wedge \psi)]} \vee \brak{\gamma[\Agt  \gass  \nexttime (\varphi \wedge \neg\psi)]} )$, \ 
where $\gamma(\Agt) = \nexttime \varphi$. 

Any strategy profile generates a unique successor state, on which any state formula $\psi$ is either true or its negation is true, so either $\psi$ or $\lnot \psi$ can be added to the nexttime goal of the grand coalition $\Agt$.

\item[(Case)] $\brak{\gamma} \rightarrow (\brak{\gamma[C  \gass  \nexttime (\varphi \wedge \psi)]} \vee \brak{\gamma\vert_C[(\Agt  \gass \nexttime \neg\psi]})$, where $\gamma(C) = \nexttime \varphi$. 

For any coalition $C$, state formula $\psi$, and a strategy profile $\strprof$, either its projection 
$\strprof_C$ to $C$ ensures the truth of $\psi$ in all successor states enabled by $\strprof_C$ -- in which case $\psi$ can be added to the nexttime goal of $C$ enforced  by $\strprof$ -- or else $\lnot \psi$ is true in some of these successor states, in which case it can be added to $\gamma\vert_C$ as the nexttime goal of the grand coalition $\Agt$ enforced  by $\strprof$.  

\item[(Con)] $\cgoal{\gamma} \to \cgoal{\gamma[C \gass \nexttime (\varphi \wedge \psi)]}$ where $\gamma(C) = \nexttime \varphi$ and  $\gamma(C') = \nexttime \psi$ for some $C' \subseteq C$. 

Given any coalition $C$ and sub-coalition $C'$, the nexttime goal of $C'$ can be added to the nexttime goal of  $C$, in sense that if there is any strategy profile $\strprof$ that ensures that $C$ and  $C'$ can force their respective nexttime goals $\nexttime \varphi$ and $\nexttime \psi$, then $\strprof$ also ensures that $C$ can force the conjunction of these goals.
         
\end{description}


 \subsubsection{II. General rules of inference:} 
 ~
 
\textbf{Modus Ponens} 
and 
\textbf{Goal Monotonicity (G-Mon)}:  
\[\frac{\phi \to \psi}{\cgoal{\gamma[C\! \gass  \X \phi]} \to \cgoal{\gamma[C\! \gass  \X \psi]}}\] 

The meaning of the rule is clear: if $\phi$ implies $\psi$, then any coalition $C$ that can ensure the nexttime goal $\phi$ within the context of some strategic goal assignment, can also ensure $\psi$ with the same context.

\subsubsection{III. Axioms and rules for the long-term goal assignments} 
\label{sec:TemporalAxioms}
~

The axioms and rules for the goal assignments of Types 1 and 2, involving long-term temporal operators are given on Figure \ref{fig:rules}. They are adapted from the respective axioms and rules for least and greatest fixed points in the modal mu-calculus.  
In the axiom \textsf{Fix}, $\gamma$ is any goal assignment. In the rule \textsf{R-Ind} it is a long-term temporal assignment of type $\until$, and in \textsf{R-CoInd} it is a long-term temporal assignment of type $\always$.

\begin{figure}[htb]
\fbox{
\begin{minipage}[t]{.9\textwidth}
\begin{prooftree}
\AxiomC{\textsf{Fix:}  \; $\unf{\gamma} \leftrightarrow \brak{\gamma}$ }
\end{prooftree}

\begin{prooftree}
\AxiomC{$\indf{\gamma}{\phi} \to \phi$ }
\LeftLabel{\textsf{R-Ind:}}
\RightLabel{\;\;($\gamma \in \typeone$)}
\UnaryInfC{$\brak{\gamma} \rightarrow \phi $}
\end{prooftree}
\begin{prooftree}
\AxiomC{$\phi \to \indf{\gamma}{\phi} $ }
\LeftLabel{\textsf{R-CoInd:}}
\RightLabel{\;\;($\gamma \in \typetwo$)}
\UnaryInfC{$\phi \to \brak{\gamma} $}
\end{prooftree}
\end{minipage}
}
\caption{Fixpoint axiom and induction rules}
\label{fig:rules}
\end{figure}

We denote the axiomatic system above by $\textsf{Ax}_{\cga}$ and will denote derivability in it by $\textsf{Ax}_{\cga} \vdash$, but will often write just  $\vdash$. 
Here are some important validities that are derivable in $\textsf{Ax}_{\cga}$, some of which will be used in the proofs further:   

\begin{description}
\itemsep = 2pt

\item[\textsf{Ind}] $ \brak{\gamma} \ifff \indf{\gamma}{\brak{\gamma}}$ \ for every long-term temporal goal assignment $\gamma$. 

(Immediately from (Fix), due to Proposition \ref{prop:unfold}).

\item[(Weakening)]  \ 
\label{GPCL-G0}
$\cgoal{\gamma} \to 
\cgoal{C\gass \gamma(C)}$, \ 
for any $C \subseteq \Agt$.  \ \ 
 (Using (Triv) and (G-Mon).)

\item[$\Agt$-Maximality]  \ 
\label{GPCL-G1}  
            $\cgoal{\emptyset \gass\X\phi}        
            \lor
            \cgoal{\Agt \gass\X\neg\phi}$. \ (Using (Triv) and (Case).)

\item[(Superadditivity)]  \  
\label{GPCL-G2}
$\cgoal{C_1\gass\X\phi_1} \land \cgoal{C_2\gass\X\phi_2}
\to \cgoal{C_1  \cup C_2\, \gass \X(\phi_1 \land \phi_2); \ 
C_1 \gass \X\phi_1; \ C_2 \gass \X\phi_2 }$,  \     
 if $C_1 \cap C_2 = \emptyset$. 
 
This subsumes the Superadditivity axiom for Coalition Logic CL. It is derivable from (Merge)) using twice (Con) to add $\X(\phi_1 \land \phi_2)$ as a goal assignment to $C_1  \cup C_2$.

\item[(Merge')] $\bigwedge_{a_i \in \Agt} \cgoal{f(\aga_i)} \rightarrow \cgoal{\mathsf{merge}(f)}$, where $f$ is any 
voting profile. 

This is an essentially equivalent formulation of (Merge). Indeed, (Merge) is a particular case of (Merge'), whereas (Merge') is derivable from (Merge) by first using (Weakening) to detach each $\brak{C_j \gass \theta_j}$ from every $f(\aga_i)$, for $\aga_i \in C_j$, 
{if $C_j \neq \emptyset$ and} $f(\aga_i)(C_j) = \theta_j$ for all such $\aga_i$.

\item[\textsf{Fix}($\always$)]    
\label{PostFP(G)} 
\ 
$ \brak{C \gass \always \chi} \to 
\chi \land \brak{C \gass \X \brak{C \gass \always \chi}}$. \\
This is a special case of $\mathsf{Fix}$.

\item[\textsf{CoInd}($\always$)]    
\label{CoInd(G)} 
If 
$\vdash \phi \to \chi \land \brak{C \gass \X \phi}$
then $\vdash \phi \to \brak{C \gass \always \chi}$. \\ 
This is a special case of the rule $\mathsf{CoInd}$. 

In particular, 
by using \textsf{Fix}($\always$) and applying \textsf{G-Mon} we obtain 
 
$\vdash (\chi \land \brak{C \gass \X \brak{C \gass \always \chi}}) \to \chi \land \brak{C \gass \X (\chi \land \brak{C \gass \X \brak{C \gass \always \chi}})}$
 
Now, by applying \textsf{CoInd}($\always$) 
 for $\phi =  \chi \land \brak{C \gass \X \brak{C \gass \always \chi}}$, we derive 
 \[
\vdash \chi \land \brak{C \gass \X \brak{C \gass \always \chi}} \to 
 \brak{C \gass \always \chi}.
 \] 
 Thus, we have derived the fixpoint equivalence for $\always$: 
\item[\textsf{FP}($\always$)]    
\label{FP(G)}  
$ \brak{C \gass \always \chi} \ifff 
\chi \land \brak{C \gass \X \brak{C \gass \always \chi}}$.

\item[\textsf{PreFP($\until$)}]    
\label{PreFP(U)} 
$ \beta \lor (\alpha \land \brak{C \gass \X \brak{C \gass \alpha \until \beta}}) \to 
\brak{C \gass \alpha \until \beta}$. \\ 
This is a special case of $\mathsf{Fix}$. 

\item[\textsf{Ind}($\until$)]    
\label{Ind(U)} 
If 
$\vdash \beta \lor (\alpha \land \brak{C \gass \X \phi}) \to \phi$
then $\vdash \brak{C \gass \alpha \until \beta} \to \phi$. \\ 
This is a special case of the rule $\mathsf{Ind}$.

In particular, 
by applying the rule \textsf{G-Mon} to \textsf{PreFP($\until$)} 
we derive 
\[\brak{C \gass \X (\beta \lor (\alpha \land \brak{C \gass \X \brak{C \gass \alpha \until \beta}}))} \to \brak{C \gass \X \brak{C \gass \alpha \until \beta}}.\]
Then, by simple propositional inference we derive
 
$(\beta \lor (\alpha \land \brak{C \gass \X (\beta \lor (\alpha \land \brak{C \gass \X \brak{C \gass \alpha \until \beta}}))})) \to 
(\beta \lor (\alpha \land \brak{C \gass \X \brak{C \gass \alpha \until \beta}}))$. 
 
Now, by applying \textsf{Ind}($\until$) for 
 $\phi = \beta \lor (\alpha \land \brak{C \gass \X \brak{C \gass \alpha \until \beta}})$, 
 we derive 
\[\vdash \brak{C \gass \alpha \until \beta} \to 
 \beta \lor (\alpha \land \brak{C \gass \X \brak{C \gass \alpha \until \beta}})
 \]
 Thus, we have derived the fixpoint equivalence for $\until$: 
\item[FP($\until$)]    
\label{FP(U)} 
$\brak{C \gass \alpha \until \beta} \ifff 
 \beta \lor (\alpha \land \brak{C \gass \X \brak{C \gass \alpha \until \beta}})$
\end{description}

\begin{proposition}[Soundness of $\textsf{Ax}_{\cga}$]
\label{prop:soundness} The axiomatic system $\textsf{Ax}_{\cga}$ is sound: every derivable formula in $\textsf{Ax}_{\cga}$ is valid. 
\end{proposition}

\begin{proof}
We show that every axiom is valid and all rules of inference preserve validity. 

Checking validity of the general axiom schemes is fairly routine. Most of these, as well as the preservation of validity by the general rules II, follow from the soundness of the logic \GPCL in  \cite{GorankoEnqvist18}. 

The validity of the axiom scheme
$\mathsf{Fix}$ 
follows from 
Theorem 
\ref{p:fixpoint-property}. 

The preservation of validity by the special rule $\mathsf{R-Ind}$ can be shown as follows. Suppose $\indf{\gamma}{\phi} \to \phi$ is valid. 
Take any concurrent game model $\gmod$.
Then $\gmod \sat \indf{\gamma}{\phi} \to \phi$, hence 
$\tset{\phi}_{\gmod}$ is a pre-fixed point of the set operator induced by the formula 
$\indf{\gamma}{z}$ in $\gmod$. By  Proposition \ref{prop:typeone}, 
$\brak{\gamma}$ is semantically equivalent to the least fixed point $\mu z. \indf{\gamma}{z}$, which is also the least pre-fixed point of $\indf{\gamma}{z}$. 
Therefore, $\gmod \sat \brak{\gamma} \rightarrow \phi$. 
Thus, $\brak{\gamma} \rightarrow \phi$ is valid.

The preservation of validity by the special rule $\mathsf{R-CoInd}$ is proved analogously, using Proposition \ref{prop:typetwo} and the fact that the 
greatest fixed point $\nu z. \indf{\gamma}{z}$, is also its greatest post-fixed point. 
\end{proof}

Recall (cf Section \ref{subsec:TypesAssignments}) that a formula $\phi \in \langcga$ is in \defstyle{normal form} if, for every subformula of the form $\brak{\gamma}$, the goal assignment $\gamma$ is either a nexttime or a long-term temporal goal assignment. 

\begin{proposition}
\label{prop:NF}
For every formula $\varphi$ there is a formula $\psi$ which is in normal form, and such that $\textsf{Ax}_{\cga} \vdash \varphi \leftrightarrow \psi$.
\end{proposition}

\begin{proof}
By induction on the structure of formulas, using the axiom \textsf{Fix} for the crucial steps. By design, the unfolding $\unf{\gamma}$ of any goal assignment $\gamma$ is a nexttime goal assignment, and all new goal assignments appearing in the scope of nexttime operators in the codomain of $\unf{\gamma}$ will be long-term temporal. So, all mixing of nexttime and long-term temporal path formulas in $\brak{\unf{\gamma}}$ will appear in proper subformulas of $\brak{\gamma}$, where the inductive hypothesis is applied.  
\end{proof}

By the soundness, the proposition above implies the following corollary. 

\begin{corollary}
\label{cor:NF}
For every formula $\varphi$ there is a semantically equivalent formula $\psi$ which is in normal form. 
\end{corollary}

\subsection{Formula types, components, and extended FL-closure of \cga formulae} 
\label{subsec:components-and-closure}

We use some generic notions and terminology from the literature on tableaux-based satisfiability decision methods (see e.g. \cite[Ch.13]{TLCSbook}).
Formulae of \cga in normal form can be classified as:
\defstyle{literals}: $\top$, $\neg \top$, $p, \neg p$, where $p \in \Prop$, 
\defstyle{conjunctive formulae}, of the type $(\phi \land \psi)$  and $\lnot (\phi \lor \psi)$; 
\defstyle{disjunctive formulae}, of the type $(\phi \lor \psi)$ and $\lnot (\phi \land \psi)$; 
\defstyle{successor formulae}: $\brak{\gamma}$ and $\lnot \brak{\gamma}$, where 
$\gamma$ is a local (nexttime) goal assignment;  
and 
\defstyle{long term temporal formulae}, of the type $\brak{\gamma}$ and $\lnot \brak{\gamma}$, where $\gamma$ is a long term goal assignment. 
The formulae in the last four classes have respective \defstyle{components} 
that are given by Figure \ref{alphaTable}. 
Clearly, every conjunctive (respectively, disjunctive)  formula  in the table is equivalent to the conjunction (respectively, disjunction) of its components.

\begin{figure}[h]
\caption{Types of formulae and their components} 
\label{alphaTable}
\centering
\small
{\label{fig:edge-a}
$
\begin{array}{cc}
\begin{array}{l|l}
\mbox{Conjunctive formula} &  \mbox{Components} \\
\hline
\neg \neg \varphi &  \varphi \\
\varphi \land \psi &  \varphi, \;  \psi \\ 
\neg(\varphi \lor \psi) &  \neg \varphi, \; \neg \psi \\
\end{array}
& 
\hspace{5mm}
\begin{array}{l|l}
\mbox{Disjunctive formula} &  \mbox{Components} \\
\hline
\varphi \lor \psi &  \varphi, \;  \psi \\
\neg(\varphi \land \psi) &  \neg \varphi, \; \neg \psi \\
~ & ~ \\ 
\end{array}
\end{array}
$
}
\\ 
\vspace{3mm} 
{\label{fig:edge-c}
$
\begin{array}{l|l}
\mbox{Local goal formulae} &  \mbox{Components} \\ 
\hline
\brak{\gamma} \;\; \mbox{(positive)} \;\;  & 
\{\psi \mid \gamma(C) = \X \psi, \ C \subseteq \Agt \} 
\\
\neg \brak{\gamma} \;\; \mbox{(negative)} \;\;  & 
\{\neg \psi \mid \gamma(C) = \X \psi, \ C \subseteq \Agt \} 
\\
\end{array}
$
}
\\ 
\vspace{5mm}
{
$
\begin{array}{l|l}
\mbox{Temporal goal formulae} &  \mbox{Components}  \\ 
\hline
\brak{\gamma} \;\; \mbox{(positive)} \;\;  & \indf{\gamma}{\brak{\gamma}}  \\
\neg \brak{\gamma} \;\; \mbox{(negative)} \;\;  & \neg \indf{\gamma}{\brak{\gamma}} \\
\end{array}
$
}
\end{figure}

Given a formula $\varphi$, we define $\overline{\varphi} := \psi$ if $\varphi$ is of the form $\neg \psi$, and $\overline{\varphi} := \neg \varphi$ otherwise.

\begin{definition}
\label{def: extended (Fischer-Ladner) closure}
The \defstyle{extended (Fischer-Ladner) closure} of a \cga formula in normal form 
$\varphi$ is the least set of formulae $\ecl(\varphi)$ such that: 

\begin{enumerate}
\item 
$\varphi \in \ecl(\varphi)$,

\item 
$\ecl(\varphi)$ is closed under taking all components of formulae in $\ecl(\varphi)$.  

\item 
$\overline{\psi} \in \ecl(\varphi)$ whenever  $\psi \in \ecl(\varphi)$. \    
 \end{enumerate}
 
\smallskip

For any set of formulae  $\Phi$ we define
$\ecl(\Phi) :=  \bigcup \{\ecl(\varphi) \; \mid \; \varphi \in \Phi\} $.

A set $\Phi$ of \cga formulae in normal form is said to be \defstyle{(Fischer-Ladner) closed} iff $\ecl(\Phi) = \Phi$.  
\end{definition}

The proof of the following proposition can be found in the appendix:

\begin{proposition}
\label{prop:closure}
The extended closure of any finite set $\Phi$ of \cga formulae in normal form is finite.  
\end{proposition}

\subsection{One-step completeness}
\label{subsec:One-stepCompleteness}

Hereafter derivability/provability and consistency refer to  the axiomatic system 
$\textsf{Ax}_{\cga}$. Given a set of formulae $\Phi$, the maximal consistent subsets of $\Phi$ are defined as usual.  

\begin{definition}
Given a 
closed set of formulae $\Phi$, a \defstyle{$\Phi$-atom} is a maximal consistent subset of $\Phi$. We denote by $\mathsf{At}(\Phi)$ the set of all $\Phi$-atoms.
\end{definition}

\begin{definition}
Let $\Phi$ be a finite set of formulae. A \defstyle{nexttime assignment} over $\Phi$ is a formula of the shape $$\brak{C_1  \gass  \nexttime\varphi_1,..., C_k  \gass  \nexttime \varphi_k}$$
where each formula  $\varphi_i$ belongs to $\Phi$.  A
 \defstyle{modal one-step theory} over $\Phi$ is a finite set of formulae $\Gamma$, such that every formula in $\Gamma$ is either a nexttime assignment over $\Phi$ or the negation of such a formula. 
\end{definition}

\begin{definition}
Let $\Phi$ be a finite set of formulae. A \defstyle{consistent game form for $\Phi$} is a game form 
$(\Act,\act,\psf (\Phi),\out)$ 
over the set of outcomes $\psf (\Phi)$ such that, for each action profile $\actprof$, $\out(\actprof)$ is a consistent set of formulae. A \defstyle{maximal consistent game form for $\Phi$} is a game form 
$(\Act,\act,\psf (\Phi),\out)$ 
over outcomes $\psf (\Phi)$ such that, for each action profile $\actprof$, $\out(\actprof)$ is a maximal consistent subset of $\Phi$.
\end{definition}

Note that, if $\Phi$ is a 
closed set of formulae, then a consistent game form for $\Phi$  is maximal if and only if for every action profile $\actprof$ the set $\out(\actprof)$ is a $\Phi$-atom.

 Given a strategic game form $G = (\Act,\act,\outcomes,\out)$, a coalition $C$ and action profiles $\actprof', \actprof$, we write $\actprof' \sim_C \actprof$ to state that $\actprof' \vert_C = \actprof \vert_C$. 

\begin{theorem}[One-step completeness]
\label{t:one-step-comp}
Let $\Gamma$ be a consistent modal one-step theory over a finite set of formulas $\Phi$ and assume that $\Phi$ contains all components of $\Gamma$, also contains $\overline{\psi}$ whenever $\psi \in \Phi$, and is closed under conjunctions (up to provable equivalence). 
Then there exists a maximal consistent game form 
$\gmod(\Gamma) = (\Act,\act,\psf (\Phi),\out)$ 
for $\Phi$
such that, for every goal assignment $\gamma$:
\begin{enumerate}
\item If $\brak{\gamma} \in \Gamma$, then there is a profile $\actprof \in \Pi_{a \in \Agt}\act_a$ such that for all $C$ in the support of $\gamma$ with $\gamma(C) = \nexttime \phi$ and all $\actprof' \sim_C \actprof$, we have $\phi \in \out(\actprof')$. 
\item If $\neg \brak{\gamma} \in \Gamma$, then for every profile $\actprof \in \Pi_{a \in \Agt}\act_a$ there is some $C$ in the support of $\gamma$, and some $\actprof' \sim_C \actprof$, for which  we have 
$\overline{\phi} \in \out(\actprof')$ where $\gamma(C) = \nexttime \phi$.   
\end{enumerate}
Furthermore, the size of $\Act$ is at most exponential in $\vert \Phi \vert$.
\end{theorem}

\begin{proof}
We may assume without loss of generality that, for every nexttime goal assignment $\gamma$ over $\Phi$, the set $\Gamma$ contains either $\brak{\gamma}$ or $\neg \brak{\gamma}$, since otherwise we can extend $\Gamma$,  using Lindenbaum's lemma, to a consistent (and still finite) set satisfying this assumption. 

We consider nexttime goal assignments over the set of  conjunctions of subsets of $\Phi$.
We say that such a goal assignment  $\gamma$ is \defstyle{deterministic} if $\gamma(\Agt)$ is (provably equivalent to) the conjunction of a maximal consistent subset of $\Phi$. 
We then say that a 
goal assignment $\gamma'$ is a \defstyle{strengthening} of $\gamma$ if, for all $C$ with $\gamma(C) = \X\phi$ and $\gamma'(C) = \X\phi'$, the formula $\phi' \to \phi$ is provable.   Note that every formula $\brak{\gamma}$ provably implies the disjunction of all formulas $\brak{\gamma'}$ where $\gamma'$ is a deterministic strengthening of $\gamma$ over $\Phi$. 
This follows by repeated applications of the axiom (GrandCoalition). 

If $\gamma$ is deterministic the we let \defstyle{$\mathsf{next}(\gamma)$} denote the maximal consistent set $\Psi$ for which $\gamma(\Agt) = \nexttime ( \bigwedge \Psi )$. Note that for any deterministic strengthening $\gamma'$ of a goal assignment  $\gamma$ over $\Psi$ and any $C$ for which $\gamma(C) = \X \varphi$, $\mathsf{next}(\gamma')$ must contain $\varphi$ as a conjunct. This is a consequence of axioms (Con), (Safe) and the assumption that $\mathsf{next}(\gamma')$ is maximal consistent. 

For technical convenience, in this proof we fix the enumeration  of $\Agt$ to be 
$\aga_0,...,\aga_{K-1}$,   
where $K = \vert \Agt \vert$. 
Given an agent $\aga \in \Agt$, we define $\gacta$ 
to be the set of all triples  
$(\gamma,f, k)$ such that $\gamma$  is a goal assignment with  $\brak{\gamma} \in \Gamma$, $0 \leq k < K$, and $f$ is a   function mapping each goal assignment 
$\gamma' : \psf(\Agt) \to \Phi$  to one of its deterministic strengthenings. To count the number of actions, this is the number of goal assignments $\gamma$ with $\brak{\gamma} \in \Gamma$, times the number of functions $f$ mapping goal assignments over $\Phi$ to deterministic strengthenings, times $K$. The number of goal assignments over $\Phi$ is $\vert \Phi \vert^{2^K}$, and more generally the number of goal assignments over conjunctions of subsets of $\Phi$ is $(2^{\vert \Phi \vert})^{2^K} = 2^{2^K \cdot \vert \Phi \vert}$. So the number of functions $f$ that can appear in an action is:
$$(2^{2^K \cdot \vert \Phi \vert})^{\vert \Phi \vert^{2^K}} = 2^{2^K \cdot \vert \Phi \vert^{2^K + 1}}$$
Since the number of agents is fixed, $2^K$ is a constant and the term $2^K \cdot \vert \Phi \vert^{2^K + 1}$ is a polynomial in $\vert \Phi \vert$. Hence the number of actions is exponential in $\vert \Phi \vert$.

{Note that 
$\gacta \neq \emptyset$ for all $\aga \in \Agt$ since there is at least one  goal assignment with $\brak{\gamma} \in \Gamma$ by the axiom (Triv), and $\brak{\gamma}$ is equivalent to the disjunction of its deterministic strengthenings, so one of these must also be in $\Gamma$.  Note also that the goal assigned to $\Agt$ by any $\brak{\gamma} \in \Gamma$ must be consistent by the axiom (Safe).} 
We set $\gAct = \bigcup_{\aga \in \Agt} \gacta$. 
Given an action profile 
$\actprof$ and $\aga \in \Agt$, if $\actprof_\aga = (\gamma,k,f)$ we write 
$\mathsf{vote}(\aga,\actprof) = \gamma$, $\mathsf{bet}(\aga,\actprof) = k$ and 
$\mathsf{choice}(\aga,\actprof,\gamma') = f(\gamma')$ for every goal assignment 
$\gamma'$ with $\brak{\gamma'} \in \Gamma$. We write 
$\mathsf{vote}(\actprof)$ for the voting profile mapping each 
$\aga \in \Agt$ to $\mathsf{vote}(\aga,\actprof)$. We define the \emph{voting winner} $\mathsf{win}(\actprof)$ to be player $\aga_i$ where $i$ is determined as follows:
\[
i := \left( \sum_{\aga \in \Agt} \mathsf{bet}(\aga,\actprof) \right) \; \text{mod} \; K
\]
Finally, we define the outcome of a given action profile $\actprof$ as follows:
\[
\out(\actprof) := \mathsf{next}\left( \mathsf{choice}(\mathsf{win}(\actprof), \actprof, \mathsf{merge}(\mathsf{vote}(\actprof)) \right)
\]
We will show that the game form $\gmod(\Gamma) = (\gAct,\gact,\psf (\Phi),\out)$ we have constructed satisfies the criteria listed in the statement of the theorem. First, we shall prove a rather technical auxiliary claim. Before going through its proof, the reader may want to skip ahead to see how the claim is used in the main argument. 
\begin{customclaim}{1}
Let $\varphi$ be any formula in $\Phi$, $C$ any coalition, and  let $\actprof$ be an action profile in $\gmod(\Gamma)$ such that for every action profile $\actprof' \sim_C \actprof$, we have $\varphi \in \out(\actprof')$.  Let $\gamma$ be any deterministic strengthening of $\mathsf{merge}(\mathsf{vote}(\actprof))$ such that $\brak{\gamma} \in \Gamma$ and $\gamma(\Agt) = \nexttime (\bigwedge \out(\actprof))$, and let $\gamma(C) = \nexttime \psi$. Then there exists a deterministic strengthening $\gamma'$ of $\gamma$ such that $\brak{\gamma'} \in \Gamma$, and:
$$\gamma'(C) = \nexttime(\psi \wedge \varphi).$$
\end{customclaim}

\append{
\paragraph{Proof of Claim 1:}
We need to distinguish two cases, for $C = \Agt$ and $C \neq \Agt$. We begin with the easier case where $C = \Agt$. In this case there is only one action profile $\actprof' \sim_{C} \actprof$, namely $\actprof$ itself. Our assumption thus gives $\varphi \in \out(\actprof)$. Further, we have $\gamma(C) = \gamma(\Agt) = \nexttime (\bigwedge \out(\actprof))$ by assumption. But, since $\psi \in \out(\actprof)$ the formula $\bigwedge \out(\actprof)$ is equal (up to provable equivalence) to $\psi \wedge \bigwedge \out(\actprof)$, which is provably equivalent to $\psi \wedge \varphi$ (because $\gamma'$ is a deterministic strengthening). Hence  we can set $\gamma' = \gamma$.

The case where $C \neq \Agt$ is more involved.  We assumed that $\brak{\gamma} \in \Gamma$.  By (Case), we have:
\[
\brak{\gamma[C \gass \X (\psi \wedge \varphi)]} \vee \brak{\gamma \vert_{C}[\Agt \gass \X\neg \varphi]} \in \Gamma
\]
We first show that  $\brak{\gamma \vert_{C}[\Agt \gass \X\neg \varphi]} \notin \Gamma$. Suppose the contrary, that $\brak{\gamma \vert_{C}[\Agt \gass \X\neg \varphi]} \in \Gamma$. Let $\gamma^*$ be an arbitrary deterministic strengthening of $\gamma \vert_{C}[\Agt \gass \X\neg \varphi]$ such that $\brak{\gamma^*} \in \Gamma$. Such strengthening must exist, since we recall that every formula of the form $\brak{\delta}$ provably implies the disjunction of all its deterministic strengthenings over $\Phi$. 
Now we define a new action profile $\actprof'$ as follows. First, pick an {arbitrary}
$\agc \notin C$, which exists since $C \neq \Agt$. For each $\aga \in C$ set $\actprof'_\aga = \actprof_\aga$. 
For each $\aga \notin C$ and $\aga \neq \agc$, set $\actprof'_\aga = (\gamma^\top,f,0)$  where $f$ is arbitrary (recall that $\gamma^\top$ is the trivial goal assignment with empty support). For $\agc$, set $\actprof'_\agc = (\gamma^\top,f,h)$ where the choice function $f$ chooses $\gamma^*$ whenever possible, and the bet $h$ is chosen so that the index of the player $\agc$ is equal to $\sum_{\aga \in C} \mathsf{bet}(\aga,\actprof) + h \ (\text{mod } K)$. This guarantees that $\agc$ will be the voting winner in $\actprof'$. Clearly $\actprof' \sim_{C} \actprof$. Furthermore, since $\gamma^*$ is a strengthening of $\gamma \vert_{C}[\Agt \gass \X\neg \varphi]$, and $\gamma$ is a strengthening of $\mathsf{merge}(\mathsf{vote}(\actprof))$ by assumption, it follows that $\gamma^*$ is a  deterministic strengthening of $\mathsf{merge}(\mathsf{vote}(\actprof'))$. This is because the only coalitions not mapped to $\X\top$ by  $\mathsf{merge}(\mathsf{vote}(\actprof'))$ are the ones contained in $C$, and for any such coalition $D$ we have $\mathsf{merge}(\mathsf{vote}(\actprof'))(D) = \mathsf{merge}(\mathsf{vote}(\actprof))(D)$. So, we get:
\[
\begin{aligned}
 \out(\actprof') & =  \mathsf{next}\left( \mathsf{choice}(\mathsf{win}(\actprof'), \actprof', \mathsf{merge}(\mathsf{vote}(\actprof')) \right) \\
 & = \mathsf{next}\left( \mathsf{choice}(\agc, \actprof', \mathsf{merge}(\mathsf{vote}(\actprof')) \right) \\
& = \mathsf{next}(\gamma^*).
\end{aligned}
\]
But $\neg \varphi \in \mathsf{next}(\gamma^*)$ since $\gamma^*$ is a strengthening of $\gamma \vert_{C}[\Agt \gass \X\neg \varphi]$. By consistency of $\mathsf{next}(\gamma^*)$, we have thus found an action profile $\actprof'$ such that $\actprof \sim_{C} \actprof'$ and $\varphi \notin \out(\actprof')$. This is a contradiction with our assumption on the action profile $\actprof$. 
Thus, we have proved $\brak{\gamma \vert_{C}[\Agt \gass \X\varphi]} \notin \Gamma$, as desired.
It follows that $\brak{\gamma[C \gass \X (\psi \wedge \varphi)]} \in \Gamma$. We then define: 
\[\gamma' := \gamma[C \gass \X (\psi \wedge \varphi)].\]
Thus, we have showed that $\brak{\gamma'} \in \Gamma$, $\gamma'$ is clearly a strengthening of $\gamma$, and it is deterministic since $\gamma'(\Agt) = \gamma(\Agt)$. This concludes the proof of the claim. 
\qed. 
}

\medskip
We now prove that the properties (1) and (2) listed in the theorem hold for the game form 
$\gmod(\Gamma) = (\gAct,\gact,\psf (\Phi),\out)$. 

\paragraph{Item (1):} Suppose $\brak{\gamma} \in \Gamma$. Let $\actprof$ be defined by letting all players vote for  $(\gamma, f, 0)$ where $f$ is an 
{arbitrary, fixed choice function}. 
Let $C$ be in the support of $\gamma$, where $\gamma(C) = \nexttime \varphi$ and let $\actprof' \sim_C \actprof$. 
 Since all players in $C$ vote for $\gamma$ in $\actprof'$, and the outcome 
 $\out(\actprof')$ is $\mathsf{next}(\gamma')$ for a deterministic strengthening of  
 $\mathsf{merge}(\mathsf{vote}(\actprof'))$, it follows that $\varphi \in \out(\actprof')$.

\paragraph{Item (2):} Suppose $\neg \brak{\gamma} \in \Gamma$, and let $\actprof$ be an arbitrary action profile. We want to show that there is some coalition $C$ and some action profile $\actprof' \sim_C \actprof$ such that 
$\overline{\phi} 
\in \out(\actprof')$, where $\gamma(C) = \varphi$. 

We will prove this by reductio ad absurdum. Suppose that for every coalition $C$ with $\gamma(C) = \nexttime \varphi$ and every action profile $\actprof' \sim_C \actprof$, we have 
$\overline{\phi}\notin \out(\actprof')$. This means that $\varphi \in \out(\actprof')$ since both $\overline{\phi}$
and $\varphi$ are in the closure of $\brak{\gamma}$  
and $\out(\actprof')$ is maximal consistent. 
Let us list all coalitions in the support of $\gamma$ as $C_1,...,C_m$. Let $\gamma_0$ denote the goal assignment $\mathsf{choice}(\mathsf{win}(\actprof), \actprof, \mathsf{merge}(\mathsf{vote}(\actprof))$. Then $\gamma_0(\Agt) = \nexttime (\bigwedge \out(\actprof))$. Furthermore $\brak{\gamma_0} \in \Gamma$ by definition of $\mathsf{choice}$, and $\gamma_0$ is a deterministic strengthening of $\mathsf{merge}(\mathsf{vote}(\actprof))$.  For each $i \in \{1,...,m\}$ let $\psi_i$ be the formula such that $\gamma(C_i) = \X \psi_i$, and let $\psi^0_i$ be the formula such that $\gamma_0(C_i) = \X \psi^0_i$. We define, for each $i \in \{0,...,m\}$, a deterministic goal assignment $\gamma_i$ such that:
\begin{itemize}
\item $\brak{\gamma_i} \in \Gamma$,
\item $\gamma_i$ is a deterministic strengthening of $\gamma_j$ for all $j < i$,
\item $\gamma_i(C_j) = \nexttime (\psi^0_j \wedge \psi_j)$ for all $j$ with $1 \leq j  \leq i$, and $\gamma_i(\Agt) = \nexttime (\bigwedge \out(\actprof))$.  
\end{itemize}  

The goal assignment $\gamma_0$ has already been defined, and we can extend the definition inductively to all $i$ by repeatedly applying Claim 1. Note that the induction hypothesis has been tailored so that Claim 1 applies at each inductive step.  

Now, consider the goal assignment  $\gamma_m$. We have $\brak{\gamma_m} \in \Gamma$. But by definition $\brak{\gamma_m}$ is a strengthening of $\brak{\gamma}$, hence $\brak{\gamma} \in \Gamma$, by Goal Monotonicity. 
Since we assumed that $\neg \brak{\gamma} \in \Gamma$, 
we have reached a contradiction with the consistency of $\Gamma$. 
This concludes the proof of item (2) and thus the proof of the theorem.
\end{proof}

\section{Completeness of \cga}
\label{sec:completeness} 

\subsection{Networks} 
\label{sec:Networks}

Throughout the rest of this section, we fix a finite, 
closed set $\Phi$ of \cga-formulae in normal form. 
To prove completeness, we will show how to construct a model for each (consistent) atom in $\mathsf{At}(\Phi)$, using the technique of \emph{networks}. The idea behind this technique is to construct a series of approximations to the satisfying model. At each finite stage of the construction, the current approximating network will generally have a number of \emph{defects}, each of which represents is some particular reason that the network cannot yet be regarded as a satisfying model. 
For example, consider a formula of the form $\brak{A \gass \alpha \until \beta, B\gass \always \chi}$. If this formula belongs to the label associated with some node $u$ in a network, then there must exist some strategy profile $\strprof$ such that, when restricted to a joint strategy for the coalition $A$, it ensures that every play generated by it eventually leads to a state where $\beta$ is in the label, and all states that are visited meanwhile have $\alpha$ in their labels; likewise, when $\strprof$ is restricted to a joint strategy for the coalition $B$, it ensures that on every play generated by it all states that are visited have $\chi$ in their labels. 
For the respective conditions for $\alpha$ and $\chi$ it suffices that they are satisfied \emph{locally}, at every step of the construction of the network. 
Thus, there are two types of conditions to be satisfied by the network: \emph{local conditions}, that can be ensured on-the-fly, i.e. the defects arising from their violation can be fixed step-by-step in the process of the construction and updates of the network; and \emph{eventualities}, which need to be taken special care of. In the example above, if it is not already the  case that every play generated by the joint strategy for $B$ obtained from $\strprof$ eventually leads to a state where $\beta$ is in the label, that creates a defect which we needs to be eventually fixed. To do so we first show that we can ``push'' the defect towards leaves in the network, in the sense that it suffices to fix the defect at each leaf in order to make sure each occurrence of the defect in the current network disappears. Next, we fix each defect associated with a leaf. To do that, we prove that for every atom that contains the formula $\brak{A \gass \alpha \until \beta, B\gass \always \chi}$, we can find a network whose root is labelled by that atom, and in which the occurrence of the formula at the root is not a defect (though it may appear as a defect elsewhere in the network). This network can then by ``plugged in'' at appropriate leaves in another network in order to fix a defect there. Once a specific occurrence of a defect is fixed it never reappears, but of course each round of the construction may introduce new defects, like new heads of a hydra appearing where one was previously cut off.  These new defects are then taken care of in the \emph{next} round, and so on. By taking a limit of the approximating series we obtain a \emph{perfect} network, i.e. a network with no defects. A ``truth lemma'' for perfect networks assures that we can view the network obtained in the limit as a model in which each formula attached to the root is true.

\begin{definition}
A \defstyle{$\Phi$-network} is a triple $\network = (T,L,\gform)$ such that:
\begin{itemize}
\item $T$ is a rooted, finitely branching directed tree, 

\item $L : T \to \mathsf{At}(\Phi)$ is a map that assigns to each node of $T$ an atom from $\mathsf{At}(\Phi)$. 

\item $\gform$ is a map that assigns to each non-leaf node $u$ of $T$ a game form 
$\gform(u) = (\Act^u,\act^u,T,\out^u)$, where  $\out^u$ is subject to the constraint that its codomain is the set of children nodes of $u$ in $T$. 
\end{itemize}
\end{definition}

\begin{definition}
A network $\network = (T,L,\gform)$ is said to be a \defstyle{sub-network} of $\network' = (T',L',\gform')$, written $\network \sqsubseteq \network'$, if:
\begin{itemize}
\item $T$ is a subgraph of $T'$ and the root of $T$ is also the root of $T'$,
\item If $u$ is any non-leaf node in $T$ then it has the same children in $T$ as in $T'$ and, furthermore, $\gform(u) = \gform'(u)$,
\item $L = L'\vert_T$.
\end{itemize}
\end{definition}

\begin{definition}
Given a $\Phi$-network $\network = (T,L,\gform)$, a \defstyle{marking} of $\network$ is a map $\mrk$ from $T$ to the powerset of $\Phi$ such that $\mrk(v) \subseteq L(v)$ for all $v \in T$. (In particular, note that $L$ itself is a marking of $\network$.) Given a marking $\mrk$ of $\network$, a nexttime goal assignment $\gamma$ such that  $\brak{\gamma} \in \Phi$, and a non-leaf node $u \in T$ with $\gform(u) = (\Act,\act,T,\out)$, we say that the marking $\mrk$ \defstyle{verifies the goal assignment $\gamma$ at $u$} if there is a strategy profile $\strprof \in \Pi_{a \in \Agt}$ such that, for every $C$ in the support of $\gamma$ such that $\gamma(C) = \nexttime \psi$ and for every strategy profile $\strprof'$ with $\strprof' \sim_C \strprof$, we have $\psi \in \mrk(\out(\strprof',u))$. We say that $\mrk$ \defstyle{refutes the goal assignment $\gamma$  at $u$} if for every strategy profile $\strprof \in \Pi_{a \in \Agt}$ there is some $C$ in the support of $\gamma$ with $\gamma(C) = \nexttime \psi$ and some strategy profile $\strprof'$ with $\strprof' \sim_C \strprof$ such that $\overline{\psi} \in \mrk(\out(\strprof',u))$. 
\end{definition}

\begin{definition}
A  $\Phi$-network $\network = (T,L,\gform)$ is said to be \defstyle{one-step coherent} if, for every non-leaf node $u \in T$ such that $\gform(u) = (\Act,\act,T,\out)$, the marking $L$ verifies every nexttime goal assignment $\gamma$ such that $\brak{\gamma} \in L(u)$ and refutes every nexttime goal assignment $\gamma$ such that $\neg \brak{\gamma} \in L(u)$. 
\end{definition}

\subsection{Eventualities and defects}
\label{sec:EventualitiesAndDefects}

\begin{definition}
A \defstyle{$\typeone$-eventuality} is a formula $\brak{\gamma}$ where $\gamma \in \typeone$. 
A \defstyle{$\typetwo$-eventuality} is a formula of the form $\neg \brak{\gamma}$ where $\gamma \in \typetwo$.  
\end{definition}

\begin{definition}
Let $\brak{\gamma}$ be a $\typeone$-eventuality, where $\gamma$ is a goal assignment for the family $\fac = \{C_1,...,C_n,D_1,...,D_m\}$  or $\fac = \{C_1,...,C_n\}$ defined by: 
\[
\gamma(C_1) = \alpha_1 \until \beta_1, ...,\ \gamma(C_n) = \alpha_n \until \beta_n
\]
and  (if $m > 0$) 
\[
\gamma(D_1) = \always \chi_1, ...,\ \gamma(D_m) = \always \chi_m. 
\]
Let $\network = (T,L,\gform)$ be a one-step coherent network. 
Given a node $u \in T$, we say that  \defstyle{$\brak{\gamma}$ is partially fulfilled in $0$ steps  at $u$ in $\network$} if there is some $i \in \{1,...,n\}$ such that $\beta_i \wedge \brak{\gamma \setminus C_i} \in L(u)$. For any natural number $k \geq 0$, we say that  \defstyle{$\brak{\gamma}$ is  partially  fulfilled in $k + 1$ steps at $u$} if it is either partially fulfilled in $0$ steps, or $u$ is a non-leaf node and the following  conditions hold:
\begin{itemize}
\item $\alpha_i \in L(u)$ for all $i \in \{1,...,n\}$,
\item $\chi_j \in L(u)$ for all $j \in \{1,...,m\}$,
\item there is a marking $\mrk$ that verifies $\brak{\diffof{\gamma}}$ at $u$  and is such that for all $v \in T$ such that $v$ is a child of $u$,  $\brak{\gamma}$ is partially fulfilled  in $k$ steps at $v$ whenever $\brak{\gamma} \in \mrk(v)$. 
\end{itemize}
Lastly, we say that  \defstyle{$\brak{\gamma}$ is partially fulfilled at $u$} if it is partially fulfilled in $k$ steps at $u$ for some $k\geq 0$. 

\end{definition}

\begin{definition}
Let $\neg\brak{\gamma}$ be an eventuality in $\typetwo$, where $\gamma$ is the goal assignment for the family $\fac = \{D_1,...,D_m\}$ defined by: 
\[
\gamma(D_1) = \always \chi_1,..., \gamma(D_m) = \always \chi_m. 
\]
Let $\network = (T,L,\gform)$ be a one-step coherent network. Given a node $u \in T$, we say that  \defstyle{$\neg \brak{\gamma}$ is partially fulfilled  in $0$ steps at $u$ in $\network$} if there is some $i \in \{1,...,n\}$ such that $\overline{\chi_i} \in L(u)$. 

For any natural number $k \geq 0$, we say that  \defstyle{$\neg\brak{\gamma}$ is partially fulfilled  in $k + 1$ steps at $u$} if it is either partially fulfilled  in $0$ steps, or $u$ is a non-leaf node and there exists a marking $\mrk$ that refutes $\brak{\diffof{\gamma}}$ at $u$,  and such that for all $v \in T$ 
{such that $v$ is a child of $u$},  $\neg \brak{\gamma}$ is partially fulfilled  in $k$ steps at $v$ whenever $\neg \brak{\gamma} \in \mrk(v)$. 
Lastly, we say that  \defstyle{$\neg\brak{\gamma}$ is partially fulfilled at $u$} if it is partially fulfilled in $k$ steps at $u$ for some $k\geq 0$. 
\end{definition}

\begin{definition}
A \defstyle{defect} of a network $\network = (T,L,\gform)$ is a pair $(u,\phi)$ such that $u \in T$, $\phi \in L(u)$ is an eventuality which is not partially fulfilled at $u$. 
\end{definition}

\begin{proposition}
\label{p:stayfinished}
Let $\network \sqsubseteq \network'$ and let $(u,\varphi)$ be a defect of $\network'$. If $u$ belongs to $\network$, then $(u,\varphi)$ is a defect of $\network$, as well. 
\end{proposition}

\begin{proof}
A trivial induction on $k$ shows that, if an eventuality of any of the two types is partially fulfilled  in $k$ steps at $u$ in $\network$, then it is partially fulfilled  in $k$ steps at the same node in $\network'$ as well. 
\end{proof}

\begin{definition}
A network is said to be \defstyle{perfect} if it is one-step coherent, has no leaves, and no defects. 
\end{definition}

\begin{definition}
Given a perfect network  $\network = (T,L,\gform)$ we define a game model $\gmod(\network) = (\states,\Act,\gmap,V)$ 
as follows. We take $\states$ to be the set of all nodes in $T$, and $\gmap = \gform$.  
Finally, we set $V(p) = \{v \in T \mid p \in L(v)\}$. We call $\gmod(\network)$ the  \defstyle{induced model of the network} $\network$. 
\end{definition}
{The following proposition will relate truth sets $\tset{\varphi}_{\gmod(\network)}$ of formulas in the induced model of a network to the set of  nodes $v$ with $\varphi \in L(v)$. To clearly distinguish the latter from the former, we introduce the following notation:
$$\lset{\varphi}{\network} := \{v \in T \mid \varphi \in L(v)\}$$}
{For the proof of the following proposition, see the appendix.}

\begin{proposition}
\label{prop:network-sat} 
Every $\Phi$-atom that is the label of some node in a perfect $\Phi$-network 
$\network = (T,L,\gform)$ is true at the respective state of the model 
$\gmod(\network)$ induced by that network.
\end{proposition}

\subsection{Constructing a perfect network}
\label{sec:PerfectNetwork}

\subsubsection{Step 1: extending leaves in coherent networks}

\begin{proposition}
\label{p:remove-leaves}
Let $\network$ be any finite, one-step coherent network, and let $u$ be a leaf in $\network$. Then there exists a finite and one-step coherent network $\network'$ such that $\network \sqsubseteq \network'$ and such that $u$ is not a leaf in $\network'$.
\end{proposition}
\begin{proof}
Let $\Gamma = \{\brak{\gamma} \mid \brak{\gamma} \in L(u)\} \cup  \neg\{\brak{\gamma} \mid \neg \brak{\gamma} \in L(u)\}$.  This is a consistent modal one-step theory, so let $(\Act,\act,\psf (\Phi),\out)$ be the maximal consistent game form $\gmod(\Gamma)$ provided by Theorem \ref{t:one-step-comp}. We construct the network $\network' = (T',L',\gform')$ as follows.  For each atom $\Psi$ in the image of the function $\out$ (which is always non-empty), add a  new successor $v_\Psi$ to $u$, and set $L(v_\Psi) = \Psi$. Let $S$ denote the set of successors of $u$ added in this manner. We construct a new game form 
$\gform'(u) = (\Act,\act,S,\out')$ by setting, for each action profile $\actprof$, $\out'(\actprof) := \{v_\Psi\}$ where $\Psi = \out(\actprof)$. 
This completes the definition of $\network'$. It is clear that $\network \sqsubseteq \network'$, and the conditions given in Theorem \ref{t:one-step-comp} directly entail (by design) that the network $\network'$ is one-step coherent. Since at least one successor was added to $u$, this node is no longer a leaf in $\network'$.   
\end{proof}

\subsubsection{Step 2: pushing defects towards leaves}

\begin{proposition}
\label{p:push}
Let $\mathcal{N}$ be a finite, one-step coherent network and let $(u,\varphi)$ be a defect of $\network$. Then there exists a 
set $\{v_1,...,v_k\}$ of leaves in $\network$ such that:
\begin{itemize}
\item For each $i \in \{1,...,k\}$, the pair $(v_i,\varphi)$ is a defect of $\network$, and
\item For any one-step coherent network $\network'$ such that $\network \sqsubseteq \network'$, if $(u,\varphi)$ is still a defect in $\network'$ then there is some $i \in \{1,...,k\}$ such that $(v_i,\varphi)$ is a defect of $\network'$.
\end{itemize}  
\end{proposition}

\begin{proof}
We focus on the case of a type $\until$ eventuality; the case of type $\always$  eventualities is very similar. Let $\brak{\gamma}$ be a type $\until$ eventuality.  We say that a defect $(u,\brak{\gamma})$ \defstyle{one-step generates} a defect $(v,\brak{\gamma})$ if $v$ is one of the children of $u$, and $(v,\brak{\gamma})$ is a defect of $\network$. We then say that a defect $(u,\brak{\gamma})$ \defstyle{generates} a defect $(v,\brak{\gamma})$ if $(v,\brak{\gamma})$ is a successor of $(u,\brak{\gamma})$ with respect to the transitive closure of the one-step generation relation.  

Now suppose that $(u,\brak{\gamma})$ is a defect of $\network$.  We claim that the set of leaves $l$ such that $(l,\brak{\gamma})$ is a defect generated by the defect $(u,\brak{\gamma})$ satisfies the conditions of the proposition. The first condition holds by definition. To prove this, consider any  one-step coherent  network $\network'$ such that $\network \sqsubseteq \network'$. We show that for every defect $(w,\brak{\gamma})$ of $\network$, if $(w,\brak{\gamma})$ is still a defect of $\network'$, then the same holds for some defect $(v,\brak{\gamma})$ of $\network$ that is one-step generated by $(w,\brak{\gamma})$. By repeatedly applying this claim, starting with the defect $(u,\brak{\gamma})$, we eventually reach a leaf $l$ such that $(l,\brak{\gamma})$ is a defect generated by the defect $(u,\brak{\gamma})$, and is still a defect in $\network'$. 

So, let  $(w,\brak{\gamma})$ be a non-leaf defect of $\network$, such that $(w,\brak{\gamma})$ is still a defect of $\network'$. Suppose, for a contradiction, that for all the children $v$ of $w$, $(v,\brak{\gamma})$ is \emph{not} a defect of $\network'$. This means that for all children $v$ of $w$ in $\network$, and hence for all children of $w$ in $\network'$ since $\network \sqsubseteq \network'$, 
there is some $k_v$ for which the eventuality $\brak{\gamma}$ is 
partially fulfilled in $k_v$ 

steps at $v$ in $\network$. Let $K$ be the maximum of these numbers $k_v$, which exists since the set of successors of $w$ is finite. By one-step coherence of the network $\network'$ it follows that $\brak{\gamma}$ is
 partially fulfilled in $K+1$ steps at $w$ in $\network'$, witnessed by the labelling function $L$ of $\network$ regarded as a marking of $\network'$. Thus, we have reached a contradiction, which completes the proof.  
\end{proof}

\subsubsection{Step 3: removing defects}

We now show how to remove defects from a network:

\begin{proposition}
\label{p:remove-defects}
Let $(u,\varphi)$ be a defect of some finite, one-step coherent network $\network$. Then there exists a finite, one-step coherent network $\network'$ such that $\network \sqsubseteq \network'$, and such that $(u,\varphi)$ is not a defect of $\network'$. 
\end{proposition}
\begin{proof}
By Proposition \ref{p:push}, we may assume w.l.o.g. that the defect $(u,\varphi)$ is such that $u$ is a leaf: if we can show how to remove the defect $\varphi$ at a single leaf, then, clearly, we can repeat the procedure to remove $\varphi$ at each leaf in the set $\{v_1,...,v_k\}$. {(Note that our procedure for removing a defect at a single leaf $v$ given below will not affect any other leaves, i.e. each leaf in the original network besides $v$ will still be a leaf in the new network.)}
Combined with  Proposition \ref{p:push} this proves the result. 

So, suppose that $(u,\varphi)$ is a defect and $u$ is a leaf. It is sufficient to show that there is a finite, one-step coherent network $\network''$ in which the root has the same label as $u$ in $\network$, and in which the eventuality $\varphi$ is partially fulfilled. We can then simply identify the root of the network $\network'$ with the leaf $u$ in $\network$ to form a finite, one-step coherent network $\network'$ such that $\network'' \sqsubseteq \network'$ and $\network \sqsubseteq	 \network'$. By Proposition \ref{p:stayfinished}, the eventuality $\varphi$ is partially fulfilled at $u$ in $\network'$.

Consider the $\Phi$-atoms $\Psi$ such that $\varphi \in \Psi$ and there exists a finite, one-step coherent network in which the root is labelled by $\Psi$ and the eventuality $\varphi$ is {partially}  fulfilled. 
Let $\delta$ be the disjunction of all conjunctions of the form $\bigwedge \Psi$ for all   
such $\Phi$-atoms $\Psi$. 
(This is well-defined since the set of all such conjunctions is finite, as long as we disallow conjunctions with redundant multiple occurrences of the same conjunct.) The result then follows from the following claim, which is proved in the appendix. 
\begin{innercustomclaim}
$\vdash \varphi \to \delta$.
\end{innercustomclaim}
\end{proof}

\subsubsection{Final step: putting everything together}

The following proposition is proved in the appendix.
\begin{proposition}
\label{prop:network-extension} 
Let $\network$ be a finite, one-step coherent $\Phi$-network. Then there exists a finite, one-step coherent $\Phi$-network $\network'$ such that:
\begin{enumerate}
\item $\network \sqsubseteq \network'$,
\item no leaf of $\network$ is a leaf of $\network'$,
\item no defect of $\network$ is a defect of $\network'$.
\end{enumerate}
\end{proposition}

\begin{proposition}
\label{prop:network-limit}
Every  $\Phi$-atom is the label of the root of some perfect  $\Phi$-network.
\end{proposition}

\begin{proof}
Take any $\Phi$-atom $\Psi$. We construct an infinite chain of finite one-step coherent $\Phi$-networks $\network_0 \sqsubseteq \network_1 \sqsubseteq...$ inductively as follows. 

We start with $\network_0 = (T_0,L_0,\gform_0)$ where $T_0 = \{u_0 \}$ is a singleton, 
$L_0(u_0) = \Psi$ and $\gform_0$ is empty (there are no non-leaf nodes). 
This is trivially one-step coherent. 

Suppose, we have constructed the finite one-step coherent $\Phi$-networks  
$\network_0 \sqsubseteq \network_1 \sqsubseteq... \network_n$. 
 Then we apply Proposition \ref{prop:network-extension} to construct a  finite, one-step coherent network $\network_{n+1}$ such that  $\network_n \sqsubseteq \network_{n+1}$, no leaves in $\network_n$ are still leaves in $\network_{n+1}$, and no defects in $\network_n$ are still defects in $\network_{n+1}$. 

\smallskip
Finally, we construct the network $\network$ as union of the chain $\network_0 \sqsubseteq \network_1 \sqsubseteq...$.  Clearly, it is still one-step coherent and has no leaves and no defects, i.e. it is perfect. 
\end{proof}

We can now state and prove the completeness theorem. 

\begin{theorem}[Completeness of $\textsf{Ax}_{\cga}$]
\label{thm:completeness}
Let  $\Gamma$ be a finite $\textsf{Ax}_{\cga}$-consistent set of \cga-formulae. 
Then $\Gamma$ is satisfied in some concurrent game model. 
\end{theorem}

\begin{proof}
Let $\Phi$ be the extended Fischer-Ladner closure of  $\Gamma$ and let  
$\Gamma^*$ be a $\Phi$-atom containing $\Gamma$ (which exists, by a standard version of Lindenbaum's lemma). 
By Proposition \ref{prop:network-limit}, $\Gamma^*$  is the label of the root of some perfect  $\Phi$-network. Then, by Proposition \ref{prop:network-sat},  
$\Gamma^*$ is true at the respective state of the model $\gmod(\network)$ induced by that network.
\end{proof}

\section{Finite model property and decidability }
\label{sec:FMP}
In this section we show finite model property and decidability for our logic \cga. Since we have a truth-preserving and effective translation of the language $\langcga$ into the fixpoint logic $\xlangcga_\mu$, it suffices to prove finite model property and decidability for the latter. Here, we will avail ourselves of some abstract results from the literature on coalgebraic modal fixpoint logic. In particular, a general bounded-size model property for coalgebraic $\mu$-calculi was proved in \cite{fontaine2010automata}, and since  $\xlangcga_\mu$ is an instance of coalgebraic $\mu$-calculus, we are almost done. There is one subtlety that we need to deal with, concerning the notion of ``finiteness'' of a model. There are two distinct notions of ``finite model'' that we may consider:
\begin{definition}
Let $\gmod = (\states,\Act,\gmap,\out,V)$ be a concurrent game model.  We say that $\gmod$ is \emph{state-finite} if $\states$ is a finite set. We say that $\gmod$ is \emph{action-finite} if $\Act$ is a finite set. We say that $\gmod$ is \emph{finite} if it is both state-finite and action-finite. 
\end{definition}

We get the following ``state-finite model property'', as a direct corollary of the general finite model theorem from \cite{fontaine2010automata}:
\begin{theorem}
Any satisfiable formula of $\xlangcga_\mu$ is satisfiable in a state-finite model.
\end{theorem}
However, what we want is a proper finite model property. We can obtain this with a little bit of extra work. First, we obtain the following ``action-finite model property'':
\begin{lemma}
\label{l:afmp}
Any satisfiable formula of $\langcga$ is satisfiable in an action-finite model.
\end{lemma} 
\begin{proof}
Since the size of the set of actions in the game forms constructed in the proof of the one-step completeness theorem (Theorem \ref{t:one-step-comp}) has an exponential upper bound that depends on the size of a finite consistent modal one-step theory, and since the  
our construction of a model for a consistent \langcga-formula can easily be seen to provide an action-finite model.  
\end{proof}

\begin{theorem}[Finite model property]
Any satisfiable formula of $\langcga$ is satisfiable in a finite model. 
\end{theorem}
\begin{proof}
Let $\varphi$ be a satisfiable formula of $\langcga$. By soundness of our proof system  $\varphi$ is consistent, and hence is satisfiable in a model $\gmod = (\states,\Act,\gmap,\out,V)$ where the set $\Act$ is finite. Hence the corresponding equivalent formula $\varphi'$ in $\xlangcga_\mu$ is satisfiable in this model too. 

But the frame $(\states,\Act,\gmap,\out)$ can be equivalently represented as a coalgebra  $f : \states \to \gamefun^{\Act} \states$ for a functor $\gamefun^{\Act}$ in which the set of actions $\Act$ is explicitly encoded, so the triple $(\states,f,V)$ is a $\gamefun^{\Act}$-model in which $\varphi'$ is satisfiable. By the finite model property theorem 
proved in \cite{fontaine2010automata}, $\varphi'$ is satisfiable in a $\gamefun^{\Act}$-model $(\states',f',V')$ for which $\states'$ is finite. This model can equivalently be represented as a concurrent game model $(\states',\Act,\gmap',\out')$ in which $\varphi'$ is satisfiable, hence also $\varphi$. Since $\states'$ is finite and $\Act$ is finite, this is a finite model for $\varphi$.
\end{proof}

Together with our completeness result for \cga, this implies:

\begin{theorem}
\label{t:cga-is-decidable}
The satisfiability problem for $\cga$ is decidable. 
\end{theorem}

Note that there is no need for an explicit bound on the size of a satisfying finite model for this result; completeness for \cga ensures that we can computably enumerate the valid formulas of $\langcga$, and the finite model property means that we can computably enumerate the non-valid formulas as well. Hence, the set of valid formulas is a computable set. 

Since the logic $\cga^+$ also embeds into $\xlangcga_\mu$, we get a further corollary:
\begin{theorem}
The logic $\cga^+$ has the finite model property.
\end{theorem}

The proof of Theorem \ref{t:cga-is-decidable} took a detour via finite model property for the language $\xlangcga_\mu$. We shall now look closer at the latter language and obtain an upper bound on the complexity of satisfiability for $\xlangcga_\mu$, under the assumption that the set $\Agt$ of agents is fixed.  Decidability of $\cga^+$ will then follow as a corollary. To obtain these results, we will require a closer analysis of the one-step satisfiability problem.  
\begin{definition}
Let $V$ be a fixed set of  proposition variables. 
The set of \defstyle{positive one-step formulas} over $V$ is generated by the following grammar:
$$\brak{\gamma_0} \mid \neg \brak{\gamma_1} \mid \varphi \wedge \varphi \mid \varphi \vee \varphi$$
where $\gamma_0$ is a goal assignment such that for each coalition $C$ in the support of $\gamma_0$
we have $\gamma_0(C) = \X p$ for some $p \in V$, called a \defstyle{positive goal assignment}, and $\gamma_1$ is a goal assignment such that for each coalition $C$  in the support of $\gamma_0$
we have $\gamma_1(C) = \X \neg p$ for some $p \in V$, called a \defstyle{negative goal assignment}. 

A \defstyle{positive one-step sequent} over $V$ is a finite set $\Gamma$ of formulas, each of which is either of the form $\brak{\gamma_0}$ where $\gamma_0$ is a positive goal assignment, or of the form $\neg \brak{\gamma_1}$ where $\gamma_1$ is a negative goal assignment.  
 \end{definition}

Let $\Gamma$ be a positive one-step sequent over $V$. A \defstyle{redistribution} of $\Gamma$ is a tuple $(C_1,\gamma_1,....,C_n,\gamma_n)$ where $C_1,...,C_n$ is a set of pairwise disjoint coalitions and each $\gamma_i$ is a positive goal assignment with $\brak{\gamma_i} \in \Gamma$. 
The intuition behind this notion is as follows: suppose that $\gamma_1,...,\gamma_n$ are positive goal assignments with $\brak{\gamma_i} \in \Gamma$ for each $i \in \{1,...,n\}$. Then according to the sequent $\Gamma$, there are action profiles $\actprof_1,...,\actprof_n$ such that for each $\actprof_i$, every coalition $C$ ensures its goal formula $\gamma_i(C)$ with respect to the action profile $\actprof_i$. But then, if $C_1,...,C_n$ are any pairwise disjoint coalitions, the coalition $C_1 \cup ... \cup C_n$ has a joint action in which each coalition $C_i$ ensures its goal formula $\gamma_i(C_i)$, by playing in accordance with the action profile $\actprof_i$. In other words, the action profiles $\actprof_1,...,\actprof_n$ can be combined into a joint action for $C_1 \cup ... \cup C_n$ be restricting each action profile $\actprof_i$ to the coalition $C_i$, and then ``gluing together'' the resulting joint actions into a joint action for $C_1 \cup ... \cup C_n$. The set of redistributions of $\Gamma$ can be viewed as a set of notations, for each of the different ways that action profiles corresponding to the positive goal assignments occurring in $\Gamma$ can be combined in this manner.

Note that a redistribution $(C_1,\gamma_1,....,C_n,\gamma_n)$ of $\Gamma$ can be uniquely represented by a map: 
$$f :  \psf \Agt \to \Gamma \cup \{*\}$$
defined by setting $f(C_i) = \brak{\gamma_i}$ and $f(B) = *$ for $B \notin \{C_1,...,C_n\}$. Hence, 
there are at most $(\vert \Gamma \vert + 1)^{k}$ redistributions of $\Gamma$, where $k = 2^{\vert \Agt \vert}$.  Recall that the set $\Agt$ is fixed, 
so $2^{\vert \Agt \vert}$ is a constant, giving a polynomial bound on the number of redistributions.

A goal assignment $\gamma$ with $\brak{\gamma} \in \Gamma$ or $\neg \brak{\gamma} \in \Gamma$ we be called a \defstyle{relevant goal assignment} for $\Gamma$.  Given a redistribution $R = (C_1,\gamma_1,...,C_n,\gamma_n)$, we let $F(R)$ denote the set: 
$$\{p \mid \exists i, B \subseteq C_i:\;\gamma_i(B) = \X p\}.$$

Intuitively, $F(R)$ is the set of goals forced by coalitions that have formed in the action profile represented by the redistribution $R$.
Given a triple $(R,\gamma',C')$ consisting of a redistribution $R = (C_1,\gamma_1,...,C_n,\gamma_n)$, a negative relevant goal assignment $\gamma'$, and a coalition $C'$ such that $\gamma'(C') = \X \neg q$, we denote by $F(R,\gamma',C')$ the set:  
$$\{p \mid \exists i, B \subseteq C_i \cap C': \; \gamma_i(B) = \X p\} \cup \{q\}.$$
Intuitively, the set $F(R,\gamma',C')$ is the minimal set of variables that have to be in the outcome of a strategy profile $\actprof$, in which each coalition $C_i \in \{C_1,...,C_n\}$ acts towards its goal with respect to the goal assignment $\gamma_i$, and at the same time the players in $\Agt\setminus C'$ act  to block the goal of $C'$ according to the goal assignment $\gamma'$. 

\begin{definition}
Let $V$ be a given set of proposition variables. A \defstyle{satisfiability constraint} over $V$ is a subset of $\psf(V)$. Given a positive one-step sequent $\Gamma$ over $V$ and a satisfiability constraint $\mathcal{S}$ over $V$, we say that $\Gamma$ is $\mathcal{S}$\defstyle{-satisfiable} if there exists a game form $\gform = 
(\Act,\act,\psf(V),\out)$ such that:
\begin{itemize}
\item If $\brak{\gamma} \in \Gamma$ then there exists an action profile $\actprof$ such that, for every coalition $C$, if $\gamma(C) = \X  p$ then $p \in \out(\actprof')$ for every $\actprof'\sim_C \actprof$.
\item  If $\neg \brak{\gamma} \in \Gamma$ then for every action profile $\actprof$, there is some coalition $C$ and some $\actprof' \sim_C \actprof$ such that $p \in \out(\actprof')$, where $\gamma(C) = \X \neg p$. 
\item For every action profile $\actprof$, there is some $V'  \in \mathcal{S}$ with $\out(\actprof) \subseteq V'$.
\item For every $V' \in \mathcal{S}$ there is some action profile $\actprof$ with $\out(\actprof) \subseteq V'$.
\end{itemize}
Given a game form $\gform$ satisfying these conditions, we say that $\gform$ $\mathcal{S}$\defstyle{-satisfies} $\Gamma$. 
\end{definition}

The following proposition reduces the question of whether a positive one-step sequent $\Gamma$ is $\mathcal{S}$-satisfiable to checking of a few relatively simple combinatorial conditions on the set $\mathcal{S}$. The proof is in the appendix.

\begin{proposition}
\label{p:c-sat}
Given a set $V$ of variables, a satisfiability constraint $\mathcal{S}$ over $V$ and a positive one-step sequent $\Gamma$ over $V$,  $\Gamma$ is $\mathcal{S}$-satisfiable if and only if:
\begin{enumerate}
\item For each redistribution $R$  of $\Gamma$, there is $Z \in \mathcal{S}$ with $F(R) \subseteq Z$.
\item For each redistribution $R$ of $\Gamma$ and each negative relevant goal assignment $\gamma'$, there is some $C' \subseteq \Agt$ such that:
\begin{itemize}
\item either $C' = \Agt$ and $q \in Z$ for every $Z \in \mathcal{S}$, where $\gamma'(\Agt) = \X\neg q$, or
\item $C' \neq \Agt$ and there is some $Z \in \mathcal{S}$ such that $F(R,\gamma',C') \subseteq Z$. 
\end{itemize}
\end{enumerate}
\end{proposition}

To illustrate this proposition, let $\Agt = \{\aga,\agb\}$, $V = \{p,q,r\}$ and let: $$\Gamma = \{\brak{\{\aga\} \gass \X p},\brak{\{\agb\} \gass \X q }, \neg \brak{\{\agb\}\gass \X \neg r}\}$$
Then $\Gamma$ is $\mathcal{S}$-satisfiable for $\mathcal{S} = \{\{p,q\},\{q,r\}\}$. However, $\Gamma$ is not $\mathcal{S}$-satisfiable for $\{\{p,q\},\{p,r\}\}$. Consider the redistribution $R$ in which $\aga$ ``votes'' for $\{\brak{\{\aga\} \gass \X p}$ and $\agb$ votes for $\brak{\{\agb\} \gass \X q }$. Then $\aga$ forces $p$ on their own, and $\agb$ forces $q$ on their own, so there should be a possible outcome containing both $p$ and $q$ -- which there is. But $\agb$ should not be able force $\neg r$ on their own since $\neg \brak{\{\agb\}\gass \X \neg r} \in \Gamma$, so there should be some outcome containing $q$, which $\agb$ is forcing on their own, but also $r$. Formally, this is captured by $F(R,\gamma,\{\agb\}) = \{q,r\}$, where $\brak{\gamma} = \brak{\{\agb \}\gass \X \neg r}$.

\begin{proposition}
\label{p:polytime}
Let $V$ be a fixed set of proposition variables and $\mathcal{S}$ a fixed satisfiability constraint. Then the problem of checking whether a given positive one-step sequent $\Gamma$ is $\mathcal{S}$-satisfiable is solvable in polynomial time.  
\end{proposition}

\begin{proof}
Easy consequence of Proposition \ref{p:c-sat}: we only have to check the conditions (1) and (2) on $\mathcal{S}$ for each redistribution of $\Gamma$ and each negative relevant goal assignment. Recalling that there are at most polynomially many redistributions of $\Gamma$, the result follows. 
\end{proof}

We extend the notion of $\mathcal{S}$-satisfiability to arbitrary positive one-step formulas over $V$ in the obvious manner. 
Let $V$ be a set of  proposition variables and $\mathcal{S} \subseteq V$ a satisfiability constraint. The \defstyle{one-step satisfiability problem} for $V,\mathcal{S}$ is: given a positive one-step formula $\varphi$ over $V$, decide whether $\varphi$ is $\mathcal{S}$-satisfiable.

\begin{proposition}
\label{p:exptime}
The one-step satisfiability problem for given $V,\mathcal{S}$ can be solved in nondeterministic polynomial time (measured in the length of a given formula $\varphi$).
\end{proposition}

\begin{proof}
The procedure for checking $\mathcal{S}$-satisfiability of $\varphi$ can be done as follows: we keep in memory a list of formulas $L$, beginning with the singleton list $[\varphi]$. If the leftmost formula in the memory $L$ that is not ``atomic'', i.e. not of the form $\brak{\gamma}$ or $\neg \brak{\gamma}$ for some $\gamma$, has  $\vee$  as main connective, then we guess one of the disjuncts non-deterministically and replace the disjunction by the guessed disjunct. If the main connective of the leftmost non-atomic formula is $\wedge$ then we remove that formula and instead add each conjunct to the list. Clearly this reduction can only go on for a number of steps that is bounded by the length of the formula $\varphi$, and once each formula in the list is atomic, we apply Proposition \ref{p:polytime} to check whether the one-step sequent read off from the list is $\mathcal{S}$-satisfiable.     
\end{proof}

It follows from a  general result in \cite{fontaine2010automata} that, if the one-step satisfiability problem for fixed $V,\mathcal{S}$ is solvable in exponential time, then the satisfiability problem for the full $\mu$-calculus language is in double exponential time.  From Proposition \ref{p:exptime} we thus get:
\begin{theorem}
\label{t:mu-dexptime}
The satisfiability problem for $\xlangcga_\mu$ is in $\textrm{2ExpTime}$. 
\end{theorem}
As a corollary, we get:
\begin{theorem}
The logic $\cga^+$, (hence, also \cga) is decidable in $\textrm{3ExpTime}$.
\end{theorem}
\begin{proof}
The translation from $\cga^+$ to $\xlangcga_\mu$ increases the size of a formula by at most a single exponential, so the result follows from Theorem \ref{t:mu-dexptime}.
\end{proof}
We do not know whether the $\textrm{2ExpTime}$ bound given by Theorem \ref{t:mu-dexptime} is tight. Many variants of the $\mu$-calculus share the same $\textrm{ExpTime}$-complete complexity of the satisfiability problem as the standard modal $\mu$-calculus, as demonstrated in  \cite{cirstea2011exptime}. However, the conditions under which the general $\textrm{ExpTime}$-bound of \cite{cirstea2011exptime} holds are not trivial to verify; in fact we believe they may fail for the language $\xlangcga_\mu$. Furthermore, even if the satisfiability problem for $\xlangcga_\mu$ happens to be in $\textrm{ExpTime}$, this bound does not immediately transfer to the language $\langcga$, since our translation of $\langcga$ into $\xlangcga_\mu$ can produce an exponential blowup in the formula size. 
In summary, the $\textrm{2ExpTime}$-bound for satisfiability in $\xlangcga_\mu$ stated above is the best we can currently claim for sure about the complexity of that satisfiability problem.  The problem of determining the precise complexity of satisfiability, both for $\langcga$ and for $\xlangcga_\mu$, may very well turn out to be a difficult one. We leave this as a question for future research.

\section{Concluding remarks}
\label{sec:concluding}

The present work falls in the line of research employing formal logical methods for modelling, expressing, and reasoning about strategic interactions in multi-agent systems, and in particular multi--player games, initiated with introduction of logics such as  \CL,  \ATL, and Strategy Logic. The coalitional goal assignment operator $\brak{\cdot}$  introduced here covers as special cases the modal operators for strategic abilities featuring in \CL and \ATL. Furthermore, whereas these strategic operators assume purely adversarial behaviour of the opposing agents and coalitions, and express unconditional ability of the proponent coalition to achieve its objective against any such objective, the  operator $\brak{\cdot}$ expresses a natural combination of cooperative and non-cooperative interactions which is more common and realistic in `real-life' multi-agent scenarios. Whereas this operator is already quite expressive, it formalises but one general and important pattern of such interaction. Other such patterns, formalising variants of \emph{conditional} strategic reasoning, have been proposed and studied in \cite{LORIVII-GorankoJu}, and certainly more are currently being identified and studied. The patterns of strategic interaction formalised by $\brak{\cdot}$ also enable the expression of new versions of socially relevant solution concepts, such as co-equilibrium, discussed in Section \ref{sec:applications}, opening new perspectives towards formal analysis of multi-agent strategic interaction.

On the more technical aspects of ours work, we note that, while the most important logical questions about \cga have been resolved here, there is still substantial technical follow-up work to do, including complete axiomatizations of important variations and extensions, such as $\cga^+$ and possibly the full  $\xlangcga_\mu$, as well as development of practically efficient tableau-based algorithmic methods for deciding satisfiability and well as for solving the relevant model checking problems. In particular, identifying the exact complexities of these problems, for which reasonable upper bounds have been established in  Section \ref{sec:FMP}, is yet to be completed.  
For solving these problems and exploring further research directions, we hope that the interaction with fragments of Strategy logics, initiated in Section \ref{subsec:ST2SL}, will prove fruitful both ways. In particular, the distinction between path-based and play-based semantics established in Section \ref{subsec:CGAsemantics-variations} suggests that  the latter type of semantics is worth exploring in the context of Strategy logics, too.
 
 Lastly, we note that in this paper we have mainly explored the theoretical foundations and the purely logical and computational aspects of the logics \cga, $\cga^+$ , and 
 $\xlangcga_\mu$, whereas their potential for applications to game-theory and, more generally, to formal modelling  and reasoning about multi-agent strategic interaction has  only been indicated in Section \ref{sec:applications}, but is yet to be unfolded in future work.




\appendix

\section{Appendix}

\subsection{Proof of Proposition \ref{prop:typeone}}

\begin{proof}
We prove each implication of the equivalence separately.

\medskip 
\textbf{Right to left:} 
By Theorem \ref{p:fixpoint-property}, the truth set of $\brak{\gamma}$ is a fixpoint of the operator defined by the formula $\indf{\gamma}{z}$. Therefore,  
$\mu z. \indf{\gamma}{z} \to \brak{\gamma}$ is semantically valid.

\medskip 
\textbf{Left to right:}
 Conversely, suppose that $\gmod,v \sat \brak{\gamma}$. Since $\gamma$ is a long-term temporal goal assignment of type $\until$, we can assume that its support is $\fac = \{C_1,...,C_n,D_1,...,D_m\}$, and that it maps each $C_i$ to $\alpha_i \until \beta_i$, and maps each $D_j$ to $\always \chi_j$. 
Then $\gamma$ is witnessed by some strategy profile $\strprof$, such for each $i \in \{1,...,n\}$, the formula $\alpha_i \until \beta_i$ holds on every computation path in $\paths(v,\strprof,C_i)$ and for each $j \in \{1,...,m\}$, the formula $\always \chi_j$ holds on every computation path in $\paths(v,\strprof,D_j)$. We define a set $T$ of histories as follows: we set $(v_0\tau_0 v_1 ... v_{h - 1}\tau_{h - 1} v_h) \in T$ if and only if:

(i) \ $v_0 = v$ and the word $v_0  ... v_h$ is an initial segment of some computation path in $\paths(v,\strprof,\bigcup \fac)$, and

(ii) \ there are no indices $i,j$ such that $0 \leq j < h$ and $0 \leq i \leq n$ and 
$\gmod, v_j \sat \beta_i$.

\smallskip
It is clear that $v \in T$, and that $T$ can be viewed as a tree rooted at $v$ where the successors of a node $\vec{w}$ in $T$ are the elements of $T$ of the form $\vec{w}u$ for some $u$. Furthermore, since 
$\strprof, v \Vvdash \gamma$ 
and $C_i \subseteq \bigcup \fac$ for each $i \in \{1,...,n\}$, the formula $\alpha_i \until \beta_i$ holds on each computation path in $\paths(v,\strprof,\bigcup \fac)$. Hence, the tree $T$ is \emph{Noetherian}, i.e., it has no infinite branches. 
This means that we may reason by \emph{bar induction}: 
to show that a property $P(x)$ holds for every element of $T$, we show that $P(l)$ holds if $l$ is a leaf, and that if $P(\vec{w}\cdot w')$ holds for every child $\vec{w} \cdot w' \in T$ of some non-leaf node $\vec{w}$ of $T$, then $P(\vec{w})$ holds as well. 

We shall prove the following statement by bar induction: for every word $\vec{w}$ in $T$ with last element $u$, we have 
$\gmod, u \sat \mu z. \indf{\gamma}{z}$. 
It follows eventually that 
$\gmod, v \sat \mu z. \indf{\gamma}{z}$, as required. 
\\\\
\textbf{Case 1:} $\vec{w}$ is a leaf.
Since $\vec{w}$ is a leaf  in $T$, the last element $u$ of $\vec{w}$ must satisfy one of the formulae $ \beta_1,...,\beta_n$, since otherwise the extension $\vec{w} \cdot \out((\strprof(\aga)(\vec{w})_{\aga \in \Agt}), u)$
(where $(\strprof(\aga)(\vec{w})_{\aga \in \Agt})$ is the action profile assigned by 
$\strprof$ at $u$)
would be a child of $\vec{w}$ in $T$. Furthermore, none of these formulae are true in any elements of $\vec{w}$ \emph{except} the last one, since $\vec{w} \in T$. Since, for every computation path $\pi \in \paths(u,\strprof,\bigcup  \fac)$, 
the infinite word $\vec{w} \cdot \pi$ belongs to $\paths(v,\strprof,\bigcup \fac) \subseteq \paths(v,\strprof,C_i)$, it  follows that for each $i \in \{1,...,n\}$ the formula $\alpha_i \until \beta_i$ holds for all computation paths in $\paths(u,\strprof,C_i)$. Likewise, for each $j \in \{1,...,m\}$ the formula $\always \chi_j$ holds for all computation paths in $\paths(u,\strprof,D_j)$. Let us fix some $\beta_i$ that holds at $u$. We get:
\[
\gmod, u \sat \beta_i \wedge \brak{\gamma \setminus C_i}
\]
where  the conjunct on the right is witnessed by the strategy profile $\strprof$. Since this formula is in $\mathsf{Finish}(\gamma)$, we have 
$\gmod, u \sat \mu z. \indf{\gamma}{z}$ as required. 
\\\\
\textbf{Case 2:} $\vec{w}$ is not a leaf, and the IH holds for all of its children in $T$.

Since $\vec{w}$ is not a leaf  in $T$, it has some child in $T$. By definition of $T$ this means that the formulae $\beta_1,...,\beta_n$ are false on every state in $\vec{w}$, including its last element, which we denote by $u$.
Let us pick any $B \subsetneq \bigcup \fac$.
It is not hard to show 
that the following formula 
$$\brak{B_1 \gass  \gamma(B_1),...,B_k \gass  \gamma(B_k)},$$
where $B_1,...,B_k$ are the coalitions in $\fac$ that are contained in $B$,  
holds for each $u' \in \out(u,\strprof,B)$, witnessed by the strategy profile $\strprof$. The key thing to note is that, if $C_i \in \{B_1,...,B_k\}$ for some $i \in \{1,...,n\}$, then for every $u' \in \out(u,\strprof,B)$ and every path $\pi \in \paths(u',\strprof,C_i)$, the path $\vec{w} \cdot \pi$ is in  $\paths(v,\strprof,C_i)$ and therefore satisfies the path formula $\alpha_i \until \beta_i$. 
But, since $\beta_i$ was false everywhere in $\vec{w}$, this means that the path formula $\alpha_i \until \beta_i$ holds on the computation path $\pi$ as well. 

It now follows that for all $B \subsetneq \bigcup \fac$,  
the path formula 
$$\nexttime \brak{B_1 \gass  \gamma(B_1),...,B_k \gass  \gamma(B_k)},$$
where $B_1,...,B_k$ are as above, 
holds for each computation path in $\paths(u,\strprof,B)$.  
Furthermore, by the bar induction hypothesis the path formula $\nexttime \mu z.  \indf{\gamma}{z}$ holds for all $\pi \in \paths(u,\strprof,\bigcup \fac)$, since for every such path $u s_0 s_1 s_2...$ the word $\vec{w} \cdot s_0$ is a child of $\vec{w}$ in $T$. 
Putting all this together, we get:
\[
\gmod, u \sat  \brak{\diffof{\gamma}[\bigcup \fac  \gass   \nexttime \mu z . \indf{\gamma}{z}]}
\]
witnessed by the strategy profile $\strprof$. 
Besides, each $\alpha_i$ for $i \in \{1,...,n\}$ and each $\chi_j$ for $j \in \{1,...,m\}$ holds at $u$, since $\vec{w} \in T$, so we get:
\[
\gmod, u \sat \big( \alpha_1 \wedge... \wedge \alpha_n \wedge \chi_1 \wedge ... \wedge \chi_m \wedge \brak{\diffof{\gamma}[\bigcup \fac  \gass   \nexttime \mu z . \indf{\gamma}{z}]}\big). 
\]
This is a disjunct of the unfolding of the fixpoint formula 
$\mu z. \indf{\gamma}{z}$, and so we get 
$\gmod, u \sat \mu z.  \indf{\gamma}{z}$, as required. 
This concludes the proof.  
\end{proof}

\subsection{Proof of Proposition \ref{prop:typetwo}}
\begin{proof}
 Given a model 
$\gmod = (\states,\Act,\gmap,\out,V)$, 
we prove separately each implication of the equivalence in $\gmod$.
The proof structure is dually analogous to the proof of Proposition \ref{prop:typeone}, but the argument is different.

\medskip 
\textbf{Left to right:} 
By Theorem \ref{p:fixpoint-property}, the truth set of $\brak{\gamma}$ is a fixpoint of the operator defined by the formula $\indf{\gamma}{z}$. Therefore,  
$\brak{\gamma} \to \nu z. \indf{\gamma}{z}$ is semantically valid.

\medskip 
\textbf{Right to left:} 
Let $\theta =  \nu z. \indf{\gamma}{z}$ and suppose $\gmod, v \sat  \theta$. 
Consequently, we get $\gmod, v \sat  \indf{\gamma}{\theta}$,  
where $\indf{\gamma}{\theta} = \chi_1 \wedge... \wedge  \chi_m \wedge 
\brak{\gammaof{\theta}}$, 
 where $\gamma$ is supported by  $\fac = \{ D_1,...,D_m \}$,  respectively  mapped to $\always \chi_1,...,\always \chi_n$.

Thus, $\gmod, v \sat \chi_i$ for each $i$ and 
$\gmod, v \sat  \brak{\gammaof{\theta}}$. 

So, there exists a strategy profile $\strprof_0$ such that 
$\strprof_0, v \Vvdash \gammaof{\theta}$, 

hence, for each $D_i \in \fac$ and every play $\pi \in \Plays(v,\strprof_0,D_i)$, we have 
$\gmod, \pi \sat \gammaof{\theta}(D_i)$. 

Recall that for each $C \subsetneq \bigcup \fac$,  the respective goal of $\gammaof{\theta}$ is 

$\gammaof{\theta}(C) =  \diffof{\gamma}(C) = \X \brak{ \big(D_j \gass \always \chi_j \big)_{D_j \subseteq C}}$ 

and note that  $ \big(D_j \gass \always \chi_j \big)_{D_j \subseteq C}$ is precisely 
$\gamma\vert_{C}$. 
Thus, for each such $C$ and every state $w \in \Out[v,\strprof_0\vert_{C}]$ we have 
$\gmod, w \sat \brak{\gamma\vert_{C}}$.  
Let $\strprof^{w,C}$ be a strategy profile witnessing 
$\gamma\vert_{C}$ at w.

Now we will define a  strategy profile $\strprof$ that witnesses $\gamma$ at $v$. 
Intuitively, for each player $\aga$, $\strprof$ will combine $\strprof_0$ at  $v$ with the strategies $\strprof^{w,C}$ for $\aga \in C$,  
ensured to exist at the respective  outcome states $w$,    
 applied on the respective histories 
 passing through $w$.  
For any player $\aga \in \bigcup \fac$ 
we define $\strprof(\aga)$ as follows, on any history $h$ in $\gmod$ starting at $v$.

\begin{enumerate}
\item If $h$ has length 0, i.e. $h = v$, 
we define $\strprof(\aga, h) := \strprof_0(\aga, v)$. 

\item  
Let $h =  v \actprof h'$, where $h' = v_1 \actprof_1 ... v_{n -1 }\actprof_{n - 1} v_n$.   
We first compare the action profile $\actprof$ with the action profile $\actprof_0 = \strprof_0(v)$ prescribed by $\strprof_0$ at $v$. Let $C_{\actprof}$ be the set of players in $\bigcup \fac$ 
 whose actions in $\actprof$ and $\actprof_0$ are the same. 
 Then $v_1 \in \Out[v,\strprof_0\vert_{C_{\actprof}}]$, so 
$\gmod, v_1 \sat \brak{\gamma\vert_{C_{\actprof}}}$.  
  If $\aga \notin C_{\actprof}$, then $\strprof(\aga, h)$ is defined arbitrarily. 
 Suppose $\aga \in C_{\actprof}$. 
 
 If $C_{\actprof} \subsetneq \bigcup \fac$, then we use the strategy $\strat_{\aga}$ for player $\aga$ from the strategy profile 
 $\strprof^{v_1,C_{\actprof}}$ to define $\strprof(\aga, h) := \strat_{\aga}(h')$. 
 
 In the case when $C_{\actprof} = \bigcup \fac$, we have that 
 $\gammaof{\theta}(\bigcup \fac) =  \X \theta$, hence 
 $\gmod, v_1 \sat \theta$.  
In this case, we define $\strprof(\aga, h): = \strprof_1(\aga, h')$, where $\strprof_1$ is the strategy recursively defined by the same procedure on paths starting from
$v_1$. Note that this definition is correct because for every path $h$, being finite, the procedure will eventually reach the case (1) of a subpath of length 0, where the strategy is defined explicitly.  
\end{enumerate}

For all other histories $h$, the action $\strprof(\aga, h)$ is defined arbitrarily. 
  
 \medskip 
 By virtue of the construction of $\strprof$, for $D_j \in \fac$ and 
 every play $\pi \in \Plays(v,\strprof,D_j)$, we have $\gmod, \pi \sat \always \chi_j $. The proof is direct for all $D_j \subsetneq \bigcup \fac$, while for 
 $D_j = \bigcup \fac$, if applicable, it  follows from the fact,  proved by induction on $n$, that $\gmod, v_n \sat \chi_j$, where $v_0 = v$ and 
$\pi=  v_0 \actprof_0 v_1 \actprof_1 v_2 \actprof_2...$. 

Therefore, $\strprof, v \Vvdash \gamma$, hence $\gmod, v \sat  \brak{\gamma}$. 

 \end{proof}

\subsection{Proof of Proposition \ref{prop:closure}}
\begin{proof} 
The proof can be done by several nested inductions: first, on the number of formulae in $\Phi$, then -- in the inductive step -- by structural induction on the additional formula $\varphi$. The only non-trivial point now is when 
$\varphi = \brak{\gamma}$, where $\gamma$ is a long-term temporal goal assignment, to show that $\varphi$ adds finitely many formulae to $\ecl(\Phi)$ when recursively taking components.  Let $\fac$ be the support of $\gamma$. 
We can assume that $\gamma \in \typeone$, as the case of $\gamma \in \typetwo$ is similar but simpler. 
We will prove the claim by a 3rd nested induction on the number of coalitions in  
$\fac = \{C_1,...,C_n,D_1,...,D_m\}$ (we can assume $m > 0$) 
Let 
$
\gamma(C_1) = \alpha_1 \until \beta_1,..., \gamma(C_n) = \alpha_n \until \beta_n
$
and 
$
\gamma(D_1) = \always \chi_1,..., \gamma(D_m) = \always \chi_m
$. 
Then the component of $\gamma$ added to $\ecl(\Phi)$ is (by Proposition \ref{prop:unfold})
$\indf{\gamma}{\brak{\gamma}} = \unf{\gamma} = \bigvee \mathsf{Finish}(\gamma) \vee \bigg(\bigwedge \ugam \wedge \bigwedge \agam \wedge \brak{\diffof{\gamma}}\bigg)$. 
By the (nested) inductive hypothesis in the structural induction on $\varphi$ and the innermost inductive hypothesis, every disjunct 
$( \beta_i \wedge \brak{\gamma \setminus C_i}) \in \mathsf{Finish}(\gamma)$ 
adds only finitely many new components to $\ecl(\Phi)$. Again by the inductive hypothesis on the structure of $\varphi$, all subformulae $\alpha_i$ and $\chi_j$, hence all formulae in $\ugam$ and $\agam$, add only finitely many new components, too. 
Finally, note that $\diffof{\gamma}$
is a goal assignment and all goals in it are either $\X$-prefixed goals in $\gamma$ or $\X\gamma$, hence 
$\brak{\diffof{\gamma}}$ 
only adds finitely many new components, too.  
The case of negated goal assignment $\varphi = \lnot \brak{\gamma}$, is completely analogous. That completes all inductive steps of the nested inductions, and the proof itself. 
\end{proof}

\subsection{Proof of Proposition \ref{prop:network-sat}}
\begin{proof}
We prove, by structural induction on all formulae $\varphi \in \Phi$, 
 that for every $u \in T$, $\varphi \in L(u)$ iff $\gmod(\network),u\sat \varphi$. The clauses for propositional variables and boolean connectives are standard, so we omit them. Since all formulae in $\Phi$ are assumed to be in normal form, for each formula in $\Phi$ of the shape $\brak{\gamma}$, the goal assignment $\gamma$ is either nexttime  or long-term temporal. 
The induction step in the first case follows from the fact that the network $\network$ is assumed to be one-step coherent and the inductive hypothesis (IH). 
Indeed, since $L(u)$ is a $\Phi$-atom, exactly one of $\brak{\gamma}$ and $\neg \brak{\gamma}$ is in $L(u)$. Then, since $\gamma$ is nexttime and $\Phi$ is closed, for every coalition $C$, $\gamma(C) = \nexttime \psi$ where $\psi \in \Phi$. 
Now, if $\brak{\gamma} \in L(u)$ then the marking $L$ verifies $\gamma$ at $u$ and,  by the IH applied to all $\gamma(C)$, we obtain that $\gmod(\network),u \sat \brak{\gamma}$. 
Likewise, if $\neg \brak{\gamma} \in L(u)$, then the marking $L$ refutes $\gamma$ at $u$ and, by the IH applied to all $\gamma(C)$, we obtain that $\gmod(\network),u \not\sat \brak{\gamma}$. This completes the case. 

Now, we focus on the case where $\gamma$ is long-term temporal. The claim will be proved by a sub-induction on the size (number of coalitions) of the support $\fac$ of $\gamma$. 

\paragraph{Base case:} $\fac = \emptyset$. Then $\gamma$ is the trivial goal assignment, and since a perfect network has no leaves this case is trivial. 

\paragraph{Induction step:} We now suppose that $\fac = \{C_1,...,C_{n}\}$ is of size $n$, and that the induction hypothesis holds for all $\gamma'$ with support of size $< n$. We divide the proof that the induction hypothesis holds for $\gamma$ into two sub-cases, depending on the type of $\gamma$. 
We will first prove the easier sub-case. 

\paragraph{Case:} $\gamma$ is in $\typetwo$. 
Let $\gamma$ be supported by $\fac = \{D_1,...,D_m\}$ and defined by: 
$
\gamma(D_1) = \always \chi_1,..., \gamma(D_m) = \always \chi_m. 
$
We claim that the set
\[
 \lset{\brak{\gamma}}{\network}  := \{u \in T \mid \brak{\gamma} \in L(u)\}
\]
is the greatest post-fixpoint $\nu f$ of the monotone map $f$ induced by $\indf{\gamma}{z}$ in $\gmod(\network)$, 
defined by 
\[
f(Z) = \big\{u \in T \mid  
\gmod(\network), u \sat_{[z \mapsto Z]} \indf{\gamma}{z} \big\}
\]
\\ 
By Proposition \ref{prop:typetwo}, the greatest post-fixpoint $\nu f$ defined above is
{equal} 
to $\tset{\brak{\gamma}}_{\gmod(\network)}$, whence the case will follow. 

For the inclusion $\lset{\brak{\gamma}}{\network} \subseteq \nu f$, we reason by co-induction. That is, we prove that the set $\lset{\brak{\gamma}}{\network}$ is a post-fixpoint of $f$.
Recall that 
$\indf{\gamma}{z} = \chi_1 \wedge ... \wedge \chi_m \wedge \brak{\gammaof{z}}$,  where 
$\gammaof{\phi}$ denotes $\diffof{\gamma}[\bigcup \fac  \gass  \X \phi]$.
To show that $\lset{\brak{\gamma}}{\network}$ is a post-fixpoint of $f$, suppose that $u \in \lset{\brak{\gamma}}{\network}$, i.e. that $\brak{\gamma} \in L(u)$. 
Since $\Phi$ is closed under taking components, 
$\indf{\gamma}{\brak{\gamma}} \in \Phi$, and hence 
 $\brak{\gammaof{\brak{\gamma}}}  \in \Phi$  
 whenever $\brak{\gamma} \in \Phi$.  By the post-fixpoint axiom $\mathsf{Fix}$ and  Proposition \ref{prop:unfold} 
we have $\indf{\gamma}{\brak{\gamma}} \in L(u)$, as well.

Since $L(u)$ is an atom, it follows that $\chi_i  \in L(u)$ for each $i=1,..,n$  and $\brak{\gammaof{\brak{\gamma} }}\in L(u)$. 
Then, by the inductive hypothesis, it follows (note that
$z$ does not occur free in  $\chi_i$) that  $\gmod(\network), u \sat_{[z  \mapsto \lset{\brak{\gamma}}{\network}]} \chi_i$  for each $i=1,..,n$. 
{Furthermore, using one-step coherence of the network $\network$, and since $\brak{\gammaof{\brak{\gamma}}} \in L(u)$, we have  $\gmod(\network), u \sat_{[z \mapsto \lset{\brak{\gamma}}{\network}]}  \brak{\gammaof{z}}$. To see this, $\gammaof{\brak{\gamma}}$ is a nexttime goal assignment, so one-step coherence guarantees that there is some $\strprof \in \Pi_{a \in \Agt}$ such that, for every $C$ in the support of $\gammaof{\brak{\gamma}}$ such that $\gammaof{\brak{\gamma}}(C) = \nexttime \psi$ and for every strategy profile $\strprof'$ with $\strprof' \sim_C \strprof$, we have $\psi \in L(\out(\strprof'),u)$. 
In particular, this means that for each $C \subsetneq \bigcup \fac$ in the support of $\gammaof{\brak{\gamma}}$, each $D_i \subseteq C$, and each $\strprof' \sim_C \strprof$, we have $\chi_i \in L(\out(\strprof',u))$, and so by the induction hypothesis $\gmod(\network),\out(\strprof',u) \sat_{[z  \mapsto \lset{\brak{\gamma}}{\network}]} \chi_i$. Also, for each $\strprof' \sim_{\bigcup \fac} \strprof$, we have $\brak{\gamma} \in L(\out(\strprof',u))$, and hence $\gmod(\network),\out(\strprof',u) \sat_{[z  \mapsto \lset{\brak{\gamma}}{\network}]} z$ by definition of $\lset{\brak{\gamma}}{\network}$. Putting these facts together we get $\gmod(\network), u \sat_{[z \mapsto \lset{\brak{\gamma}}{\network}]}  \brak{\gammaof{z}}$ as claimed. 
}
We now get
$\gmod(\network), u \sat_{[z \mapsto \tset{\brak{\gamma}}_{L}]} \indf{\gamma}{z}$, 
hence $u \in f(\lset{\brak{\gamma}}{\network})$. 
Thus, $\lset{\brak{\gamma}}{\network}$ is a post-fixpoint of $f$.

For the converse inclusion $\nu f \subseteq \lset{\brak{\gamma}}{\network}$, we reason contrapositively: suppose that $v$ is some node in $T$ that does not belong to $\lset{\brak{\gamma}}{\network}$. This means that $\brak{\gamma} \notin L(v)$, hence $\neg \brak{\gamma} \in L(v)$ since $L(v)$ is an atom over $\Phi$, and $\Phi$ is a closed set of formulas and hence closed under single negations. We will prove that $v \notin \nu f$. Since $\nu f$ is the intersection of its approximants $f^\xi(T)$ where $\xi$ ranges over ordinals, it suffices to find a \emph{finite} ordinal $k < \omega$ such that $v \notin f^k(T)$. Here we recall that $f^k(T)$ is defined inductively by $f^0(T) = T$, $f^{i + 1} = f(f^i(T))$.  

Note that the formula $\neg \brak{\gamma}$ is a $\typetwo$-eventuality. 
So, since $\network$ was assumed to be a perfect network, there is some $k < \omega$ for which this eventuality is partially fulfilled in $k$ steps at $v$.  Hence it suffices to prove, by induction on $k$, that for all $w \in \lset{\neg \brak{\gamma}}{\network}$, if the eventuality $\neg \brak{\gamma}$ is fulfilled in $k$ steps at $w$ then $w \notin f^{k + 1}(T)$. We refer to the induction hypothesis on $k$ as the \emph{innermost} induction hypothesis. We refer to the induction hypothesis of the structural induction on complexity of formulas as the \emph{outermost} induction hypothesis.

For the base case of the innermost induction, when $k = 0$, 
let $w$ be some element of 
$\lset{\neg \brak{\gamma}}{\network}$
at which the eventuality $\neg\brak{\gamma}$ is partially fulfilled in 0 steps.
If $w \in  f^1(T) = f(T)$, then $\gmod(\network),w \sat_{[z \mapsto T]} \indf{\gamma}{z}$, hence $\gmod(\network),w \sat \chi_i$ for each $\chi_i$. By the outermost induction hypothesis on $\chi_i$, each $\chi_i$ is in $L(w)$. Since the eventuality $\neg \brak{\gamma}$ is partially fulfilled in $0$ steps at $w$ and 
 $w \in  \lset{\neg \brak{\gamma}}{\network}$, we have $\overline{\chi_i} \in L(w)$ for some $\chi$, which contradicts the consistency of $L(w)$. Hence $w  \notin f(T)$.

For the induction step, suppose $k = j + 1$ and the innermost induction hypothesis holds for $j$. Let $w$ be some element of $\lset{\neg \brak{\gamma}}{\network}$
at which the eventuality $\neg\brak{\gamma}$ is partially fulfilled in $k$ steps. 
Suppose, for a contradiction, that 
$w \in f^{k + 1}(T) = f(f^k(T))) = f(f^{j + 1}(T))$. Then $\gmod(\network),w \sat_{[z \mapsto f^{j + 1}(T)]} \indf{\gamma}{z}$. 
Let $F^j$ 
be the set of all 
$x \in \lset{\neg \brak{\gamma}}{\network}$  
such that  $\neg \brak{\gamma}$ is partially fulfilled in $j$ steps at $x$.  By the innermost induction hypothesis on $j$, we get $f^{j + 1}(T)\cap F^j = \emptyset$. Since $\neg\brak{\gamma}$ is partially fulfilled in $k = j + 1$ steps at $w$, 
either $\overline{\chi_i} \in L(w)$ for some $\chi_i$ 
or there exists a marking $\mrk$ that refutes $\brak{\diffof{\gamma}}$ at $w$,  and such that for all $w' \in T$ 
such that $w'$ is a child of $w$,  $w' \in F^j$ whenever $\neg \brak{\gamma} \in \mrk(w')$. In the former case we immediately get a contradiction. So we focus on the latter case.

{ 
Since $\gmod(\network),w \sat_{[z \mapsto f^{j + 1}(T)]} \indf{\gamma}{z }$, we have $\gmod(\network),w \sat_{[z \mapsto f^{j + 1}(T)]} \brak{\gammaof{z}}$. Let $\strprof$ be a witnessing strategy profile for $\gammaof{z}$ at $w$. Then, for every $\strprof' \sim_{\bigcup \fac} \strprof$ we have $\gmod(\network),\out(\strprof',w) \sat_{[z \mapsto f^{j + 1}(T)]}  z$, i.e. $\out(\strprof',w)  \in f^{j + 1}(T)$. Furthermore, for each $C \subsetneq \bigcup \fac$ 
in the support of $\gammaof{z}$, each $D_i \subseteq C$ and each $\strprof' \sim_{C} \strprof$ we have $\gmod(\network),\out(\strprof',w) \sat_{[z \mapsto f^{j + 1}(T)]}  \chi_i$, and so $\chi_i \in L(\out(\strprof',w))$ by the outermost induction hypothesis on $\chi_i$. On the other hand, since the marking $\mrk$ refutes  $\brak{\diffof{\gamma}}$ at $w$ there must be some  $C^*$ in the support of $\diffof{\gamma} = \gammaof{\brak{\gamma}}$ and some $\strprof' \sim_{C^*} \strprof$ such that $\overline{\psi} \in \mrk(\out(\strprof',w))$, where $\gammaof{\brak{\gamma}}(C^*) = \nexttime \psi$. If $C^* = \bigcup \fac$ 
then $\psi = \brak{\gamma}$ so $\overline{\psi} = \neg \brak{\gamma}$, so $\neg \brak{\gamma} \in \mrk(\out(\strprof',w))$. Recall that $\mrk$ was such that  $w' \in F^j$ whenever $\neg \brak{\gamma} \in \mrk(w')$, so we get $\out(\strprof',w) \in F^j$. But, since $\strprof' \sim_{\bigcup \fac} \strprof$, we must also have $\out(\strprof',w) \in f^{j+1}(T)$, which is a contradiction since $f^{j + 1}(T)\cap F^j = \emptyset$. On the other hand, if $C^* \subsetneq \bigcup \fac$, then, since $C^*$ is in the support of $\diffof{\gamma}$, there must be some $D_i$ in the support of $\gamma$ with $D_i \subseteq C^*$, and so $\chi_i$ is a conjunct of $\psi$ and therefore inconsistent with $\overline{\psi} \in \mrk(\out(\strprof',w)) \subseteq L(\out(\strprof',w))$. But, since $\strprof'\sim_{C^*} \strprof$, we get $\chi_i \in L(\out(\strprof',w))$, which is a contradiction since $L(\out(\strprof',w))$ is consistent by assumption. So, in either case we get a contradiction,  hence we have shown that $w \notin f^{j + 1}(T)$, as desired.}

\paragraph{Case:} $\gamma$ is in $\typeone$. 
The argument is similar to the above, but somewhat more complicated. 
Let $\gamma$ be supported by $\fac = \{C_1,...,C_n,D_1,...,D_m\}$ and defined by: 
\[
\gamma(C_1) = \alpha_1 \until \beta_1,..., \gamma(C_n) = \alpha_n \until \beta_n, 
\]
and
\[
\gamma(D_1) = \always \chi_1,..., \gamma(D_m) = \always \chi_m. 
\]
Since $\gamma$ is in $\typeone$ we have $\{C_1,...,C_n\} \neq \emptyset$. 

We claim that the set
\[
 \lset{\brak{\gamma}}{\network}  := \{u \in T \mid \brak{\gamma} \in L(u)\}
\]
is the least pre-fixpoint $\mu f$ of the monotone map $f$ induced by $\indf{\gamma}{z}$ in $\gmod(\network)$, 
defined by 
\[
f(Z) = \big\{u \in T \mid  
\gmod(\network), u \sat_{[z \mapsto Z]} \indf{\gamma}{z} \big\}
\]
\\ 
By Proposition \ref{prop:typetwo}, the least pre-fixpoint $\mu f$ is equivalent to $\tset{\brak{\gamma}}_{\gmod(\network)}$, whence the case will follow. 

For the inclusion $\mu f  \subseteq \lset{\brak{\gamma}}{\network}$, we reason by least fixpoint induction. That is, we prove that the set $\lset{\brak{\gamma}}{\network}$ is a pre-fixpoint of $f$. Recall that
$$
\indf{\gamma}{z} = \bigvee_{1 \leq i \leq n} (\beta_i \wedge \brak{\gamma \setminus C_i}) \vee \big( \bigwedge_{1 \leq i \leq n} \alpha_i \wedge \bigwedge_{1 \leq i \leq m} \chi_i \wedge \brak{\gammaof{z}}\big),
$$  
where $\gammaof{\phi}$ denotes $\diffof{\gamma}[\bigcup \fac  \gass  \X \phi]$.
To show that $\lset{\brak{\gamma}}{\network}$ is a pre-fixpoint of $f$, suppose that $u \in f(\lset{\brak{\gamma}}{\network})$. We need to show that $u \in \lset{\brak{\gamma}}{\network}$, i.e. that $\brak{\gamma} \in L(u)$. By definition of $f$, either there is some $i \in \{1,...,n\}$ for which $\gmod(\network),u \sat_{z \mapsto \lset{\brak{\gamma}}{\network}} \beta_i \wedge \brak{\gamma\setminus C_i}$, or $\gmod(\network),u \sat_{z \mapsto \lset{\brak{\gamma}}{\network}} \alpha_i$ for each $i \in \{1,...,n\} $, $\gmod(\network),u \sat_{z \mapsto \lset{\brak{\gamma}}{\network}} \chi_i$ for each $i \in \{1,...,m\} $, and $\gmod(\network),u \sat_{z \mapsto \lset{\brak{\gamma}}{\network}} \brak{\gammaof{z}}$. 

In the former case, the induction hypothesis on $\beta_i$ gives $\beta_i \in L(u)$, and since the support of $\gamma \setminus C_i$
is smaller than that of $\gamma$,  
 the induction hypothesis for the induction on the size of the support gives $\brak{\gamma \setminus C_i} \in L(u)$.  Since $\Phi$ is a closed set, we have $\beta_i \wedge \brak{\gamma \setminus C_i} \in \Phi$, and since $L(u)$ is an atom we get $\beta_i \wedge \brak{\gamma \setminus C_i} \in L(u)$. It follows, again by closure of $\Phi$ and $L(u)$ being an atom, that $\brak{\gamma} \in L(u)$ as required. 
 
 In the latter case, the induction hypotheses on $\alpha_1,...,\alpha_n$ and $\chi_1,...,\chi_m$ ensure that these formulas are all in $L(u)$. We will show that $\brak{\gammaof{\brak{\gamma}}} \in L(u)$, from which it will follow using the fact that $L(u)$ is an atom and $\Phi$ is closed that $\brak{\gamma} \in L(u)$. 
To prove this, we use that $\gmod(\network),u \sat_{z \mapsto \lset{\brak{\gamma}}{\network}} \brak{\gammaof{z}}$. It suffices to show that the marking $L$ does not refute $\brak{\gammaof{\brak{\gamma}}}$, because by one-step coherence it follows that  $\neg \brak{\gammaof{\brak{\gamma}}} \notin L(u)$ and hence $\brak{\gammaof{\brak{\gamma}}} \in L(u)$. So suppose, for a contradiction, that $L$ refutes $\brak{\gammaof{\brak{\gamma}}}$. Since $\gmod(\network),u \sat_{z \mapsto \lset{\brak{\gamma}}{\network}} \brak{\gammaof{z}}$, there is a witnessing strategy profile $\strprof$ such that $\gmod{\network},\out(\strprof',u)\sat_{z \mapsto \lset{\brak{\gamma}}{\network}} z $, and hence $\brak{\gamma} \in L(\out(\strprof',u))$, for each coalition $\strprof' \sim_{\bigcup \fac} \strprof$, and $\gmod(\network),\out(\strprof',u)\sat_{z \mapsto \lset{\brak{\gamma}}{\network}} \psi$ for each $E \subsetneq \bigcup \fac$ in the support of $\gammaof{z}$ with $\gammaof{z}(E) = \X \psi$. On the other hand, since $L$ refutes $\brak{\gammaof{\brak{\gamma}}}$, there must be some  $E^*$ in the support of $\brak{\gammaof{\brak{\gamma}}}$ (which is the same as the support of $\gammaof{z}$), and some $\strprof' \sim_{E^*} \strprof$ such that $\psi \notin \out(\strprof',u)$ where $\gammaof{\brak{\gamma}}(E^*) = \X\psi$. If $E^* = \bigcup \fac$, then $\psi = \brak{\gamma}$, so $\brak{\gamma} \notin L(\out(\strprof',u))$. But we already know that $\brak{\gamma} \in L(\out(\strprof',u))$ for each $\strprof' \sim_{\bigcup \fac} \strprof$, so this is a contradiction. On the other hand, suppose that $E^* \subsetneq \bigcup \fac$. Then $\psi = \brak{\gamma\vert_{E^*}}$, so $\brak{\gamma \vert_{E^*}} \notin L(\out(\strprof'u))$.  But the support of $\gamma\vert_{E^*}$ is smaller than that of $\gamma$ since $E^* \subsetneq \bigcup \fac$, so the induction hypothesis for the induction on the size of supports applies, and we get $\gmod(\network),\out(\strprof',u)\nsat_{z \mapsto \lset{\brak{\gamma}}{\network}} \brak{\gamma \vert_{E^*}}$. But, since $\gammaof{\brak{\gamma}}(E^*) = \X \brak{\gamma \vert_{E^*}}$ and $\strprof' \sim_{E^*}  \strprof$, we have $\gmod(\network),\out(\strprof',u)\sat_{z \mapsto \lset{\brak{\gamma}}{\network}} \brak{\gamma \vert_{E^*}}$, so we get a contradiction again. Thus, we have shown that $u \in \lset{\brak{\gamma}}{\network}$, as required. Therefore, $\mu f \subseteq \lset{\brak{\gamma}}{\network}$.

Now for the converse inclusion, $\lset{\brak{\gamma}}{\network} \subseteq \mu f$,  suppose $u \in \lset{\brak{\gamma}}{\network}$, i.e. $\brak{\gamma} \in L(u)$. Note that $\brak{\gamma}$ is a \typeone-eventuality. So, since $\network$ is perfect, there is some $k < \omega$ such that the eventuality is partially fulfilled in $k$ steps at $u$. Since the least fixpoint $\mu f$ is the union of its ordinal approximants $f^\xi(\emptyset)$, where $\xi$ ranges over ordinals, it suffices to show by induction on $k$ that for all  $w \in \brak{\gamma}$, if the eventuality $\brak{\gamma}$ is fulfilled in $k$ steps at $w$ then $w \in f^{k+1}(\emptyset)$. We have several nested inductions at this point, so we refer to the induction hypothesis for the structural induction on complexity of formulas as the \emph{outermost induction hypothesis}, that of the induction on the size of supports as the \emph{middle induction hypothesis}, and the induction hypothesis on $k$ 
as the \emph{innermost induction hypothesis}.

For the base case, where $k = 0$, if $\brak{\gamma}$ is fulfilled in $0$ steps at $w$ then there is some $C_i$ with $\beta_i \wedge \brak{\gamma \setminus C_i} \in L(u)$. The outermost induction hypothesis on $\beta_i$ and the middle induction hypothesis on $\brak{\gamma \setminus C_i}$, together with closure of $\Phi$ and $L(w)$ being an atom, immediately give $\gmod(\network),w \sat \beta_i \wedge \brak{\gamma\setminus C_i}$, and therefore 
$\gmod(\network), w \sat_{[z \mapsto \emptyset]} \indf{\gamma}{z}$, since the variable $z$ does not appear in $\beta_i$ or $\brak{\gamma\setminus C_i}$. By definition we get $w \in f(\emptyset) = f^1(\emptyset)$. 

Now let $w \in \lset{\brak{\gamma}}{\network}$, and suppose $\brak{\gamma}$ is partially fulfilled in $k = j + 1$ steps at $w$ where the induction hypothesis holds for $j$. If $\brak{\gamma}$ is partially fulfilled in $j$ steps at $w$ then the innermost induction hypothesis applies and we are done.  Otherwise, since $\brak{\gamma}$ is partially fulfilled in $k = j + 1$ steps at $w$, we have $\alpha_i \in L(w)$ for all $i \in \{1,...,n\}$, and $\chi_j \in L(w)$ for all $j \in \{1,...,m\}$, and there exists a marking $\mrk$ that verifies $\brak{\gammaof{\brak{\gamma}}}$ at $w$ and is such that for every child $w'$ of $w$ with $\brak{\gamma} \in \mrk(w')$ the eventuality $\brak{\gamma}$ is partially fulfilled in $j$ steps at $w'$. The outermost induction hypothesis gives  $\gmod(\network), w \sat \alpha_i$ for all $i \in \{1,...,n\}$, and $\gmod(\network), w \sat \chi_j$ for all $j \in \{1,...,m\}$. Since $z$ does not appear in any of these formulas we get  $\gmod(\network), w \sat_{[z \mapsto f^j(\emptyset)]} \alpha_i$ for all $i \in \{1,...,n\}$, and $\gmod(\network), w \sat_{[z \mapsto f^j(\emptyset)]} \chi_j$ for all $j \in \{1,...,m\}$.

Let $\strprof$ be some strategy profile witnessing that the marking $\mrk$ verifies $\brak{\gammaof{\brak{\gamma}}}$ at $w$. Let $E$ be some coalition in the support of $\gammaof{z}$, which is the same as the support of $\gammaof{\brak{\gamma}}$. If $E = \bigcup \fac$, let $\strprof' \sim_{\bigcup \fac} \strprof$. Then $\brak{\gamma} \in \mrk(\out(\strprof',w))$, so $\brak{\gamma}$ is partially fulfilled in $j$ steps at $w$. By the innermost induction hypothesis, $\out(\strprof',w) \in f^j(\emptyset)$. So $\gmod(\network),\out(\strprof',w) \sat_{z \mapsto f^j(\emptyset)} z$. On the other hand, if $E \subsetneq \bigcup \fac$, then $\gammaof{z}(E)$ is of the form $\X\brak{\gamma \vert E}$. So, if $\strprof' \sim_E \strprof$ then $\brak{\gamma \vert_E} \in \mrk(\out(\strprof',w)) \subseteq L(\out(\strprof',w))$. Since the support of $\gamma \vert_E$ is smaller than that of $\gamma$, the middle induction hypothesis applies and we get $\gmod(\network), \out(\strprof',w) \sat \brak{\gamma \vert_E}$. Since the variable $z$ does not appear in $\brak{\gamma \vert_E}$ we get $\gmod(\network),  \out(\strprof',w) \sat_{[z \mapsto f^j(\emptyset)]} \brak{\gamma \vert_E}$. So, we get $\gmod(\network), w \sat_{z \mapsto f^j(\emptyset)} \brak{\gammaof{z}}$. Collecting all the facts we have established, we get $\gmod(\network), w \sat_{[z \mapsto f^j(\emptyset)]} \indf{\gamma}{z}$, so $w \in f(f^j(\emptyset)) = f^{j+1}(\emptyset) = f^k(\emptyset)$ as required. 

We have thus shown that $\lset{\brak{\gamma}}{\network} \subseteq \mu f$, and the proof is completed.  
\end{proof}

\subsection{Proof of Theorem \ref{thm:bisimulation invariance}}
\begin{proof}
Structural induction on $\varphi$. 
We prove further, in Corollary \ref{cor:NF}, that every \cga-formula is equivalent to one in normal form, so we can assume that $\varphi$ is in a normal form. 
The boolean cases are routine, so we only consider the case for 
$\varphi = \brak{\gamma}$, where $\gamma$ is either a nexttime assignment, or a long-term temporal assignment in \typeone or in \typetwo.

\smallskip
1. We first consider the case of nexttime assignments. 
Let $\gamma$ be a nexttime assignment defined by  
$\gamma(C_1) = \X \phi_1,..., \gamma(C_n) =  \X \phi_n$. 
Suppose that $\gmod,s_1 \models  \brak{\gamma}$, 
witnessed by a joint action $\actprof^1$ for $\Agt$ at $s_1$. Let $\actprof^2$ be some joint action for $\Agt$ at $s_2$ witnessing the Forth condition with respect to $\actprof^1$. We need to show that $\Out[s_2,\actprof^2_{C_i}] \subseteq \tset{\phi_i}$ for each $i \in \{1,...,k\}$.  Suppose $v \in \Out[s_2,\actprof^2_{C_i}]$. Apply the LocalBack condition to find $v' \in \Out[s_1,\actprof^1_{C_i}]$ with $v' \beta v$. Since $\Out[s_1,\actprof^1_{C_i}] \subseteq \tset{\phi_i}$ we get $\gmod,v \models \phi_i$ by the induction hypothesis on $\phi_i$, as required. 
Thus, $\gmod,s_2 \models  \brak{\gamma}$. 

The converse direction is proved in the same way.

\smallskip

2. Next, we claim that for any goal assignment $\gamma$, if the bisimulation invariance claim holds for $ \brak{\gamma}$ and for all proper subformulae\footnote{Note that the proper subformulae of $\brak{\gamma}$ are all formulae assigned as goals by $\gamma$ and their subformulae.} 
of $ \brak{\gamma}$ and all pairs of bisimilar states $s_1, s_2$ in the state space  
$\states$ of $\gmod$, then it also holds likewise for $\brak{\diffof{\gamma}}$, because this is a special case of a nexttime assignment, involving only $ \brak{\gamma}$  and subformulae of $ \brak{\gamma}$. 
This claim we will use further, when proving the inductive steps for long-term assignments. 

\smallskip
3. Let $\gamma \in \typetwo$
and assume that the bisimulation invariance claim holds for all proper subformulae of 
$ \brak{\gamma}$ and all pairs of bisimilar states $s_1, s_2$ in the state space  
$\states$ of $\gmod$. 
We will prove the bisimulation invariance of $ \brak{\gamma}$ for all such pairs of states by using the fixpoint characterisation in Proposition \ref{prop:typetwo} and proving the claim for $\nu z. \indf{\gamma}{z}$, instead. 

To prove that claim we hereafter treat $\indf{\gamma}{Z}$ as a (monotone) operator on sets of states and use the characterisation of greatest fixed points given by Knaster-Tarski theorem, according to which 
$\nu z. \indf{\gamma}{z} = \bigcap_{\alpha \in  \mathsf{Ord}}  \indf{\gamma}{Z^{\alpha}} $, where the family $\{Z^{\alpha}\}_{\alpha \in \mathsf{Ord}}$ of subsets of $\states$ is defined by transfinite induction, as usual:   

$Z^{0} = \states$; \ $Z^{\alpha+1} =  \indf{\gamma}{Z^{\alpha}}$; \ 
$Z^{\lambda} = \bigcap_{\alpha < \lambda} \indf{\gamma}{Z^{\alpha}} $ for limit ordinals $\lambda$. 
 
 For technical convenience we will treat each $Z^{\alpha}$ both as a set of states and as a formula for which this set is its extension in $\cgm$,  
 noting that  bisimulation invariance of a formula $\phi$ in a model $\cgm$ is equivalent to the closure under bisimulation of its extension $\tset{\phi}_\cgm$.

It suffices to prove bisimulation invariance (resp. bisimulation closure) of each approximant formula $Z^{\alpha}$,   
as the closure under bisimulation is preserved in the intersection of any family of sets.  
We prove these closures by a nested transfinite induction on $\alpha$. 
The only non-trivial case is that of successor ordinals. 

Recall that, for a goal assignment $\gamma \in \typetwo$, we have 
$
\indf{\gamma}{\phi} =
 \bigwedge \agam \wedge 
\brak{\gammaof{\phi}} 
$.

Thus, $
Z^{\alpha+1} = 
\indf{\gamma}{Z^{\alpha}}  =
 \bigwedge \agam \wedge 
\brak{\gammaof{Z^{\alpha}}} 
$.
Since each formula in $\agam$ is a proper subformula of $ \brak{\gamma}$, its bisimulation invariance follows from the inductive hypothesis of the main induction, so it only remains to show the bisimulation invariance of $\brak{\gammaof{Z^{\alpha}}}$, assuming the bisimulation invariance of 
${\gammaof{Z^{\alpha}}}$. This claim is a particular case of the claim 2 for nexttime extensions of goal assignments, proved above. 

This completes the inductive step for $ \brak{\gamma}$ with 
$\gamma \in \typetwo$.

\smallskip
4. Lastly,  the inductive step for $ \brak{\gamma}$ with 
$\gamma \in \typeone$ is analogous, with the respective changes: 

\begin{itemize}
\item Using the fixpoint characterisation in Proposition \ref{prop:typeone} we prove the claim for $\mu z. \indf{\gamma}{z}$, instead. 

\item 
To prove that claim we use the characterisation of least fixed points given by Knaster-Tarski theorem, according to which 
$\mu z. \indf{\gamma}{z} = \bigcup_{\alpha \in  \mathsf{Ord}}  \indf{\gamma}{Z^{\alpha}} $, where the family $\{Z^{\alpha}\}_{\alpha \in \mathsf{Ord}}$ is defined by transfinite induction, as usual:   
$Z^{0} = \emptyset$; \ $Z^{\alpha+1} =  \indf{\gamma}{Z^{\alpha}}$; \ 
$Z^{\lambda} = \bigcup_{\alpha < \lambda} \indf{\gamma}{Z^{\alpha}} $ for limit ordinals $\lambda$. 

\item Again, it suffices to prove the bisimulation invariance (resp. bisimulation closure) of each approximant formula $Z^{\alpha}$, which we do by a nested transfinite induction on $\alpha$, with the only non-trivial case being that of successor ordinals, for which we use the definition of $\indf{\gamma}{\phi}$ for $\gamma \in \typeone$ and the inductive hypotheses for $Z^{\alpha}$, and the case of nexttime extensions of goal assignments. The differences from the previous case of $\gamma \in \typetwo$, coming from the additional subformulae in $\mathsf{Finish}(\gamma)$ and $\ugam$, are inessential.  

\smallskip
This completes the structural induction and the whole proof. 
\end{itemize}

\end{proof}

\subsection{Full proof of Proposition \ref{p:remove-defects}}

\begin{proof}
By Proposition \ref{p:push}, we may assume w.l.o.g. that the defect $(u,\varphi)$ is such that $u$ is a leaf: if we can show how to remove the defect $\varphi$ at a single leaf, then, clearly, we can repeat the procedure to remove $\varphi$ at each leaf in the set $\{v_1,...,v_k\}$. {(Note that our procedure for removing a defect at a single leaf $v$ given below will not affect any other leaves, i.e. each leaf in the original network besides $v$ will still be a leaf in the new network.)}
Combined with  Proposition \ref{p:push} this proves the result. 

So, suppose that $(u,\varphi)$ is a defect and $u$ is a leaf. It is sufficient to show that there is a finite, one-step coherent network $\network''$ in which the root has the same label as $u$ in $\network$, and in which the eventuality $\varphi$ is partially fulfilled. We can then simply identify the root of the network $\network'$ with the leaf $u$ in $\network$ to form a finite, one-step coherent network $\network'$ such that $\network'' \sqsubseteq \network'$ and $\network \sqsubseteq	 \network'$. By Proposition \ref{p:stayfinished}, the eventuality $\varphi$ is partially fulfilled at $u$ in $\network'$.

Consider the $\Phi$-atoms $\Psi$ such that $\varphi \in \Psi$ and there exists a finite, one-step coherent network in which the root is labelled by $\Psi$ and the eventuality $\varphi$ is {partially}  fulfilled. 
Let $\delta$ be the disjunction of all conjunctions of the form $\bigwedge \Psi$ for all   
such $\Phi$-atoms $\Psi$. 
(This is well-defined since the set of all such conjunctions is finite, as long as we disallow conjunctions with multiple occurrences of the same conjunct.) The result then follows from the following claim:

\begin{innercustomclaim}
$\vdash \varphi \to \delta$.
\end{innercustomclaim}

To prove the claim, we consider the two cases for the eventuality $\varphi$.

\paragraph{Case: $\varphi$ is a  $\typeone$ eventuality.} Then $\varphi$ is of the form $\brak{\gamma}$, where $\gamma$ is a goal assignment supported by a set of coalitions $\fac$ for which $\gamma(C)$ is an $\until$-formula for at least one $C \in \fac$. Say that  $\fac = \{C_1,...,C_n,D_1,...,D_m\}$ and $\gamma$ is defined by: 
\[
\gamma(C_1) = \alpha_1 \until \beta_1, ...,\ \gamma(C_n) = \alpha_n \until \beta_n
\]
and  (if $m > 0$) 
\[
\gamma(D_1) = \always \chi_1, ...,\ \gamma(D_m) = \always \chi_m. 
\]
Our aim is to prove:
$$\vdash \indf{\gamma}{\delta} \to \delta$$ 
and thereafter apply the induction rule to conclude the Claim. It suffices to show that, if $\Psi$ is any atom that is consistent with $\indf{\gamma}{\delta}$, then $\bigwedge \Psi$ is, in fact, one of the disjuncts of $ \delta$.  
Indeed, suppose that $\nvdash \indf{\gamma}{\delta} \to \delta$; by Lindenbaum's Lemma there exists a maximal consistent set of formulae $\Gamma$ containing $\indf{\gamma}{\delta}$ but also the negation of $\delta$. Then $\Gamma \cap \Phi$ is a $\Phi$-atom that is consistent with $\indf{\gamma}{\delta}$, but $\bigwedge (\Gamma \cap \Phi)$ cannot be a disjunct of $\delta$ since that would make $\Gamma$ inconsistent. 

So, suppose $\Psi$ is a $\Phi$-atom for which  the set $\Psi \cup \{\indf{\gamma}{\delta}\}$ is consistent. Then $\Psi$ is consistent with at least one of the disjuncts of $\indf{\gamma}{\delta}$. 

If $\Psi$ is consistent with one of the disjuncts $ \beta_i \wedge \brak{\gamma \setminus C_i})$ then, in fact, this disjunct must be a member of 
$\Psi$, 
since $\Psi$ is a $\Phi$-atom and each formula $ \beta_i \wedge \brak{\gamma \setminus C_i})$ belongs to $\Phi$. In this case it is trivial to construct a (singleton) network in which $\brak{\gamma}$ is partially fulfilled  in $0$ steps.  

If $\Psi$ is consistent with 
{$\alpha_1 \wedge... \wedge \alpha_n\wedge \chi_1 \wedge ... \wedge \chi_m \wedge \brak{\gammaof{\delta}}$, then $\alpha_1,...,\alpha_n, \chi_1,...,\chi_m \in \Psi$ and the set $\Psi^+ := \Psi \cup \{\brak{\gammaof{\delta}}\}$ is consistent.} 
Let $\Phi^+$ be the extended Fischer-Ladner closure of the set $\Phi \cup \{\delta\}$, and 
{let $\Theta$ be the maximal modal one-step theory that is contained in $\Psi^+$}. Then, since $\Theta \subseteq \Psi^+$, it is consistent by our assumption, and it is a one-step theory over $\Phi^+$. By one-step completeness (Theorem \ref{t:one-step-comp}),  there exists a finite maximal consistent game form 
$\gmod(\Theta) = (\Act,\act,\psf (\Phi^+),\out)$ 
for $\Phi^+$ such that, for every goal assignment $\gamma'$:
\begin{enumerate}
\item If $\brak{\gamma'} \in \Theta$, then there is an action profile $\actprof \in \Pi_{a \in \mathsf{Agt}}\act_a$ such that for all $C$ in the support of $\gamma'$, and all $\actprof' \sim_C \actprof$, we have 
 $\psi \in \out(\actprof')$, where $\gamma'(C) = \nexttime \psi$. 

\item If $\neg \brak{\gamma'} \in \Theta$, then for every profile $\actprof \in \Pi_{a \in \mathsf{Agt}}\act_a$ there is some $C$ in the support of $\gamma'$, and some $\actprof' \sim_C \actprof$, for which 
 $\overline{\psi} \in \out(\actprof')$, where $\gamma'(C) = \nexttime \psi$.  
\end{enumerate}
Since $\brak{\gammaof{\delta}} \in \Theta$, the first clause ensures that there is a profile $\rho \in \Pi_{a \in \mathsf{Agt}}\act_a$ such that for all $C$ in the support of $\diffof{\gamma}[\bigcup \fac  \gass  \X \delta]$, and all $\actprof' \sim_C \rho$, we have 
 $\psi \in \out(\actprof')$, where $\gammaof{\delta}(C) = \nexttime \psi$.

In particular, for $C = \bigcup \fac$ this entails that $\delta \in \out(\actprof')$ for all $\actprof' \sim_{\bigcup \fac} \rho$.

We now construct a network showing that $\bigwedge \Psi$ is a disjunct of $\delta$ as follows. For each action profile $\actprof \in \Pi_{a \in \mathsf{Agt}}\act_a$, we pick a $\Phi$-network $\network_\actprof$ according to the following rule. 

\begin{itemize}
\item 
if $\actprof \sim_{\bigcup \fac} \rho$, then let $\network_\actprof$ be a network of which the root is labelled $\out(\actprof) \cap \Phi$, and in which the eventuality $\brak{\gamma}$ is partially fulfilled  at the root. 
Such a network exists since, if $\actprof \sim_{\bigcup \fac} \rho$ then $\delta \in \out(\actprof)$, hence $\bigwedge(\out(\actprof) \cap \Phi)$ is a disjunct of $\delta$.  
\item 
Otherwise, let $\network_\actprof$ be a network consisting of a single node labelled $\out(\actprof) \cap \Phi$. 
\end{itemize}
We form the network $\network = (T,L,\gform)$ by taking the disjoint union of the networks $\network_\actprof$,  for each action profile $\actprof \in \Pi_{a \in \mathsf{Agt}}\act_a$, together with a new root $r$ labelled $\Psi$, with an edge to the root of {each} $\network_\actprof$,  
and letting $\gform(r) := (\Act,\act,T,\out')$ where $\out'$ maps each action profile $\actprof$ to the root of $\network_\actprof$. Note that $\network_\actprof \sqsubseteq \network$ for each $\actprof$, and that for each profile $\actprof$ we have $L(\out'(\actprof)) = \out(\actprof) \cap \Phi$.

Since $\Pi_{a \in \mathsf{Agt}}\act_a$ is a finite set, each $\network_\actprof$ is one-step coherent, and due to the clauses (1) and (2) of Theorem \ref{t:one-step-comp}, the network $\network$ is a finite and one-step coherent network. Moreover, to see that the eventuality $\brak{\gamma}$ is partially fulfilled  at the root $r$, we define a marking $\mrk$ by setting 
\[
\mrk(v) := \big\{\theta \in \Phi \mid  \mbox{there exists } C \subseteq \bigcup \fac \mbox{ such that }   \diffof{\gamma}(C)
=  \nexttime \theta \; \mbox{and} \; \actprof \sim_C \rho  \big\} 
\]
if $v$ is the root of the network $\network_\actprof$ for some $\actprof \in \Pi_{a \in \mathsf{Agt}}\act_a$, and $\mrk(v) := \emptyset$ otherwise. 
By definition, this marking verifies $\diffof{\gamma}$, witnessed by the action profile $\rho$. Furthermore,  if $\brak{\gamma} \in \mrk(v)$ then, since $\diffof{\gamma}(C) \neq \nexttime \brak{\gamma}$ for all $C \neq \bigcup F$, $v$ must be the root of some network $\network_\actprof$ where $\actprof \sim_{\bigcup F} \rho$. This means that $\brak{\gamma}$ is partially fulfilled  at $v$ in $\network_\actprof$, hence in $\network$, since $\network_\actprof \sqsubseteq \network$. It follows that $\brak{\gamma}$ is partially fulfilled  in at most $k + 1$ steps at $r$ in $\network$, where $k$ is the maximum number such that $\brak{\gamma}$ is partially fulfilled  in at most $k$ steps at the root of one of the finitely many networks $\network_\actprof$.

\paragraph{Case: $\varphi$ is a $\typetwo$ eventuality.}

Then $\varphi$ is of the form $\neg \brak{\gamma}$, where $\gamma$ is a goal assignment supported by a set of coalitions $\fac$ and  $\gamma(C)$ is an $\always$-formula for all $C \in \fac$. Say that  $\fac = \{D_1,...,D_m\}$ and $\gamma$ is defined by 
$\gamma(D_1) = \always \chi_1, ...,\ \gamma(D_m) = \always \chi_m$. 
Our aim is to prove:
$$\vdash \neg \delta \to \indf{\gamma}{\neg\delta}$$ 
and thereafter apply the co-induction rule to conclude that $\vdash \neg \delta \to \brak{\gamma}$, hence $\vdash \neg \brak{\gamma} \to \delta$, as required. It suffices to show that, if $\Psi$ is any atom that is consistent with $\neg \indf{\gamma}{\neg\delta}$, then $\bigwedge \Psi$ is, in fact, one of the disjuncts of $\delta$. Indeed, suppose that $\nvdash \neg \delta \to\indf{\gamma}{\neg \delta}$. Then $\nvdash \neg \indf{\gamma}{\neg \delta} \to \delta$ so  by Lindenbaum's Lemma there exists a maximal consistent set of formulae $\Gamma$ containing $\neg\indf{\gamma}{\neg \delta}$ but also the negation of $\delta$. Then $\Gamma \cap \Phi$ is a $\Phi$-atom that is consistent with $\neg \indf{\gamma}{\neg \delta}$, but $\bigwedge (\Gamma \cap \Phi)$ cannot be a disjunct of $\delta$ since that would make $\Gamma$ inconsistent. 

We recall that $\indf{\gamma}{\neg\delta}$ is the formula $\chi_1 \wedge ... \wedge \chi_m \wedge\brak{\diffof{\gamma}[\bigcup \fac  \gass  \X \neg \delta]}$, so $\neg \indf{\gamma}{\neg\delta}$ is provably equivalent to $\overline{\chi}_1 \vee ... \vee  \neg \overline{\chi}_m  \vee \neg \brak{\diffof{\gamma}[\bigcup \fac  \gass  \X \neg \delta]}$.
 So, suppose $\Psi$ is a $\Phi$-atom for which  the set $\Psi \cup \{\neg \indf{\gamma}{\neg \delta}\}$ is consistent. Then $\Psi$ is consistent with at least one of the disjuncts of $\neg \indf{\gamma}{\neg \delta}$. If $\Psi$ is consistent with one of the disjuncts $\overline{\chi}_i$ then, in fact, this disjunct must be a member of $\Psi$, since $\Psi$ was a $\Phi$-atom and each formula $\overline{\chi}_i$ belongs to $\Phi$. In this case it is trivial to construct a (singleton) network in which $\neg \brak{\gamma}$ is partially fulfilled  in $0$ steps.  If $\Psi$ is consistent with $ \neg \brak{\diffof{\gamma}[\bigcup \fac  \gass  \X \neg \delta]}$, then the set $\Psi^+ := \Psi \cup \{\neg \brak{\diffof{\gamma}[\bigcup \fac  \gass  \X \neg \delta]}\}$ is consistent. Let $\Phi^+$ be the extended Fischer-Ladner closure of the set $\Phi \cup \{\neg \delta\}$, and let $\Theta$ be the set of all one-step formulae over $\Phi^+$ belonging to $\Psi^+$. Then, since $\Theta \subseteq \Psi^+$, it is consistent by our assumption, and it is a one-step theory over $\Phi^+$. 
 By one-step completeness (Theorem \ref{t:one-step-comp}),  there exists a finite maximal consistent game form $\gmod(\Theta) = (\Act,\act,\psf (\Phi^+),\out)$ 
 such that, for every goal assignment $\gamma'$:
\begin{enumerate}
\item If $\brak{\gamma'} \in \Theta$, then there is a profile $\actprof \in \Pi_{a \in \mathsf{Agt}}\act_a$ such that for all $C$ in the support of $\gamma'$, and all $\actprof' \sim_C \actprof$, we have 
 $\psi \in \out(\actprof')$, where $\gamma'(C) = \nexttime \psi$. 
\item If $\neg \brak{\gamma'} \in \Theta$, then for every profile $\actprof \in \Pi_{a \in \mathsf{Agt}}\act_a$ there is some $C$ in the support of $\gamma'$, and some $\actprof' \sim_C \actprof$, for which  we have 
 $\overline{\psi} \in \out(\actprof')$, where $\gamma'(C) = \nexttime \psi$.  
\end{enumerate}
Since $\neg \brak{\diffof{\gamma}[\bigcup \fac  \gass  \X \neg \delta]}\in \Theta$, the second clause ensures that for every profile $\rho \in \Pi_{a \in \mathsf{Agt}}\act_a$ there is some $C$ in the support of $\diffof{\gamma}[\bigcup \fac  \gass  \X \neg \delta]$, and some $\actprof' \sim_C \rho$, for which we have 
 $\overline{\psi} \in \out(\actprof')$, where $\gamma'(C) = \nexttime \psi$.   
We pick a choice function selecting such a pair $(c(\rho),f(\rho))$ for each $\rho$, where $c(\rho) \in \fac$, $\rho \sim_{c(\rho)} f(\rho)$ and 
 $\overline{\psi} \in \out(f(\rho))$, where $\gamma (c(\rho)) = \nexttime \psi$.  

We now construct a network showing that $\bigwedge \Psi$ is a disjunct of $\delta$ as follows. For each action profile $\actprof \in \Pi_{a \in \mathsf{Agt}}\act_a$, we pick a $\Phi$-network $\network_\actprof$ according to the following rule.
\begin{itemize}
\item 
 if $\actprof = f(\rho)$ for some $\rho$ such that $c(\rho) = \bigcup \fac$, then let $\network_\actprof$ be a network of which the root is labelled $\out(\actprof) \cap \Phi$, and in which the eventuality $\neg \brak{\gamma}$ is partially fulfilled  at the root. Such a network exists since, if  $\actprof = f(\rho)$ for some $\rho$ 
 then  
{ $\gamma(\bigcup \fac) = \X\neg\delta$ and so $\overline{\neg \delta} = \delta \in \out(\actprof)$, hence $\bigwedge(\out(\actprof) \cap \Phi)$ is a disjunct of $\delta$.}  
 
 \item 
 Otherwise, let $\network_\actprof$ be a network consisting of a single node labelled $\out(\actprof) \cap \Phi$. 
 \end{itemize}
 We form the network $\network = (T,L,\gform)$ by taking the disjoint union of the networks $\network_\actprof$, for each action profile $\actprof \in \Pi_{a \in \mathsf{Agt}}\act_a$, together with a new root $r$ labelled $\Psi$ with an edge to the root of each $\network_\actprof$, and letting $\gform(r) = (\Act,\act,T,\out')$ where $\out'$ maps each action profile $\actprof$ to the root of $\network_\actprof$. Note that $\network_\actprof \sqsubseteq \network$ for each $\actprof$, and that for each profile $\actprof$ we have $L(\out'(\actprof)) = \out(\actprof) \cap \Phi$.
Since $\Pi_{a \in \mathsf{Agt}}\act_a$ is a finite set, each $\network_\actprof$ is one-step coherent, and due to the clauses (1) and (2) of Theorem \ref{t:one-step-comp}, the network $\network$ is a finite and one-step coherent network. Moreover, to see that the eventuality $\neg \brak{\gamma}$ is partially fulfilled  at the root $r$, we define a marking $m$ by setting 
\[
\mrk(v) := \big\{\overline{\theta} \in \Phi \mid  \mbox{there exists } \rho \mbox{ such that }   \diffof{\gamma}(c(\rho)) =  \nexttime \theta \; \mbox{and} \; f(\rho) = \actprof \big\}
\]
if $v$ is the root of the network $\network_\actprof$ for some $\actprof \in \Pi_{a \in \mathsf{Agt}}\act_a$, and $\mrk(v) = \emptyset$ otherwise. 
By definition, this marking refutes $\diffof{\gamma}$. Furthermore,  if $\neg\brak{\gamma} \in \mrk(v)$ then, since $\diffof{\gamma}(C) \neq \nexttime \brak{\gamma}$ for all $C \neq \bigcup \fac$, there must be some $\rho$ such that $c(\rho) = \bigcup \fac$, and $v$ is the root of $\network_{f(\rho)}$. This means that $\neg\brak{\gamma}$ is  partially fulfilled  at $v$ in $\network_{f(\rho)}$, hence in $\network$, since $\network_{f(\rho)}\sqsubseteq \network$. It follows that $\neg\brak{\gamma}$ is partially fulfilled  in at most $k + 1$ steps at $r$ in $\network$, where $k$ is the maximum number such that $\neg \brak{\gamma}$ is partially fulfilled  in at most $k$ steps at the root of one of the finitely many networks $\network_\actprof$.
\end{proof}

\subsection{Proof of Proposition \ref{prop:network-extension}}
\begin{proof}
Repeated use of Proposition \ref{p:remove-leaves} and Proposition \ref{p:remove-defects}, as follows. 
First, we fix an enumeration $\mathcal{L} = l_1, l_2,..., l_k$ of all leaves in $\network$ and an enumeration $\mathcal{D} = (u_1,\phi_1),..., (u_d,\phi_d)$ of all defects in $\network$. 
Then, we construct a finite chain of finite one-step coherent networks $\network_0 \sqsubseteq \network_1 \sqsubseteq...$, inductively as follows. 

We begin with $\network_0 = \network$. Suppose, we have constructed the finite one-step coherent $\Phi$-networks  
$\network_0 \sqsubseteq \network_1 \sqsubseteq... \network_n$ and let  
$\mathcal{L}_n$ be an enumeration of all leaves of $\network$ that are still in $\network_n$ and $\mathcal{D}_n$ be an enumeration of all defects in $\network$ that are still in $\network_n$. Then we do the following. 

\smallskip
1. If $\mathcal{D}_n$ is non-empty, we pick the first defect $(u,\phi)$ listed in it and apply Proposition \ref{p:remove-defects} to construct a finite, one-step coherent network $\network_{n+1}$ such that  
$\network_n \sqsubseteq \network_{n+1}$ and 
$(u_1,\phi_1)$ is not a defect in $\network_{n+1}$. 
Then we update the list $\mathcal{D}_{n}$ to $\mathcal{D}_{n+1}$ by removing $(u,\phi)$ and all other defects in that list that may have been resolved in $\network_{n+1}$. Note that, by Proposition \ref{p:stayfinished}, newly occurring defects in $\network_{n+1}$ may only occur at newly added nodes. We update likewise the list $\mathcal{L}_{n}$ to $\mathcal{L}_{n+1}$ by removing the leaves in $\network_{n}$ that are no longer leaves in $\network_{n+1}$. 

\smallskip
2. If $\mathcal{D}_n$ is empty, but $\mathcal{L}_{n}$ is non-empty, 
then we pick the first leaf $l$ currently listed in $\mathcal{L}_{n}$ and apply 
Proposition \ref{p:remove-leaves} to construct a  finite, one-step coherent network 
$\network_{n+1}$ such that  $\network_n \sqsubseteq \network_{n+1}$ and 
$l$ is not a leaf in $\network_{n+1}$. Then we update the list $\mathcal{L}_{n}$ to $\mathcal{L}_{n+1}$ by removing $l$ and all other leaves in $\network_{n}$ that are no longer leaves in $\network_{n+1}$. Note that, again by Proposition \ref{p:stayfinished}, newly occurring defects in $\network_{n+1}$ may only occur at newly added nodes, so the list of original defects in $\network$ remains empty. 

3. If both $\mathcal{D}_n$ and $\mathcal{L}_{n}$ are empty, the last constructed network $\network_n$ is the desired $\network'$. 

\smallskip
Note that both lists $\mathcal{L}$ and $\mathcal{D}$ are finite, neither of them gets extended in any step, and at least one of them strictly decreases in each step. 
Therefore, they will both become empty in finitely many steps, hence the construction  is guaranteed to terminate. 
Clearly, once the procedure terminates, we obtain a network none of whose leaves are leaves of $\network$ and such that none of its defects are defects in $\network$. 
\end{proof}

\subsection{Proof of Proposition \ref{p:c-sat}}

\begin{proof}
For the direction left to right, suppose  $\gform = 
(\Act,\act,\psf(V),\out)$ is some game form that $\mathcal{S}$-satisfies $\Gamma$. By assumption, for each action profile $\actprof$ there is some $Z \in \mathcal{S}$ with $\out(\actprof) \subseteq Z$.  We show that $\mathcal{S}$ satisfies both conditions (1) and (2).

For (1), let $R = (C_1,\gamma_1,...,C_n,\gamma_n)$ be a redistribution of $\Gamma$. Since $\gamma_1,...,\gamma_n$ are positive relevant goal assignments, there must be strategy profiles $\actprof_1,...,\actprof_n$ such that for each $\actprof_i$, coalition $C$ and $\actprof' \sim_C \actprof_i$, we have $p \in \out(\actprof')$ where $\gamma_i(C) = \X p $. Define the new strategy profile $\actprof^*$ by setting $\actprof^*_\aga = (\actprof_i)_\aga$ if $\aga \in C_i$, and $\actprof^*_\aga$ arbitrarily chosen if $\aga \notin C_1 \cup ... \cup C_n$. It is easy to see that $F(R) \subseteq \out(\actprof^*)$, hence $F(R) \subseteq Z$ for some $Z \in \mathcal{S}$. 

For (2), let $R = (C_1,\gamma_1,...,C_n,\gamma_n)$ be a redistribution of $\Gamma$ and let $\gamma'$ be a negative relevant goal assignment. Let $\actprof_1,...,\actprof_n$ and $\actprof^*$ be defined as before, so that $F(R) \subseteq \out(\actprof^*)$. Since $\gform$ $\mathcal{S}$-satisfies $\Gamma$, there must be some coalition $C$ and some $\actprof' \sim_C \actprof^*$ such that $q \in \out(\actprof')$ where $\gamma'(C) = \X q$. If $C = \Agt$ then $q$ must belong to the outcome of \emph{every} action profile, hence to every member of $\mathcal{S}$. If $C \neq \Agt$, then we get $F(R,\gamma',C) \subseteq \out(\actprof') \subseteq Z$ for some $Z \in \mathcal{S}$.

For right to left, suppose $\mathcal{S}$ satisfies the constraints stated in the proposition. Enumerate $\Agt$ as $\aga_0,...,\aga_{\vert \Agt \vert -1}$. 
Let $\gform = 
(\Act,\act,\psf(V),\out)$ be defined  as follows.  First, given an agent $\aga \in \Agt$, we define the set of actions $\act_\aga$ to be the set of triples $(\gamma,f,k)$ where:
\begin{itemize}
 \item $\gamma$  is either a positive relevant goal assignment, or a distinguished symbol $*$,
\item $k \in \{0,...,\vert \Agt \vert-1)$
\item  $f$ is a map sending each redistribution $R$ of $\Gamma$  to some element  $Z$ of $\mathcal{S}$ with $F(R) \subseteq Z$. 
\end{itemize}
Note that each player has the same set of actions, and the set of actions is non-empty since, by assumption, for each redistribution $R$ of $\Gamma$  there is some element  $Z$ of $\mathcal{S}$ with $F(R) \subseteq Z$.

Given an action profile $\actprof$ we say that \defstyle{a coalition $C$ has formed behind a positive 
relevant goal assignment} 
if that goal assignment was chosen by every member of $C$ in $\actprof$.  A given action profile $\actprof$  thus partitions $\Agt$ 
into disjoint coalitions $\{C_1,...,C_n\}$ where each coalition $C_i$ is a maximal one that has formed behind some positive relevant goal assignment (or the special symbol $*$). Given an action profile $\actprof$ with corresponding partition of the form $\{C_1,...,C_n\}$, where each coalition $C_i$ has formed behind the goal assignment $\gamma_i$, or possibly of the form $\{C_1,...,C_n,D\}$, where each coalition $C_i$ has formed behind the goal assignment $\gamma_i$ and each player in $D$ has played an action with $*$ as first component, we define the redistribution $\mathsf{red}(\actprof)$ induced by $\actprof$ to be $R = (C_1,\gamma_1,...,C_n,\gamma_n)$. Suppose each $\aga\in \Agt$ has chosen the action $(\gamma_\aga,Z_\aga,k_\aga)$. Then let 
$$\mathsf{bet}(\actprof) = \sum_{\aga \in \Agt} k_\aga\text{ mod } \vert \Agt \vert$$ 
We now define the outcome function by setting $\out(\actprof) = f(\mathsf{red}(\actprof))$, where $f$ is the function chosen by player $\aga_i$ for $i = \mathsf{bet}(\actprof)$.
Informally speaking, each agent picks a goal assignment that they want to act towards, and also votes on a function intended to pick the actual outcome of an action  profile depending on which coalitions have formed in favor of which goal assignments. To decide which voter gets to determine the outcome of an action profile, each player also has to bet on a number $k$, and the winner of this bet is the agent indexed by   $\sum_{\aga \in \Agt} k_\aga\text{ mod } \vert \Agt \vert$.

We now show that $\gform$ $\mathcal{S}$-satisfies $\Gamma$. 

It clearly holds that $\out(\actprof) \in \mathcal{S}$ for every action profile $\actprof$ and, conversely, for each $Z \in \mathcal{S}$, $Z$ is the outcome of the profile in which each player chooses the action $(*,Z,0)$. Now suppose $\brak{\gamma} \in \Gamma$ for some positive goal assignment $\gamma$. We need to find an action profile $\actprof$ such that for every coalition $C$, and every $\actprof' \sim_C \actprof$, we have $p \in \out(\actprof')$ where $\gamma(C) = \X p$. But this is easy: the action profile $\actprof$ defined by $\actprof_\aga = (\gamma,Z,k)$ for all $\aga \in \Agt$, with $Z,k$ arbitrarily chosen, clearly does the job.

On the other hand, suppose $\neg \brak{\gamma'} \in \Gamma$ for some negative goal assignment $\gamma'$. Let $\actprof$ be an arbitrary action profile in $\gform$.  We want to find some coalition $B$ and some $\actprof' \sim_B \actprof$ such that $q \in \out^\Gamma(\actprof')$, where $\gamma'(B) = \X \neg q$. Let $R = (C_1,\gamma_1,...,C_n,\gamma_n)$ be the  redistribution of $\Gamma$ defined to consist of those $C_i,\gamma_i$ such that $C_i$ has formed in $\actprof$ behind the goal assignment $\gamma_i$, and each $C_i \subseteq C'$.  By the second condition on $\mathcal{S}$, there is some coalition $C'$ such that:

\begin{itemize}
\item[(i)] either $C' = \Agt$ and $q \in Z$ for every $Z \in \mathcal{S}$, where $\gamma'(\Agt) = \X\neg q$, or
\item[(ii)] $C' \neq \Agt$ and there is some $Z \in \mathcal{S}$ such that $F(R,\gamma',C') \subseteq Z$. 
\end{itemize}

In the first case, we have $\actprof \sim_{\Agt} \actprof$ and $q \in \out(\actprof) \in \mathcal{S}$. In the second case, $C' \neq \Agt$ and there is some $Z \in \mathcal{S}$ such that $F(R,\gamma',C') \subseteq Z$. Let $R' = \{C_1',\gamma'_1,...,C_m',\gamma_m'\}$ be the redistribution defined by $\{C_1',...,C_m'\} = \{C_i  \cap C' \mid i \in \{1,...,n\} \; \& \; C_i \subseteq C'\}$, and let $\gamma'_j$ be $\gamma_i$  for $C_j' = C_i \cap C'$. This is well defined since $\{C_1,...,C_n\}$ consists of pairwise disjoint sets. Furthermore, note that $F(R') = Z$, since:
\begin{align*}
F(R') & = \{p \mid \exists i, B \subseteq C_i':\;  \gamma_i'(C_i')  = \X p\} \\
& \subseteq \{p \mid \exists i,B \subseteq C_i \cap C':\;\gamma_i(B) = \X p\} 
\\ 
& \subseteq \{p \mid \exists i,B \subseteq C_i \cap C':\;\gamma_i(B) = \X p\} \cup \{q\} \\
& = F(R,\gamma',C') \\
& \subseteq Z
\end{align*}	
Define a new action profile $\actprof'$ as follows: pick some function $f$ with $f(R') = Z$, and pick some distinguished player $\agc \in \Agt \setminus C'$. For each $\aga \in C'$, set $\actprof'_\aga = \actprof_\aga$. For $\aga \in \Agt\setminus C'$, set $\actprof'_\aga = (*,f,0)$ if $\aga \neq \agc$. Finally set $\actprof'_{\agc} = (*,f,k)$, where $k$ is chosen to ensure that $\mathsf{bet}(\actprof')$ is the index of player $\agc$. Clearly, $\actprof' \sim_{C'} \actprof$, and it remains only to show $q \in \out(\actprof')$. But $q \in Z$, and since $R'= \mathsf{red}(\actprof')$ we get $\out(\actprof') = f(R') = Z$. Informally, what happens here is that the players in  $\Agt \setminus C'$  pick one member among them, $\agc$, that will do the actual work towards ensuring that the outcome of $\actprof'$ contains $q$. So that player picks $Z$ containing $q$ as the outcome, and the other players in $\Agt\setminus C'$ lay down their bets by betting on $0$ so that $\agc$ can make sure to win the bet, thus ensuring that the outcome is $Z$.
 \end{proof}

\end{document}